\documentclass[aps,pra,superscriptaddress,twocolumn,footinbib,floatfix]{revtex4-2}

\usepackage{amsmath,amssymb,amsfonts,mathtools,bm}

\usepackage{physics}

\usepackage{graphicx}
\usepackage[caption=false]{subfig} 
\usepackage{float}
\usepackage{multirow}
\usepackage{placeins} 

\usepackage[hyperfootnotes=false]{hyperref}
\usepackage[capitalise]{cleveref}
\usepackage{url}
\usepackage{xcolor}
\definecolor{darkblue}{rgb}{0,0,0.5}
\hypersetup{
  colorlinks=true,
  linkcolor=black,
  citecolor=darkblue,
  urlcolor=darkblue
}

\newtheorem{proposition}{Proposition}
\newtheorem{definition}{Definition}
\newenvironment{proof}[1][Proof]{\noindent\textbf{#1.} }{\ \rule{0.5em}{0.5em}}

\newcommand{\defeq}{\vcentcolon=}

\DeclareMathOperator*{\argmax}{arg\,max}
\DeclareMathOperator*{\argmin}{arg\,min}

\definecolor{psred}{rgb}{0.8,0.2,0.2}
\definecolor{revblue}{rgb}{0.0,0.2,0.75}

\newcommand{\SP}[1]{\textcolor{black}{#1}}

\makeatletter
\def\algbackskip{\hskip-\ALG@thistlm}
\makeatother

\begin{document}
\title{Free-Space Quantum Networks and Optimized Fiber-Reinforcement}

\author{Alasdair I.~Fletcher}
\affiliation{Department of Computer Science, University of York, York YO10 5GH, United Kingdom}
\author{Ignazio Pedone}
\affiliation{nodeQ, 71-75 Shelton
Street, Covent Garden, London WC2H 9JQ, United Kingdom}
\author{Stefano Pirandola}
\thanks{Correspondence addressed to \href{mailto:stefano.pirandola@york.ac.uk}{stefano.pirandola@york.ac.uk}}
\affiliation{Department of Computer Science, University of York, York YO10 5GH, United Kingdom}
\begin{abstract}
Free-space quantum communication provides a flexible complement to fiber-based quantum networks, but its point-to-point capacity is fundamentally limited by diffraction, atmospheric extinction and beam wandering induced by turbulence. In this work, we study the end-to-end performance of large-scale free-space quantum networks connecting randomly distributed fixed or mobile users, modelled as Waxman random graphs. We derive the mean network capacity, edge consumption and connectivity phase transitions for both single-path and multi-path (flooding) routing. We also study router-centered star networks, deriving the full distribution of end-to-end capacities as a function of the router's coverage radius. We then consider how performance may be improved by reinforcing free-space networks with a small number of optimally placed fiber-based backbone nodes. We prove that any optimal backbone configuration must correspond to a capacity-maximizing Voronoi tessellation of the network region, and show that this can be efficiently approximated by a centroidal Voronoi tessellation via Lloyd's algorithm, with backbone nodes connected according to a Delaunay triangulation. Numerical results show that even a modest number of backbone nodes substantially improves end-to-end capacity and reduces edge consumption for both mobile and fixed users. 
\end{abstract}
\maketitle

\section{Introduction\label{sec:introduction}}
The current classical internet connects a multitude of users on a vast global scale for communication, distributed computation and many other processes and services. It operates over a wide range of domains, from small scale local area networks to global satellite based communications \cite{Handley2019,Preet_2022}. It therefore necessarily contains a variety of connection types including wired connections such as fiber-optic cable and wireless connections, typically operating through free-space. Wireless communications have revolutionised modern networking \cite{Agiwal2016-5Greview,Raychaudhuri2012FrontiersWireless}, providing improved flexibility and cost-effectiveness over wired connections, whilst still maintaining sufficiently high communication rates. In particular mobile network technology is dependent on efficient wireless communication in order to operate \cite{Raychaudhuri2012FrontiersWireless}.

The future quantum internet \cite{Kimble2008QInternet,UniteQInt} will aim to emulate the current classical internet, enabling multiple users to connect over long distances for quantum-secure communication \cite{Joshi2020,CambridgeQN,Chen2021_46nodeQN,QKD250km}, quantum information processing and distributed quantum computation \cite{VanMeter2016PathToScalableQC,QINetworkChallenges2020,Cuomo2020TowardsDistQC,Cirac1999DistQcomp}. Nonetheless, the inherent fragility of quantum systems to decoherence and the inability to clone quantum information \cite{Wooters1982} render high rate distribution of quantum information far more challenging than its classical counterpart. The rate of point-to-point quantum information transmission is known to be fundamentally distance limited as described by the Pirandola-Laurenza-Ottaviani-Banchi (PLOB) bound \cite{PLOB} which decays exponentially with distance.

Given these difficulties, much attention has focused on transmission of quantum states in networks through optical fiber, both theoretically \cite{QuntaoRandQNets,ZhangQInt} and experimentally~\cite{CambridgeQN,Chen2021_46nodeQN}. Nonetheless, it is clear future large-scale quantum networks will need to operate in a hybrid manner, combining fiber and free-space links, and will require analytical methods for approximating the performance of such large-scale structures \cite{OPGQN,Harney2022HybridQNs,IRQN}.

Recent work has tightly upper-bounded the capacities of free-space quantum communication channels \cite{FS,MasoudFS} by making use of techniques from turbulence theory \cite{AndrewsTurb}, optics \cite{SiegmanLasers}, and quantum information theory \cite{Holevo19,Simeone_2026}. However, when applied to ground-based free-space links, the capacity of such channels is markedly less than the equivalent fiber-based channel, so it may not be possible to simply use free-space wireless links wherever desired.

The topology of a future quantum internet is likely to be emergent and highly complex but many expected topological features may be captured by random network models. This was first explored for fiber-based networks in \cite{QuntaoRandQNets,ZhangQInt}, which identified critical densities of users to achieve end-to-end capacity benchmarks. Nonetheless, many fundamental questions remain. Given the fragility of quantum information, it remains unclear whether wireless technology is a feasible mechanism for connecting users at sufficiently high rates of communication, and what proportion of a network must be fiber to increase the performance to the required level.

Our work is inspired by these questions. Firstly, we characterize the end-to-end optimal performance of free space quantum networks in two operational regimes, i.e., considering random fixed and mobile users. The performance is mainly quantified in terms of two-way assisted capacity (or its upper bound) for quantum communication, entanglement distribution and key generation between end points, besides considering the notion of routing consumption, i.e., the fraction of links utilized by the routing solution, following Ref.~\cite{PracticalRouting}. Secondly, we show how the performance can be improved by the inclusion of a limited number of fiber-optic channels, by means of suitably-designed backbone structures.

We structure our paper as follows. We begin by providing background notions. We first review the capacities of free-space quantum channels in \cref{subsec:caps} and consider the consequences for quantum networking through free-space. We discuss the general theoretical apparatus for analyzing quantum networks in a setup agnostic way and various routing strategies for such networks in \cref{subsec:routing_basics}. We introduce the methods for modeling large scale free-space quantum networks in \cref{subsec:rand_Q_net}, in particular by considering the Waxman random network model and routers connecting randomly positioned user nodes. Baseline results for free-space quantum networks are presented in \cref{sec:baseline}. In \cref{sec:optimal_bb_design} we introduce backbone networks as a method of improving network performance and demonstrate how to optimize such structures for maximal end-to-end performance. We present the results of the numerical investigation into the performances of backbone assisted free-space networks in \cref{sec:optimized bb results} and  our conclusions in \cref{sec:conclusions}.

\section{Preliminary notions}

\subsection{Capacities of free-space quantum channels}\label{subsec:caps}
The first step in modelling random quantum networks is to consider the constituent quantum channels which make up the connections of the network. Here we are concerned with free-space networks, possibly reinforced with additional fiber-optic links. In this section, we briefly review the capacities of thermal loss channels in optical fiber and free-space. If not otherwise specified, throughout the paper, by `capacity' we mean a generic capacity assisted by two-way classical communication. Depending on what quantum communication task we consider, i.e., quantum key distribution, quantum information transfer, or entanglement distribution, the generic capacity represents a secret-key agreement capacity ($K_2$), a two-way assisted quantum capacity ($Q_2$), or a two-way assisted entanglement distribution capacity ($D_2$), respectively. Interested readers are referred to~\cite{PLOB,Pirandola2018ChannelSim} for the precise definitions of these capacities. We then use the terminology of `capacity bound' when we refer to an upper bound to the actual capacity. This is needed since the exact capacity may not be known for certain types of quantum channels, such as the thermal-loss channels. In general, when we refer to a mix of `capacities' and `capacity bounds', we use the weaker terminology of `capacity bounds' to cover both cases.

\subsubsection{Quantum fiber channels}
In general, we may model quantum communication in a noisy optical fiber as a bosonic thermal loss channel in which the transmitted mode is mixed on a beamsplitter of transmissivity $\eta$ with an incoming thermal mode with mean photon number $\bar{n}$. Note that we may express the transmissivity $\eta$ in terms of the physical distance traversed by the channel using the expression $\eta=10^{-\gamma d}$ where $d$ is the distance in kilometers and we take $\gamma=0.02~\mathrm{dB}~\mathrm{km}^{-1}$which is the loss rate in the state of the art fiber. 

The capacity of a thermal loss channel can be upper bounded by resorting to the relative entropy of entanglement (REE) of the channel's Choi matrix. This approach leads to~\cite{PLOB}
\begin{equation}
 \mathcal{C}(\mathcal{E})\leq
\begin{cases}
-\log_2\big[(1-\eta)\eta^{\bar{n}}\big]-h(\bar{n}) & \bar{n}<\frac{\eta}{1-\eta} \\
~ 0 & \mathrm{otherwise},
\end{cases}
\label{eq:thermchannelbound}
\end{equation}
where $h$ is the entropic function
\begin{equation}
h(x) \defeq (x+1)\log_2(x+1)-x\log_2(x).
\label{eq:h_func}
\end{equation}

In the case of a lossy (pure-loss) channel, in which $\bar{n}=0$, Eq.~\eqref{eq:thermchannelbound} becomes the PLOB bound~\cite{PLOB}
\begin{equation}
\mathcal{C}(\mathcal{E}) \leq -\log_2(1-\eta).
\end{equation}
In this particular case, the upper bound coincides with the lower bound established in Ref.~\cite{PirPatron09}, so the PLOB bound represents exactly the capacity of the lossy channel, i.e., one may write~\cite{PLOB} 
\begin{equation}
\mathcal{C}(\mathcal{E}) = -\log_2(1-\eta). \label{PLOBeq}
\end{equation}

\subsubsection{Quantum free-space channels}
We are interested in free-space channels in two distinct regimes. Firstly, we consider free-space channels connecting two fixed devices, and secondly, free-space channels inter-connecting mobile nodes. The relevant setup parameters for the two different regimes are shown in \cref{table:Setups}. We begin this section by describing basic notions that affect both types of channels before considering further notions which are specific to each of the two setups.

We review here the work of Ref.~\cite{FS} which established bounds for the capacities of free-space channels. For simplicity let us begin without thermal noise and consider the pure loss case first. The basic setup is as follows: we consider Alice and Bob communicating through free-space across a distance $d$ with the two parties each located at approximately the same altitude $h$. Communication is enabled by means of a Gaussian beam of curvature $R_0$, wavelength $\lambda$ and initial spot size $w_0$. The beam is prepared by Alice and directed towards Bob's detector which has a circular aperture of radius $a_R$. There are three main causes of loss that may be immediately modeled: diffraction, atmospheric extinction and experimental setup losses. We now deal with these in turn. 

\renewcommand{\arraystretch}{1.2}
\begin{table}[t!]
\centering
\begin{tabular}{ |l|c|c|c| } 
 \hline
\textit{\hspace{6.5mm}Parameter} & ~\textit{Symbol}~ & ~\textit{Mobile}~ & ~\textit{Fixed}~ \\ 
 \hline  \hline
Altitude & $h$ & 30~m & 30~m\\ \hline
Beam Curvature & $R_{0}$ & $\infty$ & $\infty$ \\  \hline
Wavelength & $\lambda$ & 800~nm & 800~nm  \\  \hline
Initial spot-size & $w_0$ &   1~mm & 5~cm \\ \hline
Receiver Aperture & $a_R$ & 1~cm &  5~cm \\ \hline
Setup Noise  & $\bar{n}_{\text{ex}}$ & 0.01 &  0.01 \\  \hline
Detector Efficiency & $\eta_{\text{eff}}$ & 0.7 & 0.5 \\ \hline
Pointing error  & $\sigma_{\text{p}}$ & $1.75~$mrad & $1~\mu$rad \\   \hline
Pulse Duration  & $\Delta t$ & $10$ ns & $10$ ns \\ \hline
Field of View  & $\Omega_{\text{fov}}$ & $10^{-4}$ sr & $10^{-10}$ sr \\   \hline
Frequency Filter & $\Delta \lambda$ & 0.1~pm & 0.1~pm \\  \hline
\end{tabular}
\caption{\textbf{Parameter table for point-to-point links in free-space quantum networks with fixed or mobile devices.} The fixed-setup pointing error $\sigma_{\text{p}}$ is listed for completeness; as discussed in \cref{subsubsec:long_range_fixed}, turbulence dominates over pointing error at the long ranges considered for fixed users, so this value does not itself enter the fixed-channel derivation.}
\label{table:Setups}
\end{table}

The beam is diffracted as it propagates in free space causing the spot to increase in size from its initial value to $w(d)$ which can be given by:
\begin{equation}
w(d)=w_0\sqrt{\bigg(1-\frac{
d}{R_0}\bigg)^2 + \bigg(\frac{d}{d_R}\bigg)^2}, \hspace{0.5cm} d_R \defeq \frac{\pi w_0^2}{\lambda}.
\end{equation}
Here, $d_R$ is the \textit{Rayleigh length}, which is the distance at which the beam's cross-section has doubled. Since the receiver's aperture is of some finite size, some of the beam is lost which introduces the following diffraction-induced transmissivity to the channel:
\begin{equation}
\eta_{\text{diff}}(d)=1-\exp\bigg[\frac{-2a_R^2}{w(d)^2}\bigg].
\end{equation}

The second cause of loss is atmospheric extinction due to Rayleigh scattering and absorption. This can be modelled according to the Beer-Lambert relation:
\begin{equation}
\eta_{\text{atm}}(d,h)=\exp\big[-\alpha(h)~d~\big].
\end{equation}
The value of the extinction factor $\alpha(h)$ depends on the density of the atmosphere. By assuming a standard exponentially decaying model of atmospheric density, the extinction factor $\alpha$ may be given as a function of the height above sea-level and is maximized at ground level:
\begin{equation}
\alpha(h)=\alpha_0\exp(-h/h_0),
\end{equation}
where $h_0=6600\text{m}$ and $\alpha_0\simeq 5\times 10^{-6}~\text{m}^{-1}$ is an estimate for the extinction factor at sea level for a beam with wavelength $\lambda=800$nm \cite{Vasylyev2019Sat}. 

The third cause of loss is from the experimental setup due to inefficiencies at the receiver. For example, loss may be incurred due to coupling the incoming beam into an optical fiber at the receiver. Since we wish to remain relatively agnostic to the specifics of any one setup we choose to simply adopt a sensible value for this parameter. Combining each of these factors a first approximation of the channel is a lossy channel with total transmissivity given by combining all three of the transmissivity terms  as follows:
\begin{equation}
\eta(d,h)=\eta_{\text{diff}}(d)~\eta_{\text{atm}}(d,h)~\eta_{\text{eff}},
\label{eq:totaltrans}
\end{equation}
and thus the capacity is given by~\cite{FS}
\begin{equation}
\mathcal{C}(\mathcal{E})=-\log_2\big[1-\eta(d,h)\big],
\label{eq:PLOBFS}
\end{equation}
which is a simple generalization of Eq.~\eqref{PLOBeq}.

This formula may be extended by considering thermal noise and therefore modelling the free-space channel as a thermal loss channel. The main sources of thermal noise for a free-space channel are the environmental background and additional excess noise in the setup. Thus the total average number of photons introduced into the channel $\bar{n}$ is given by:
\begin{equation}
\bar{n}\defeq \eta_{\text{eff}} ~ \bar{n}_b + \bar{n}_{\text{ex}},
\label{eq:noise}
\end{equation}
where $\bar{n}_b$ is the mean number of background environmental photons, collected by a detector with efficiency $\eta_{\text{eff}}$
and $\bar{n}_{\text{ex}}$ is the mean number of photons introduced as excess noise. For a given value of $\bar{n}$ and transmissivity $\eta$ the thermal-loss channel capacity can then be bounded as in~\cref{eq:thermchannelbound}. 

The background photon number may be calculated in terms of various detector parameters and the spectral irradiance $H_{\lambda}=\frac{\pi \lambda}{h c} B_{\lambda}^{\mathrm{sky}}$, where $B_{\lambda}^{\mathrm{sky}}$ is the brightness of the sky in W~m$^{-2}$~nm$^{-1}$~sr$^{-1}$. By combining the pulse duration $\Delta t$, aperture radius $a_R$ and solid angle field of view of the detector $\Omega$ with the spectral irradiance $H_{\lambda}$ we can write:
\begin{equation}
\bar{n}_b=\Delta t ~ \Delta \lambda ~ a_R^2 ~\Omega_{\mathrm{fov}} ~H_{\lambda},
\end{equation}
Values of $H_{\lambda}$ depend on the weather and atmospheric conditions. Assuming that neither the Sun nor Moon is in direct line of sight of the receiver, typically values for the irradiance are given by:
\begin{equation}
    H_{\lambda=800nm} \approx \begin{cases}
    1.9\times 10^{18}, & \text{cloudy day time},\\
    1.9 \times 10^{13},& \text{full Moon, clear night}.
    \end{cases}
\label{eq:irradiance}
\end{equation}
in units of number of photons $\text{m}^{-2}~\text{sr}^{-1}~\text{s}^{-1} \text{nm}^{-1}$. We will choose to assume the worst-case scenario of cloudy daytime operation throughout this paper but are aware that different conditions may improve the performance.

We are now therefore able to provide first approximations for both mobile and fixed free-space quantum channels as thermal lossy channels with the necessary parameters given in \cref{table:Setups}. We note that the fixed setup is characterized by large initial spot sizes and correspondingly sized receiver apertures and a narrow field of view, which is appropriate for scenarios where the target is not moving. On the other hand, the mobile setup is characterized primarily via a wide-angle receiver aperture and smaller initial spot-size and receiver aperture which are befitting of mobile devices. 

However, we have not yet considered one final key source of loss in free-space: beam wandering. The position of the beam is caused to move with respect to the centre of the receiver aperture, due to jitter in devices, pointing error and turbulence in the atmosphere. The effect of this wandering further reduces the transmissivity of the channel. The overall effect of beam wandering is to cause the centroid of the beam to undergo a Gaussian random walk with variance $\sigma^2$.  This variance may be written as the sum of two further variances:
\begin{equation}
\sigma^2=\sigma_{\text{pe}}^2+\sigma_{\text{tb}}^2,
\end{equation}
where $\sigma_{\text{pe}}^2$ is the variance associated with the pointing error caused by misalignment and jitter and $\sigma_{\text{tb}}^2$ is the variance caused by atmospheric turbulence. The strength of turbulent effects can be assessed numerically by invoking the \textit{Rytov number}, $\sigma^2_{\text{Ry}}$, which is a dimensionless parameter describing the size of turbulent effects and defined for a plane wave as:
\begin{equation}
\sigma^2_{\text{Ry}}(d)\defeq~1.23~ C_n^2~ k^{\frac{7}{6}}~ d^{\frac{11}{6}},
\end{equation}
where $C_n^2$ is the index of refraction structure constant and $k$ is the wavenumber of the beam.

If $\sigma^2_{\text{Ry}}\gg1$ then there is strong turbulence and $\sigma^2 \simeq \sigma_{\text{tb}}^2$. Otherwise if $\sigma^2_{\text{Ry}}\sim1$ then there is moderate turbulence and if $\sigma^2_{\text{Ry}}<1$ the effects of turbulence are weak. We shall briefly review the key features that are relevant for each of the two setups we are considering, but we direct interested readers to Refs.~\cite{FS,MasoudFS} for full details.

\subsubsection{Long range fixed users}
\label{subsubsec:long_range_fixed}
For free-space communication over long ranges we expect to experience strong turbulent effects. In this case $\sigma^2 \simeq \sigma_{\text{tb}}^2$, so the pointing-error $\sigma_{\text{p}}$ given for the fixed setup in \cref{table:Setups} is subdominant here and does not otherwise enter the expressions below. The basic idea is that the turbulent effects are so strong that they cannot be resolved by the detector and instead we should average the effect via the spreading of the beam spot-size to the long-term beam waist $w_{\text{lt}}$. Ref.~\cite{MasoudFS} defined a characteristic transmission length $d_i$:
\begin{equation}
d_i\defeq\frac{1}{C_n^2 k^2 l_0^{\frac{5}{3}}}.
\end{equation}

The value $l_0$ is the turbulence inner scale, which describes the distance at which fluctuations in the index of refraction may begin to become correlated. Typically $l_0 \sim 1 ~\text{mm}$. Beyond the distance $d_i$, we expect the effects of turbulence to become pronounced. Turbulence causes further spreading of the beam to the long-term beam waist which is the mean squared radius of the region in which patches of the original beam are found at the receiver. If the distance between the sender and receiver, $d$ is greater than the characteristic distance $d_i$, then it was shown that:
\begin{equation}
w_{\text{lt}}(d)=w(d)\sqrt{1+\frac{8dq}{3kw(d)^2}},
\end{equation}
where $q$ is defined as:
\begin{equation}
q\defeq0.74\sigma^2_{\text{Ry}}(d) Q^{1/6};\hspace{1cm} Q\defeq35.05d/(kl_0^2).
\end{equation}
Alternatively, if $d<d_i$ but the transmission is still affected by moderate-to-strong turbulence then the following expression for $w_{\text{lt}}$ holds:
\begin{equation}
w_{\text{lt}}(d)=w(d)\sqrt{1+1.63~[\sigma_{R_y}^2(d)]^{\frac{6}{5}}\frac{2d}{kw(d)^2}}.
\end{equation}

By using the long term-beam waist, the long term transmissivity can be calculated: 
\begin{equation}
\eta_{\text{lt}}(d)=1-\exp\bigg[\frac{-2a_R^2}{w_{\text{lt}}(d)^2}\bigg],
\end{equation}
which now accounts for the effect of turbulence in addition to diffraction. This term replaces $\eta_{\text{diff}}$ in \cref{eq:totaltrans} and thus the overall transmissivity can be given by:
\begin{equation}
\eta(d,h) =\eta_{\text{lt}}(d)~\eta_{\text{atm}}(d,h)~\eta_{\text{eff}},\label{etalongrange}
\end{equation}
which can be applied in 
Eqs.~\eqref{PLOBeq}
and~\eqref{eq:thermchannelbound} to calculate the capacity of the channel for pure loss and the capacity bound with thermal noise respectively.

\subsubsection{Short range mobile users}
The mobile setup is characterized by its wider field of view and smaller initial spot size, which necessarily imposes short ranges for the communication. Over such ranges the effect of turbulence is effectively negligible. However the effects of the wider field of view incur a larger pointing error than in the fixed case. The variance associated with beam wandering is thus almost entirely caused by pointing error, $\sigma^2\simeq\sigma_{\text{pe}}^2$, and this pointing error variance depends on the distance between sender and receiver as $\sigma_{\text{pe}}^2=\sigma_{\text{p}}^2d^2$ where $\sigma_{\text{p}}$ is the pointing error at the transmitter. It was shown in \cite{FS} that this pointing error induced beam wandering may be described via a fading process since the fluctuations are slow enough that the receiver is capable of resolving the beam-wandering dynamics. 

Fading is a process in which the transmissivity of the quantum channel changes in time according to a stochastic process. In general, a fading channel is given by the ensemble:
\begin{equation}
\mathcal{E}_{F}(\eta_{\max})\defeq \{p(\eta),\mathcal{E}(\eta)\},\hspace{1cm} \eta \in (0,\eta_{\max}),
\end{equation}
where $p$ is the probability density function for the transmissivity $\eta$ and $\mathcal{E}(\eta)$ is the corresponding quantum channel. The capacity of a free-space fading channel with pure loss can be bounded via the following modification of the PLOB bound~\cite{FS}:
\begin{equation}
\mathcal{C}\big[\mathcal{E}_F(\eta_{\max})\big]=-\int_{0}^{\eta_{\text{max}}}~d\eta~ p(\eta) ~ \log_2(1-\eta).
\label{eq:fadingPLOB}
\end{equation}

\begin{figure}[t!]
    \vspace{-0.0cm}
    \centering
    \setlength{\tabcolsep}{0pt}
    \begin{tabular}{c}
    \subfloat[fiber]{\includegraphics[width=0.81\linewidth,trim=0pt 16pt 0pt 2pt,clip]{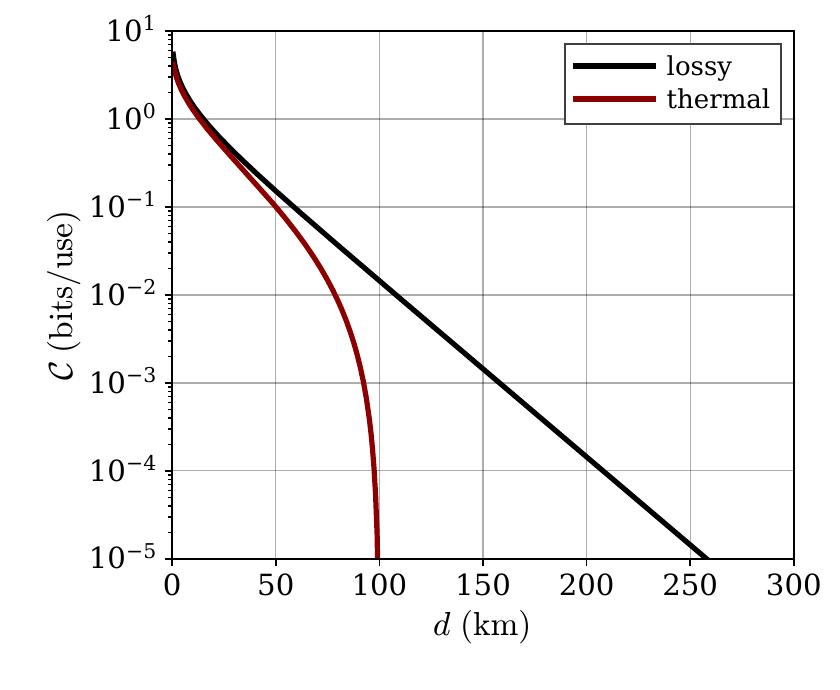}\label{fig:capacity_comp}} \\[-6pt]
    \subfloat[fixed]{\includegraphics[width=0.81\linewidth,trim=0pt 16pt 0pt 2pt,clip]{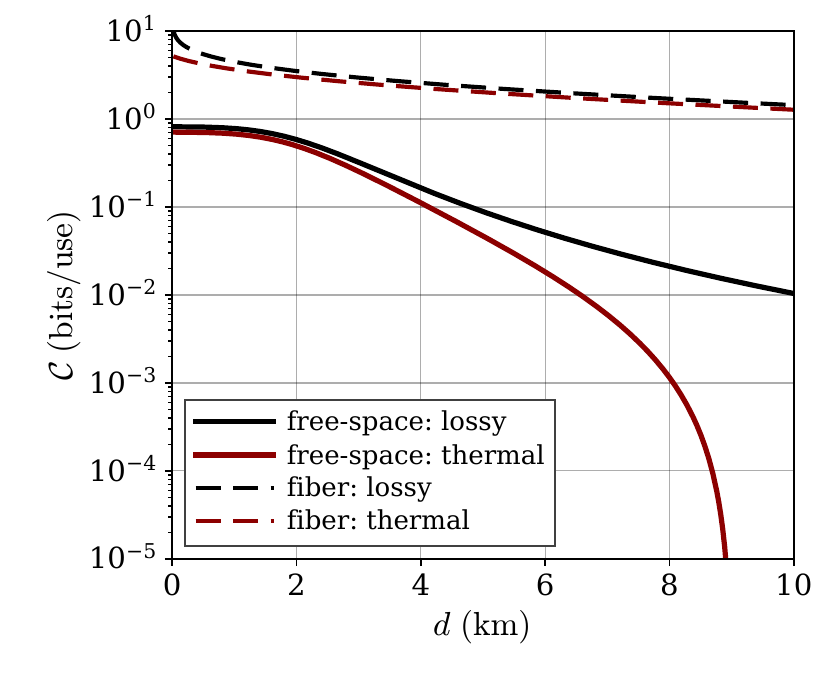}\label{fig:capacity_comp_fixed}} \\[-6pt]
    \subfloat[mobile]{\includegraphics[width=0.81\linewidth,trim=0pt 16pt 0pt 2pt,clip]{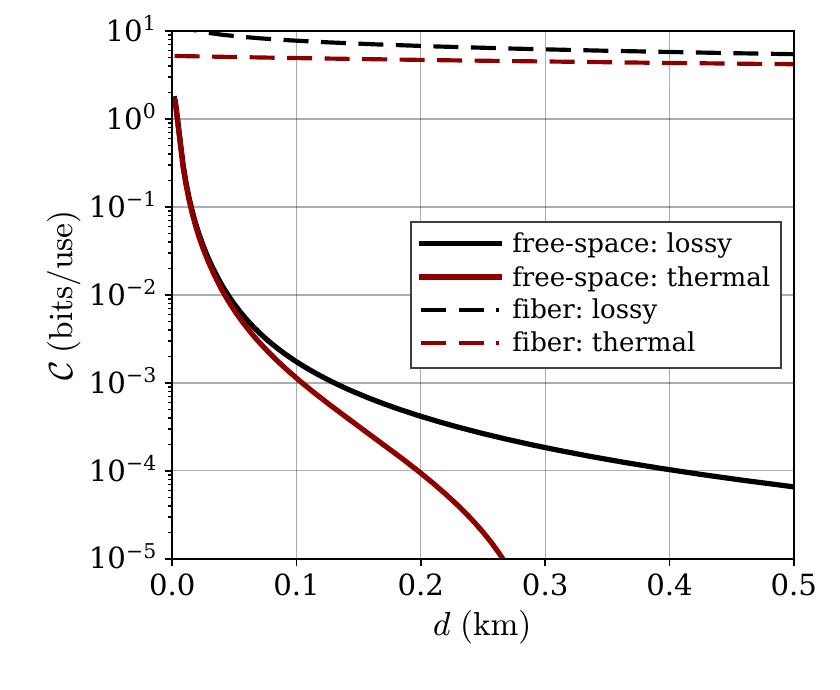}\label{fig:capacity_comp_mobile}}
    \end{tabular}
    \caption{\textbf{Capacity bounds for fiber and free-space quantum channels.} \SP{In each panel, black refers to the capacity bound of a pure-loss channel (``lossy''), while red refers to the capacity bound of a thermal-loss channel (``thermal''). In (a), we show the results for fiber. We plot Eq.~\eqref{PLOBeq} for pure loss and Eq.~\eqref{eq:thermchannelbound} for thermal-loss, assuming $\bar{n}=0.01$. In (b), we consider fixed  (long-range) free-space devices. We plot Eqs.~\eqref{PLOBeq} and~\eqref{eq:thermchannelbound} combined with Eqs.~\eqref{eq:noise} and~\eqref{etalongrange}. In (c), we consider mobile (short-range) free-space devices. Here, we plot the pure-loss formula in Eq.~\eqref{FSplobshort} and the thermal-loss expression in Eq.~\eqref{FSplobTHshort}. Note that, in (b) and (c), the dashed lines reproduce the fiber-based bounds of (a) for comparison. Free-space parameters are chosen as in Table~\ref{table:Setups}}.}
    \label{fig:capacity_comp_composite}
\end{figure}

Considering that the transmissivity follows a Weibull distribution, \cref{eq:fadingPLOB} can be specified into the following capacity bound of a free-space lossy channel~\cite{FS}:
\begin{equation}
\mathcal{C}\big[\mathcal{E}_F(\eta_{\max})\big] \leq -\Delta(\eta_{\text{max}},\sigma_{\text{pe}})\log_2(1-\eta_{\text{max}}),\label{FSplobshort}
\end{equation}
where $\eta_{\max}$ is given by the expression for $\eta$ given in \cref{eq:totaltrans} which corresponds to the case of perfect alignment with no pointing error, while $\Delta(\eta_{\text{max}},\sigma_{\text{pe}})$ is a correction for the imperfect alignment of the beam and receiver and is given by the function:
\begin{equation}
\Delta(\eta,\sigma) \defeq 1 + \frac{\eta}{\text{ln}(1-\eta)} \int_{0}^{\infty} dx \frac{\exp\left[\frac{-\varsigma_0^2}{2\sigma^2}x^{\frac{2}{\kappa}}\right]}{e^x - \eta},
\end{equation}
where $\kappa$ and $\varsigma_0$ are the shape and scale parameters of the Weibull distribution. 

The capacity of the thermal loss channel can then be bounded as~\cite{FS}:
\begin{align}
\mathcal{C}\big[\mathcal{E}_F(\eta_{\max})\big] \leq &-\Delta(\eta_{\text{max}},\sigma_{\text{pe}})\log_2(1-\eta_{\text{max}}) \nonumber \\ &- \mathcal{T}(\eta_{\max},\bar{n},\sigma_{\mathrm{pe}}),\label{FSplobTHshort}
\end{align}
where
\begin{multline}
\mathcal{T}(\eta, \bar{n},\sigma) \defeq \left[ 1 - \exp\left(\frac{-\varsigma_0^2}{2\sigma^2}\left(\text{ln}\left[\frac{\eta}{\bar{n}}\right]\right)^{\frac{2}{\kappa}}\right) \right] \\
\times \left[ \frac{\bar{n}\log_2 \bar{n}}{1-\bar{n}} + h(\bar{n}) \right] - \Delta(\bar{n},\sigma)\log_2(1-\bar{n}).
\end{multline}
Here $h$ is the entropic function of~\cref{eq:h_func} and the thermal number $\bar{n}$ is computed according to Eq.~\eqref{eq:noise} and following equations.

\subsubsection{Comparison}
We now have all the capacity bounds required to model the links in free-space quantum networks. Before we proceed further, let us take the opportunity to compare these capacities. We plot the capacity bounds for fiber-based channels and free-space channels in both fixed and mobile configurations (see \cref{fig:capacity_comp,fig:capacity_comp_fixed,fig:capacity_comp_mobile}). Unsurprisingly fiber-based channels have the best performance, both for pure loss and when thermal noise is considered. Fiber-based quantum channels can maintain a positive capacity for distances up to $\sim100$km, whereas the corresponding maximum distances are $\sim7$km and $\sim0.3$km for fixed and mobile free-space channels. 

Additionally, except for short distances ($\sim2$km for fixed free-space channels; $\sim20$m for mobile free-space channels), the capacity bound for optical fiber links is more than an order of magnitude higher than the equivalent free-space values. Indeed for distances over which the capacity bound for thermal-loss mobile communication is non-zero, the equivalent capacity bound in optical fiber is approximately constant, while in the mobile free-space case its value decays rapidly towards zero. 

The comparatively weak performance of free-space channels highlights that long-distance point-to-point free-space quantum communication is infeasible and that network structures will be required. In particular, the comparatively excellent performance of fiber channels hints that adding a small number of such links to a free-space quantum network may markedly boost the network performance.

\subsection{Quantum Networks, Routingand End-to-End Capacities} \label{subsec:routing_basics}
We can now describe how we shall model quantum networks in a manner which is agnostic to any specific experimental devices. In this case, a quantum network consists of a number of nodes each of which includes a register of quantum states. Pairs of nodes are connected via quantum channels, and states may be exchanged from one node's register to another if a channel exists between them. Such a network may be represented as a weighted graph $\mathcal{N}=(V,E)$, where each network node is represented by a graph vertex $\boldsymbol{v} \in V$. Pairs of vertices $\boldsymbol{i},\boldsymbol{j}$ on the graph share an edge $(\boldsymbol{i},\boldsymbol{j})\in E$ if there exists a quantum channel $\mathcal{E}_{\boldsymbol{i}\boldsymbol{j}}$ between the two corresponding network nodes. The weight of each edge is given by the capacity $\mathcal{C}$ (or the capacity bound) of the channel connecting the two nodes (recall that by capacity we mean the generic two-way assisted capacity, i.e., $Q_2$, $D_2$ or $K_2$).

We define an end-to-end network protocol $\mathcal{P}$ as a means to transport a quantum system from the register of $\boldsymbol{a}$ to the register of $\boldsymbol{b}$ via one or more sequences of network channels. The simplest protocol $\mathcal{P}^\omega$ utilises a single, predetermined path $\omega$ between $\boldsymbol{a}$ and $\boldsymbol{b}$. Such an end-to-end path is defined to be a sequence of distinct network edges connecting the two users:
\begin{equation}
\omega=\big((\boldsymbol{a},\boldsymbol{x}_1),(\boldsymbol{x}_1,\boldsymbol{x}_2),...(\boldsymbol{x}_n,\boldsymbol{b})\big).
\end{equation}
Performing the single-path protocol, $\mathcal{P}^\omega$, involves the sequential transmission of a quantum system from node to node along the path $\omega$ via the corresponding quantum communication channels $\mathcal{E}_{\boldsymbol{a}\boldsymbol{x}_1},\mathcal{E}_{\boldsymbol{x}_1\boldsymbol{x}_2},...\mathcal{E}_{\boldsymbol{x}_n\boldsymbol{b}}$. In general each of the transmissions may be interleaved with local operations and classical communication (LOCCs). Utilising such a path is equivalent to transmission along a repeater chain and the network capacity is given by the minimum of the capacities of the path's constituent quantum channels \cite{ETEarxiv,PirandolaEnd-to-End19}:
\begin{equation}
\mathcal{C}^{\omega}(\mathcal{N})=\underset{(\boldsymbol{i},\boldsymbol{j}) \in \omega}{\min}  \mathcal{C}(\mathcal{E}_{\boldsymbol{ij}}).
\end{equation}

More generally it is possible to optimize the routing strategy between a given pair of end-users to maximize the end-to-end capacity performance. Broadly, such strategies may be grouped into two types: \textit{single-path routing} and \textit{multi-path routing}.  A single-path protocol $\mathcal{P}^s(\boldsymbol{a},\boldsymbol{b})$ utilises sequential transmission of quantum systems from $\boldsymbol{a}$ to $\boldsymbol{b}$ along a single path through the network. Selecting the optimal path is equivalent to solving the widest path problem on the weighted graph and the optimal single-path performance is given by  the single-path capacity $\mathcal{C}^s $\cite{ETEarxiv,PirandolaEnd-to-End19}:
\begin{equation}
\mathcal{C}^s(\mathcal{N},\boldsymbol{a},\boldsymbol{b})=\underset{C}{\min}~\underset{(\boldsymbol{x},\boldsymbol{y})\in\tilde{C}}{\max} 
\mathcal{C}(\mathcal{E}_{\boldsymbol{xy}})
\label{eq:Sp_Cap}
\end{equation}
Here, $C$ refers to a network cut: a bipartition of the vertices of the graph into two sets $A$ and $B$ such that $\boldsymbol{a} \in A$ and $\boldsymbol{b} \in B$. The corresponding cutset $\tilde{C}$ is comprised of the set of network edges which connect a vertex in A with a vertex in B. 

A more powerful routing strategy is to use a multi-path protocol $\mathcal{P}^m(\boldsymbol{a},\boldsymbol{b})$, which utilises multiple paths between $\boldsymbol{a}$ and $\boldsymbol{b}$ in parallel, allowing greater end-to-end performances albeit at the cost of utilising more of the network's resources. The most powerful multi-path routing strategy is a \textit{flooding} protocol, in which each network edge is used exactly once per end-to-end transmission. This leads to the following multi-path (or flooding) capacity $\mathcal{C}^m$, which is derived by application of the min-cut max-flow theorem to quantum networks \cite{ETEarxiv,PirandolaEnd-to-End19}:
\begin{equation}
\mathcal{C}^m(\mathcal{N},\boldsymbol{a},\boldsymbol{b})= \underset{C}{\min} \sum_{(\boldsymbol{x},\boldsymbol{y})\in\tilde{C}} \mathcal{C}(\mathcal{E}_{\boldsymbol{xy}}). \label{eq:Fl_Cap}
\end{equation}

Generally, in a given network there is not typically a single privileged pair of end users $\boldsymbol{a}, \boldsymbol{b}$ whose end-to-end performance is the only metric by which we measure the performance of the network. Instead, we are interested in end-to-end capacity of a protocol $\mathcal{P}$ averaged across all possible pairs of end-users, which we denote $\langle \mathcal{C}^\mathcal{P}(\mathcal{N})\rangle$ and refer to as the mean network capacity for that protocol. It is important to note that in some networks the full set of nodes may not be valid end-users. In this case we have a set of end users $V_e$ and a set of router/repeaters $V_r$ such that $V_e \cup V_r = V$ . In this case the end-to-end performance between routers or between routers and end-users is not considered:
\begin{equation}
\langle\mathcal{C}^{\mathcal{P}}(\mathcal{N})\rangle=\binom{|V_e|}{2}^{-1}\sum_{\boldsymbol{i},\boldsymbol{j}\in V_e} \mathcal{C}^{\mathcal{P}}(\mathcal{N},\boldsymbol{i},\boldsymbol{j}).
\end{equation}
In the case where all the network nodes are in fact valid end-users then $V_e$ is clearly just replaced by the full set of vertices $V$. For large networks, it is typically infeasible to calculate the relevant capacity for every pair of possible end-users in which case we instead make use of a representative sample.

In addition to calculating the mean network performance, we are also interested in assessing the edge consumption of a given network protocol. While flooding provides the ultimate performance limit of any multi-path quantum routing protocol, in practice it is not feasible for a single end-user pair to leverage the resources of the entire network to facilitate their single end-to-end connection. In a realistic network deployment it is highly likely that multiple pairs of users will wish to communicate simultaneously and thus utilising network resources in a scalable manner is an essential feature of routing in quantum networks. To this end we utilise the edge consumption $\tilde{E}$ which is the ratio of the number of edges used by a single communicating end user pair to the total number of network edges, $\tilde{E}(\boldsymbol{a},\boldsymbol{b}) \defeq |E_\mathcal{P}(\boldsymbol{a},\boldsymbol{b})|/|E|$, so that
\begin{equation}
\langle \tilde{E} \rangle = \binom{|V_e|}{2}^{-1}\sum_{\boldsymbol{i},\boldsymbol{j}\in V_e} \tilde{E}(\boldsymbol{i},\boldsymbol{j}).
\end{equation}

\subsection{Random Quantum Networks}
\label{subsec:rand_Q_net}
As with the current classical internet, any future large-scale quantum network is unlikely to have a significantly designed topology. Instead, such networks are likely to grow and develop in an ad-hoc manner \cite{Yook,SF} with the topology and other network properties emergent from an underlying quasi-random process. Moreover in the case of mobile networks, in which users freely travel and form connections, any given set of user locations and network topology merely represents the state of the network at some particular point in time. To this end we study classes of random quantum networks whose parameters may be calibrated to closely approximate likely features of future large scale quantum networks.

We consider the user nodes to take random positions within a convex region $\mathcal{R}$ of the $x-y$ plane, specifically in 
circular regions of various radii. Additionally, the positions of user nodes are drawn from the uniform distribution supported on the region. Generating the nodal positions randomly corresponds to taking a snapshot of the network state at a given time in the mobile user case. Averaging over such instances gives both the long-term time average of a single such network in addition to the overall average of the entire network class. In the fixed user case, averaging over multiple random network instances only corresponds to an overall average of the class of networks.

Having generated the random positions of the nodes, the remaining aspect to be considered is the network edges. We consider the following model for randomly generating network edges:

\begin{definition}[Waxman Network]
A Waxman Network $\mathcal{N}_{\text{Wax}}(n,\mathcal{R},\alpha,\beta,\eta_0,r_0)$ is a simple graph consisting of $n$ nodes generated from a uniform probability distribution on the region $\mathcal{R}$ where the probability that any two nodes $\boldsymbol{i}$ and $\boldsymbol{j}$ share an edge is given by:
\begin{equation}
p[(\boldsymbol{i},\boldsymbol{j})\in E]=\beta\exp\!\left[-\left(\frac{|\vec{r_{\boldsymbol{i}}}-\vec{r_{\boldsymbol{j}}}|}{\alpha r_0}\right)^{\eta_0}\right],
\end{equation}
where $r_0>0$ is a characteristic distance scale, $\alpha>0$ controls the spatial decay of the connection probability, $\beta\in[0,1]$ controls its overall magnitude, and $\eta_0>0$ determines the shape of the distance-decay function.
\end{definition}

We note that in the limit of $\eta_0 \rightarrow \infty$ then probability of connection is one if $|\vec{r_{\boldsymbol{x}}}-\vec{r_{\boldsymbol{y}}}|\leq r_0$ and zero if $|\vec{r_{\boldsymbol{x}}}-\vec{r_{\boldsymbol{y}}}| >  r_0$. Such a network is known as a random geometric graph~\cite{Penrose03RGG}, in which a vertex connects to all other vertices located within a given radius of the vertex. Importantly, the thermal-loss capacities we consider have the property that they go to zero for large enough distances. Thus, given a network with a diameter that is sufficiently large that the capacity of a direct channel connecting two network nodes can be zero, there is inherently a degree of random geometric graph-like structure in the network. In particular, if the nodal positions are randomly generated and we allow the nodes to be as densely connected as possible, then the resulting graph will be a random geometric graph. 

We are therefore interested in modelling free-space quantum networks as Waxman graphs since they capture this dynamic whilst also providing a softer random model that does not fully connect vertices within this feasible radius. The Waxman model therefore allows us to explore more sparse random network models while still inherently maintaining a connection to this capacity dependent geometric structure. Additionally, the exponential decay in connection probability with the distance reflects the likely preference for utilising connections with the highest capacities.

We are also interested in considering router-centered star networks, in which all the end-user nodes are connected solely to a router node. Despite the fact that the topology of such a network is fixed as a star graph, the random positions of the nodes still means that the network is represented by a random graph. Whilst the overall structure of a future quantum internet will be unlikely to have a fully designed topology, router models are still useful to consider as the simplest possible topologies that guarantee interconnecting $n$ users. Additionally if router-centered star networks offer sufficiently good performance then they may well make up sub-networks of the emergent topology  of a large scale quantum network.

\section{Network Capacities of Free-Space Quantum Networks}
\label{sec:baseline}
We are now able to characterize the performance of free-space networks for Waxman and router-centered star networks in both the fixed and mobile configurations. This provides the first characterizations of the performances of such networks and will provide the baseline against which we will compare the improved performance of the equivalent networks provided with an assistive fiber-based backbone. We shall begin with router-centered star networks, where we shall derive the distribution of the end-to-end capacities, given the radial distribution of end-users from the router.

\subsection{Router-Centered Star Networks}
\label{subsec:router}
The basic setup we shall consider is as follows. Consider a circular region $\mathcal{R}\subset \mathbb{R}^2$ whose centre lies at the origin and which has a radius of $R$. The network consists of a single router node $\boldsymbol{s}$ positioned at the centre of the region $\mathcal{R}$ and $n$ user nodes whose positions $\vec{r}_i$ are generated randomly with uniform probability across $\mathcal{R}$. Each end user is connected to the router via a quantum channel $\mathcal{E}_i$ with corresponding capacity $\mathcal{C}_i(|\vec{r_i}|)$
\footnote{We adopt this slight abuse of notation and denote the capacity of a distance to be the capacity of a quantum channel of equivalent distance when the type of channel is clear or need not be specified.} the value of which depends on the distance between the user node and router.

The simple topology of a router-centered star network enables a greater degree of analytical analysis than other more complicated topologies. In particular, it is clear that only single-path routing is possible in such a network, since there is a unique path between two end-users, comprised of the edges between each of the end-users and the router. It is therefore the case that the end-to-end capacity between two users $\boldsymbol{a}$ and $\boldsymbol{b}$ is determined by the minimum of the capacities of these two channels:
\begin{equation}
\mathcal{C}(\mathcal{N},\boldsymbol{a},\boldsymbol{b})=\min \bigg(\mathcal{C}^s(|\vec{r}_{\boldsymbol{a}}|), \mathcal{C}^s(|\vec{r}_{\boldsymbol{b}}|)\bigg),
\end{equation}
where we have dropped the indication of single-path routing since this is now implied by the topology. Let us assume that all of the nodes in the network are connected to the router by the same type of bosonic quantum channels. In this case, the capacity of any network channel connecting depends only on the distance to the router and by the monotonicity of the capacity we simply need to find the maximum of the two distances. We write the cumulative distribution as:
\begin{equation}
\Pr(r_{\mathrm{max}} \leq r) =\Pr\big(|\vec{r}_{\boldsymbol{a}}|\leq r,\ |\vec{r}_{\boldsymbol{b}}|\leq r\big),
\end{equation}
and since each of these distances is independently drawn from the same underlying uniform distribution on the circle:
\begin{equation}
\Pr(r_{\mathrm{max}}\leq r) = \left(\frac{r^2}{R^2}\right)^2 = \frac{r^4}{R^4}.
\end{equation}
Differentiating this cumulative distribution gives the corresponding probability density function:
\begin{equation}
p(r_{\mathrm{max}}) = \frac{4r_{\max}^3}{R^4}.
\end{equation}

We now make use of standard techniques concerning functions of random variables~\cite{taboga-stats,Gazi2023}, in particular the following proposition, which enables us to calculate the resulting network capacity distribution under strict monotonicity assumptions.
\begin{proposition}
\label{prop:rand_var}
Let $\mathcal{X}$ be a random variable with support $R_\mathcal{X}$ and probability density function $p(x)$. Let $\mathcal{C}$ be a strictly decreasing function on $R_\mathcal{X}$. Then the random variable $\mathcal{Y}=\mathcal{C}(\mathcal{X})$ has support on $R_\mathcal{Y}$ given by:
\begin{equation}
R_\mathcal{Y}=\{y=\mathcal{C}(x)\ |\ x \in R_\mathcal{X} \},
\end{equation}
and pdf $p_{\mathcal{Y}}$ given by 
\begin{equation}
p_{\mathcal{Y}}(y)=
\begin{dcases*}
~\frac{-p_{\mathcal{X}}\big(\mathcal{C}^{-1}(y)\big)}{\mathcal{C}'\big(\mathcal{C}^{-1}(y)\big)}, & $\ y \in R_{\mathcal{Y}}$
\\
~0, & otherwise
\end{dcases*}
\end{equation}
\end{proposition}

\noindent Applying this to the case of the PLOB bound $\mathcal{C}(r)=-\log_2(1-10^{-\gamma r})$, i.e., for a pure-loss channel, we get
\begin{equation}
p(\mathcal{C})=
\begin{dcases*}
\frac{-4\big(\log_{10}(1-2^{-\mathcal{C}})\big)^3}{(\gamma R)^4\log_2(10)(2^\mathcal{C}-1)}, &  $ \mathcal{C} \in 	\left[\mathcal{C}(R),\mathcal{C}(0)\right]$\\
~0, & otherwise
\end{dcases*}.
\label{eq:PLOB_pdf}
\end{equation}

\SP{It is important to note that the assumption of strict monotonicity does not generally apply to the thermal-loss capacity bound unless the distance of the links is confined in a region where the bound is non-zero, i.e., for which Eq.~\eqref{eq:thermchannelbound} [Eq.~\eqref{FSplobTHshort} for mobile] is strictly positive.
Only in such a thermal non-zero region can we derive $p(\mathcal{C})$ for the thermal-loss channel using Proposition~\ref{prop:rand_var}. In the general case, the value of $p(\mathcal{C})$ needs to be computed numerically (and not using the inversion provided by  Proposition~\ref{prop:rand_var}). Starting from $p(\mathcal{C})$, we can then compute the mean network capacity bound:}
\begin{equation}
\langle \mathcal{C}^s(\mathcal{N})\rangle=\int_{\mathcal{C}(R)}^{\mathcal{C}(0)} \ \mathcal{C}~ p(\mathcal{C}) \ d\mathcal{C}.
\label{eq:net_mean_integral}
\end{equation}

We plot the average network capacity bound $\mathcal{C}^s(\mathcal{N})$ in \cref{fig:fiber_router_dist,fig:router_dist_fixed,fig:router_dist_mobile} for router-centered star networks comprised of fiber, fixed-user free-space channels and mobile-user free-space channels. Additionally, we provide numerical results for the distribution $p(\mathcal{C})$ in \cref{fig:router_capacity_pdf_fibre,fig:router_capacity_pdf_fixed,fig:router_capacity_pdf_mobile} and directly compare it with results from the simulated router-centered star networks.

\begin{figure*}[t]
    \centering
    \setlength{\tabcolsep}{0pt}
    \begin{tabular}{cc}
    \subfloat[fiber]{\includegraphics[width=0.42\linewidth,trim=0pt 28pt 0pt 3pt,clip]{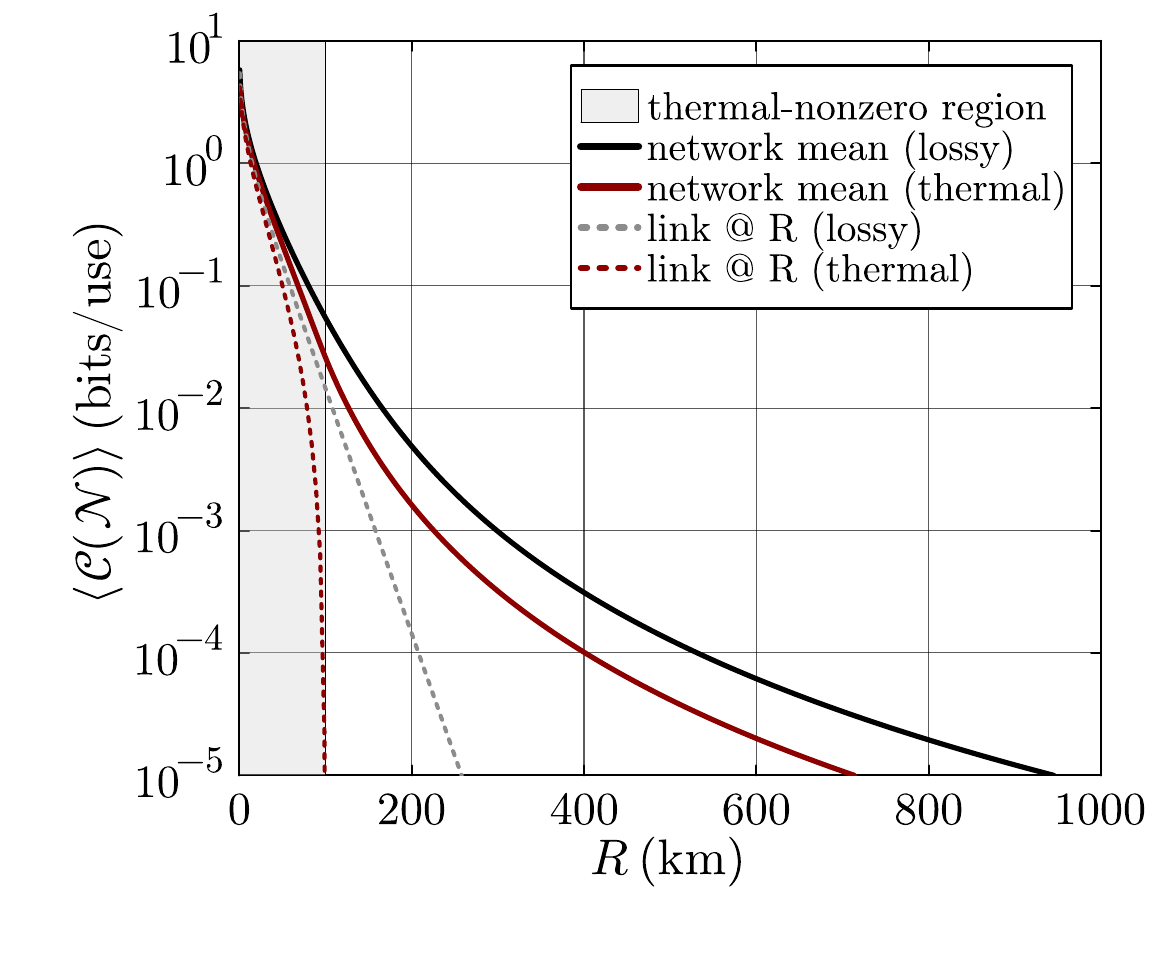}\label{fig:fiber_router_dist}} &
    \subfloat[fiber]{\includegraphics[width=0.42\linewidth,trim=0pt 28pt 0pt 15pt,clip]{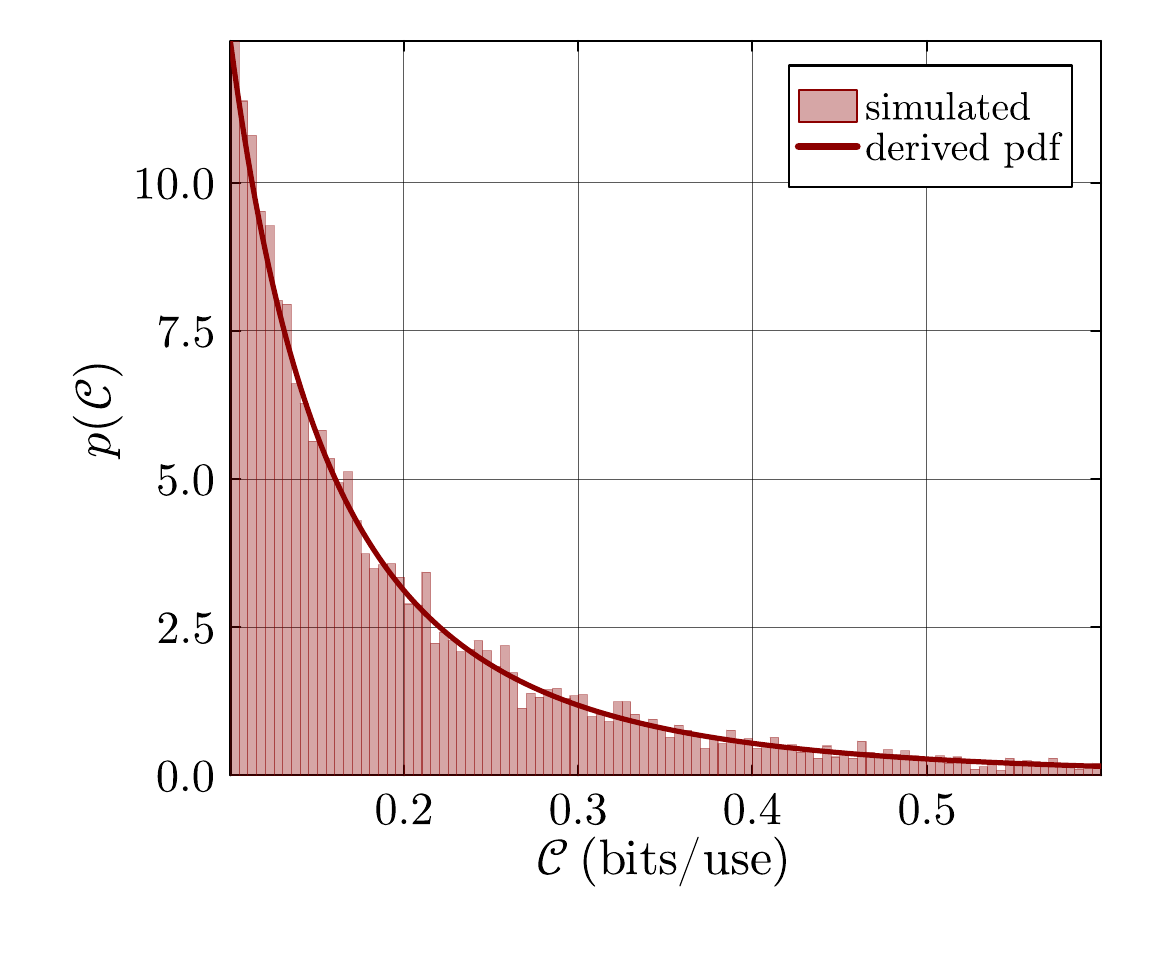}\label{fig:router_capacity_pdf_fibre}} \\[-6pt]
    \subfloat[fixed]{\includegraphics[width=0.42\linewidth,trim=0pt 28pt 0pt 3pt,clip]{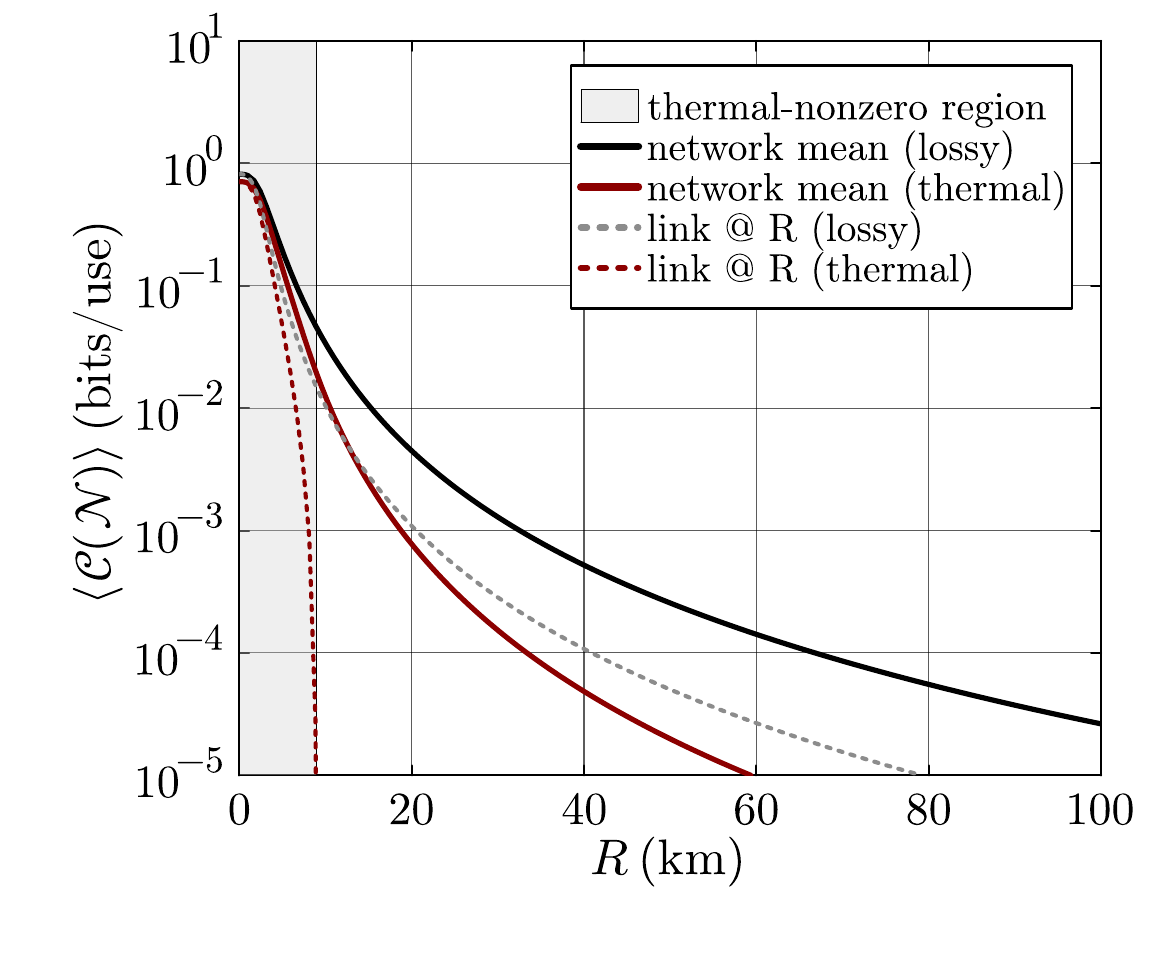}\label{fig:router_dist_fixed}} &
    \subfloat[fixed]{\includegraphics[width=0.42\linewidth,trim=0pt 28pt 0pt 15pt,clip]{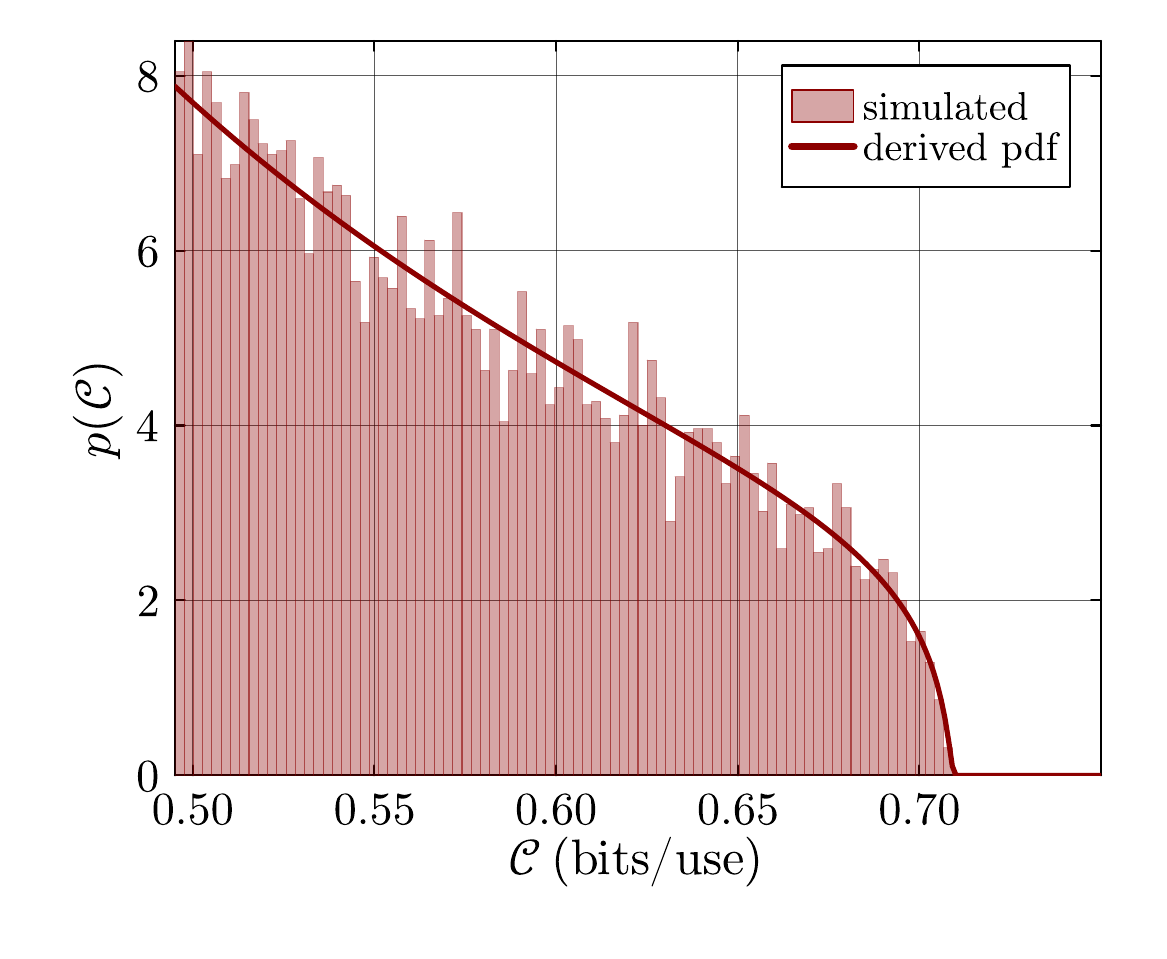}\label{fig:router_capacity_pdf_fixed}} \\[-6pt]
    \subfloat[mobile]{\includegraphics[width=0.42\linewidth,trim=0pt 28pt 0pt 3pt,clip]{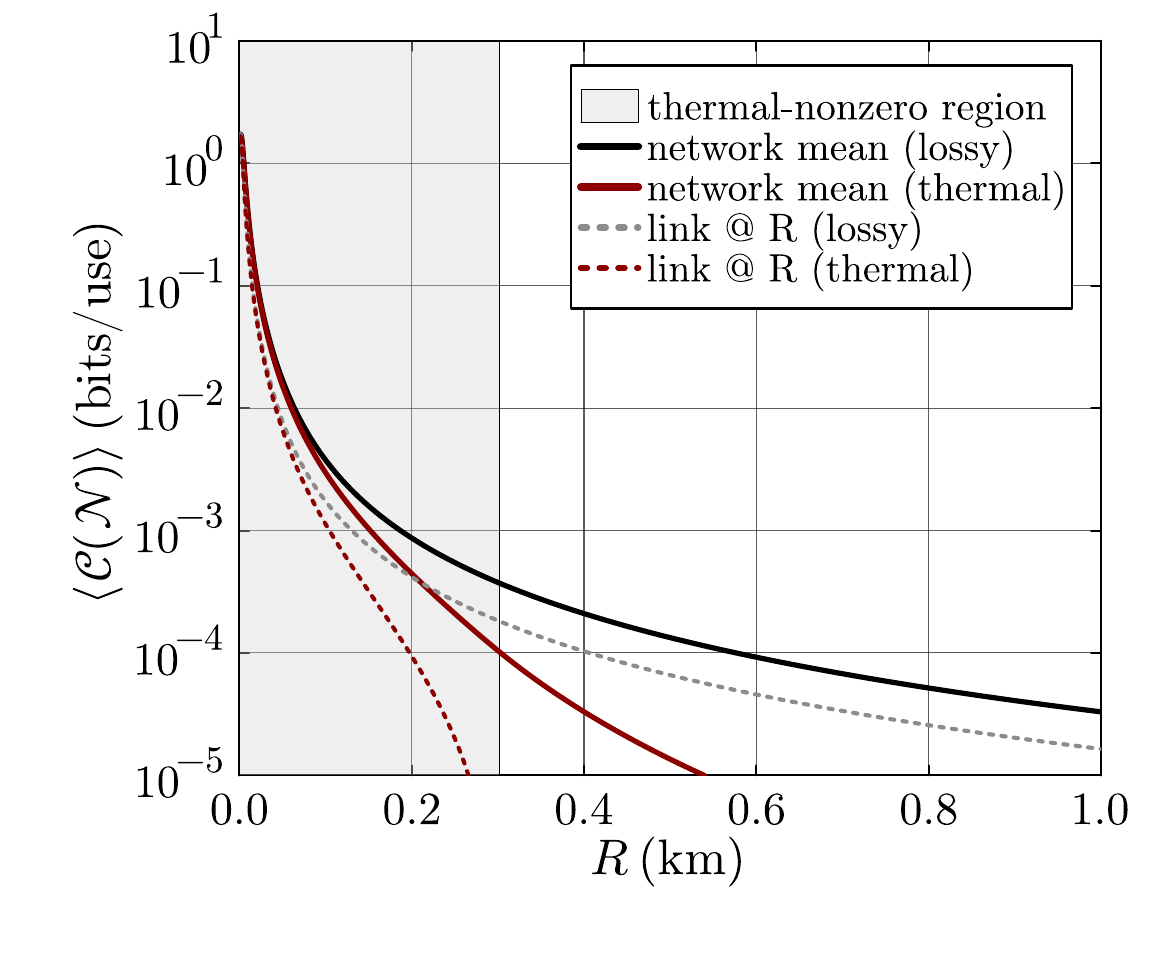}\label{fig:router_dist_mobile}} &
    \subfloat[mobile]{\includegraphics[width=0.42\linewidth,trim=0pt 28pt 0pt 15pt,clip]{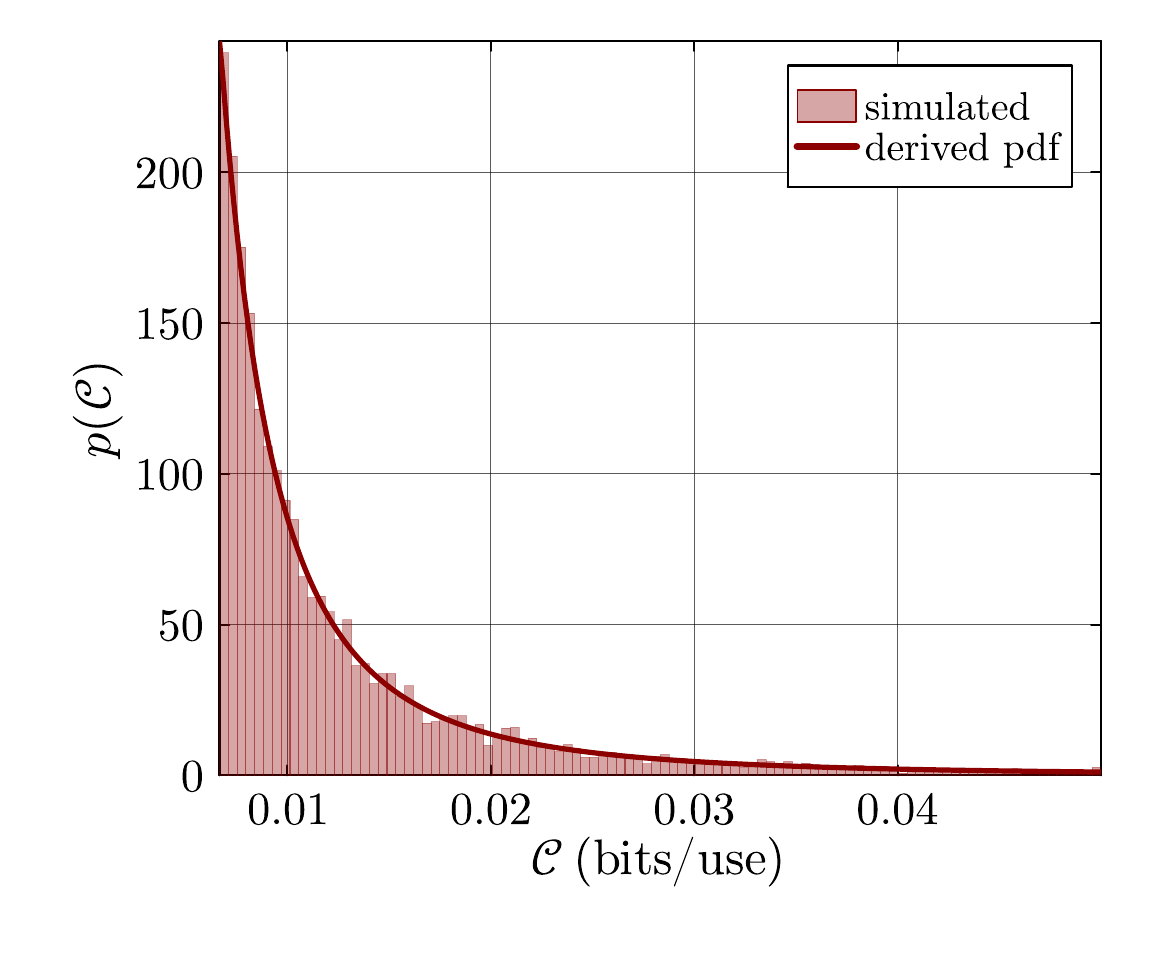}\label{fig:router_capacity_pdf_mobile}}
    \end{tabular}
    \caption{\textbf{Router-centered star networks}. 
    \SP{The left column shows the mean network capacity bound as a function of the router radius $R$ for (a) fiber-connected devices, (c) fixed free-space devices, and (e) mobile free-space devices. Since a router-centered star network has no independent nodal-density parameter, its geometry is characterized solely by $R$. Solid curves show the mean network capacity bound $\langle\mathcal{C} (\mathcal{N})\rangle$, obtained by numerical integration of Eq.~\eqref{eq:net_mean_integral}, for lossy channels (black) and thermal-loss channels (red). The free-space channels are modeled using the parameters specified in \cref{table:Setups} while for fibers we assume $\bar{n}=1$.
    Dotted curves show the corresponding single-link capacity evaluated at the maximum link distance $R$. The shaded regions indicate the values of $R$ for which the thermal-loss capacity bound remains strictly positive, i.e., for which Eq.~\eqref{eq:thermchannelbound} [Eq.~\eqref{FSplobTHshort} for mobile] is $>0$.
    Numerically, the shaded regions end at $\approx100$~km for fiber, $\approx9.01$~km for fixed free-space, and $\approx0.30$~km for mobile free-space. The right column shows the end-to-end capacity distributions for router-centered star networks with thermal-loss links and containing $10^4$ end users, for the cases of: (b) fiber with $R=50~\mathrm{km}$, (d) fixed free-space with $R=2~\mathrm{km}$, and (f) mobile free-space with $R=50~\mathrm{m}$ (all these distances are within the thermal non-zero regions). Histograms represent simulated points, while solid curves show the probability density functions derived from \cref{prop:rand_var}.} 
}
    \label{fig:router_composite}
\end{figure*}

\begin{figure*}[h!]
    \centering
    \setlength{\tabcolsep}{0pt}
    \begin{tabular}{cc}
    \subfloat[mobile]{\includegraphics[width=0.42\linewidth,trim=0pt 27pt 0pt 3pt,clip]{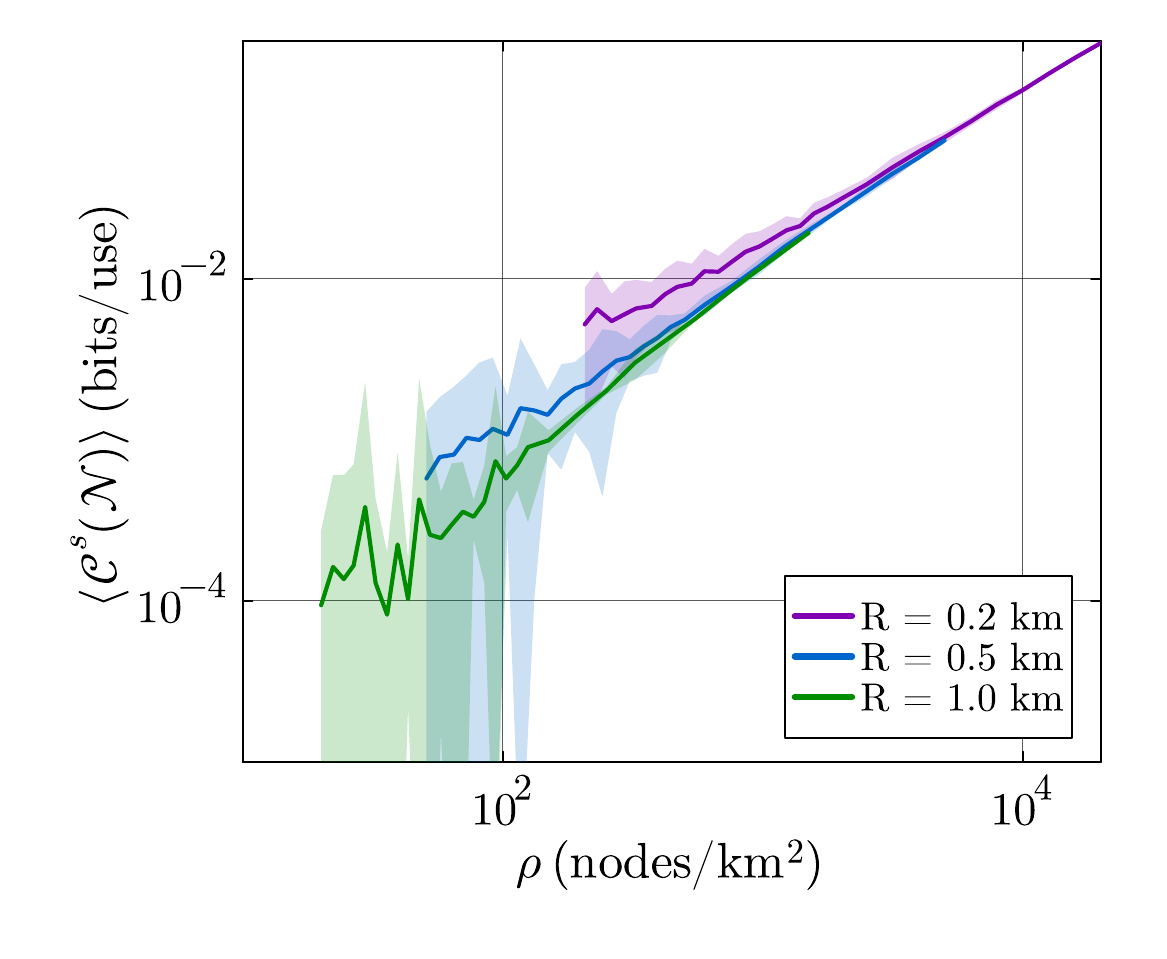}\label{fig:baseline_wax}} &
    \subfloat[fixed]{\includegraphics[width=0.42\linewidth,trim=0pt 27pt 0pt 3pt,clip]{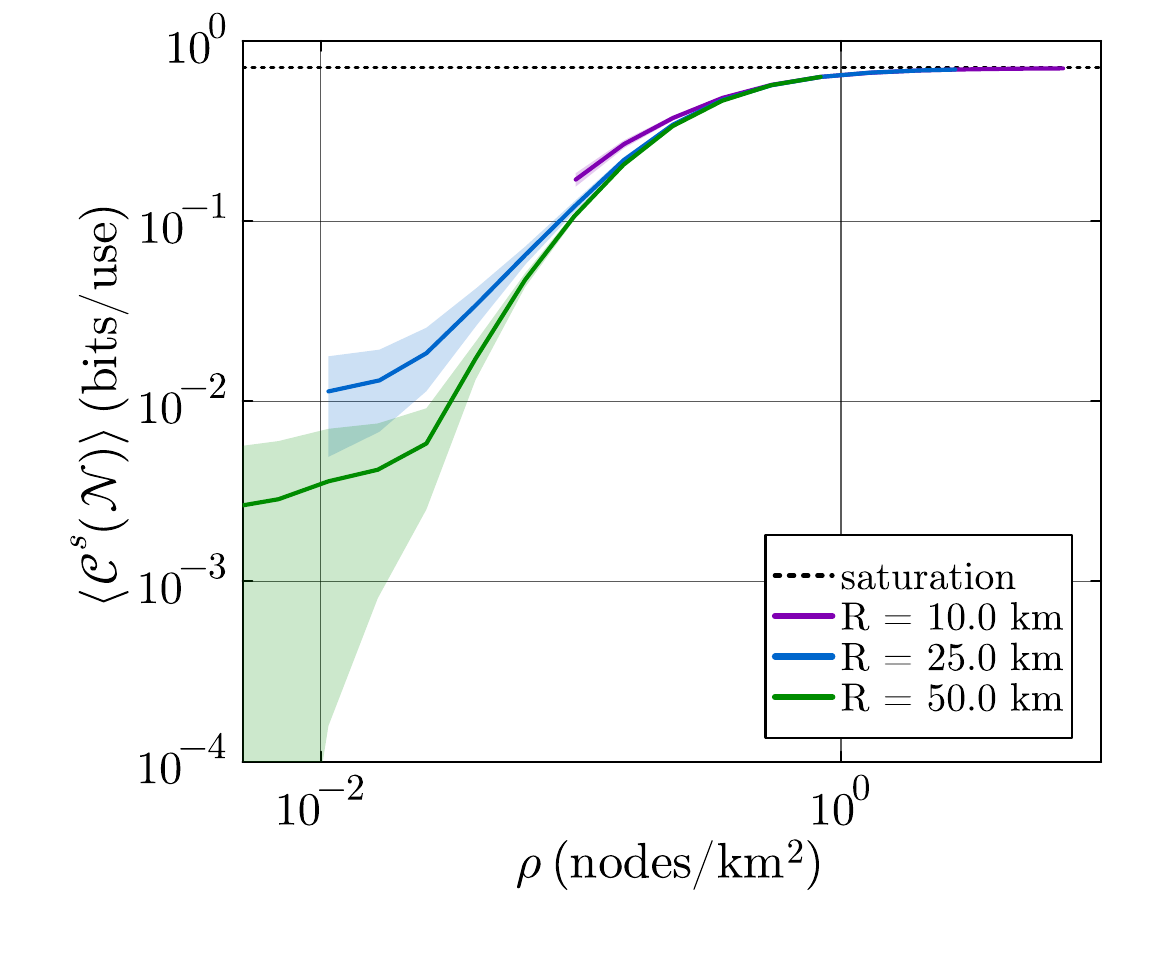}\label{fig:baseline_wax_fixed_cap}} \\[-6pt]
    \subfloat[mobile]{\includegraphics[width=0.42\linewidth,trim=0pt 27pt 0pt 3pt,clip]{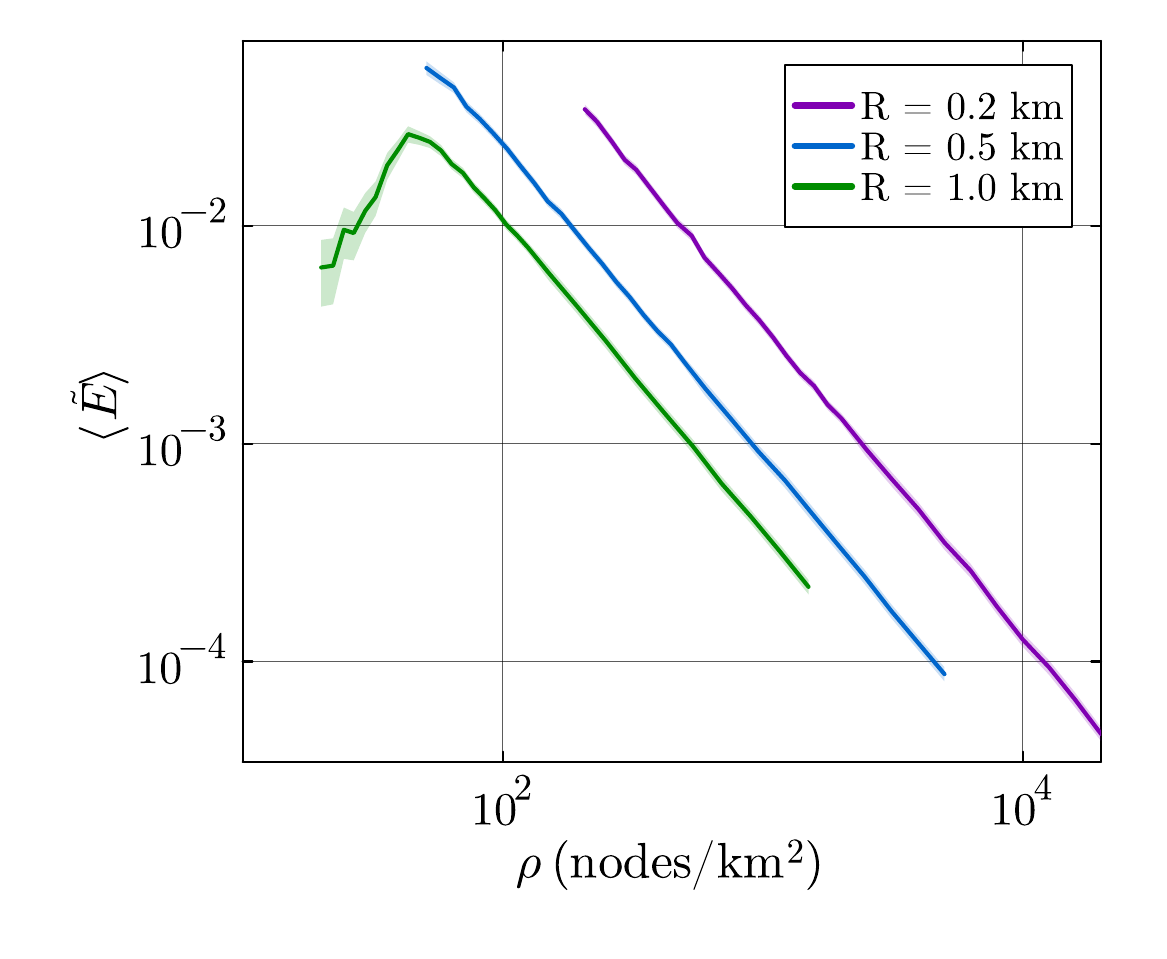}\label{fig:baseline_wax_mobile_cons}} &
    \subfloat[fixed]{\includegraphics[width=0.42\linewidth,trim=0pt 27pt 0pt 3pt,clip]{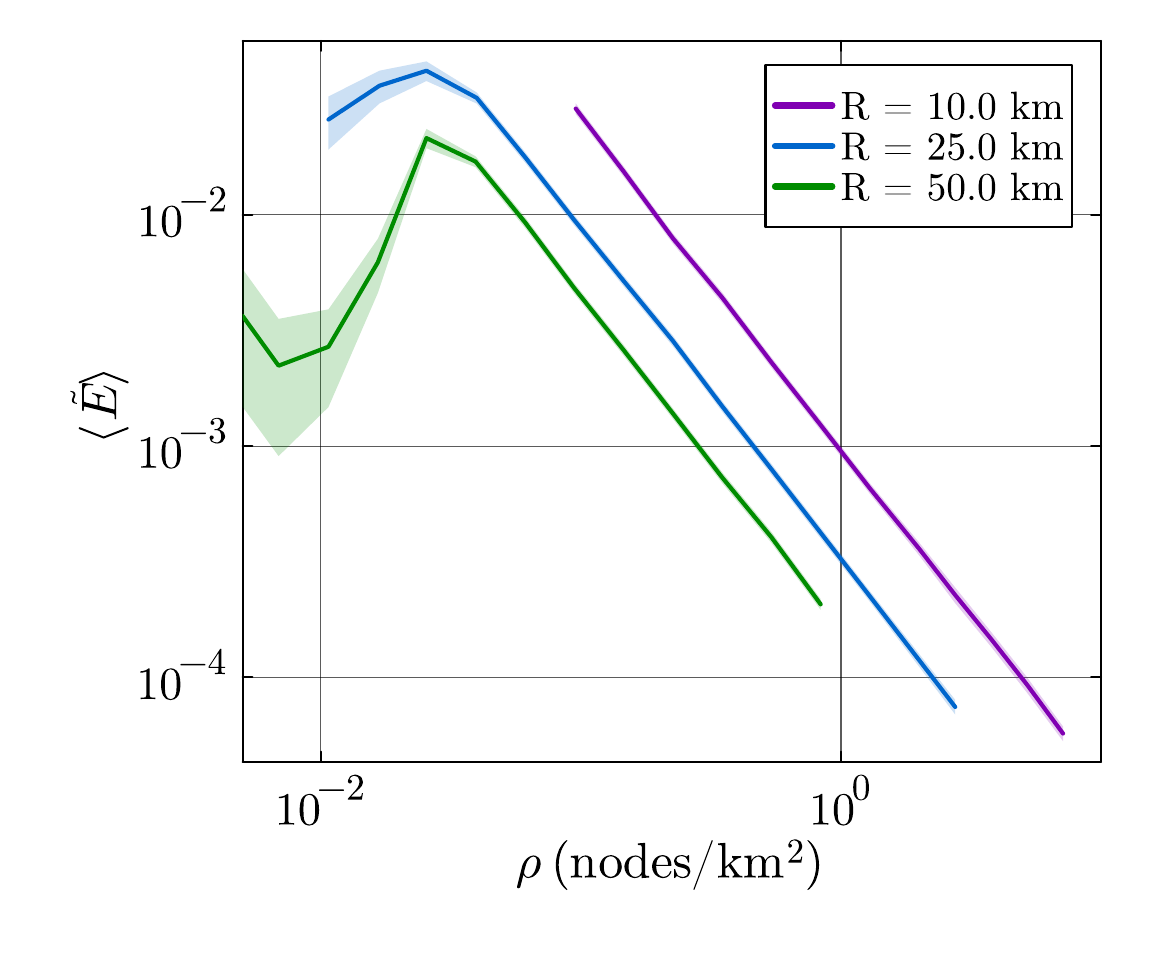}\label{fig:baseline_wax_fixed_cons}} \\[-6pt]
    \subfloat[mobile]{\includegraphics[width=0.42\linewidth,trim=0pt 27pt 0pt 3pt,clip]{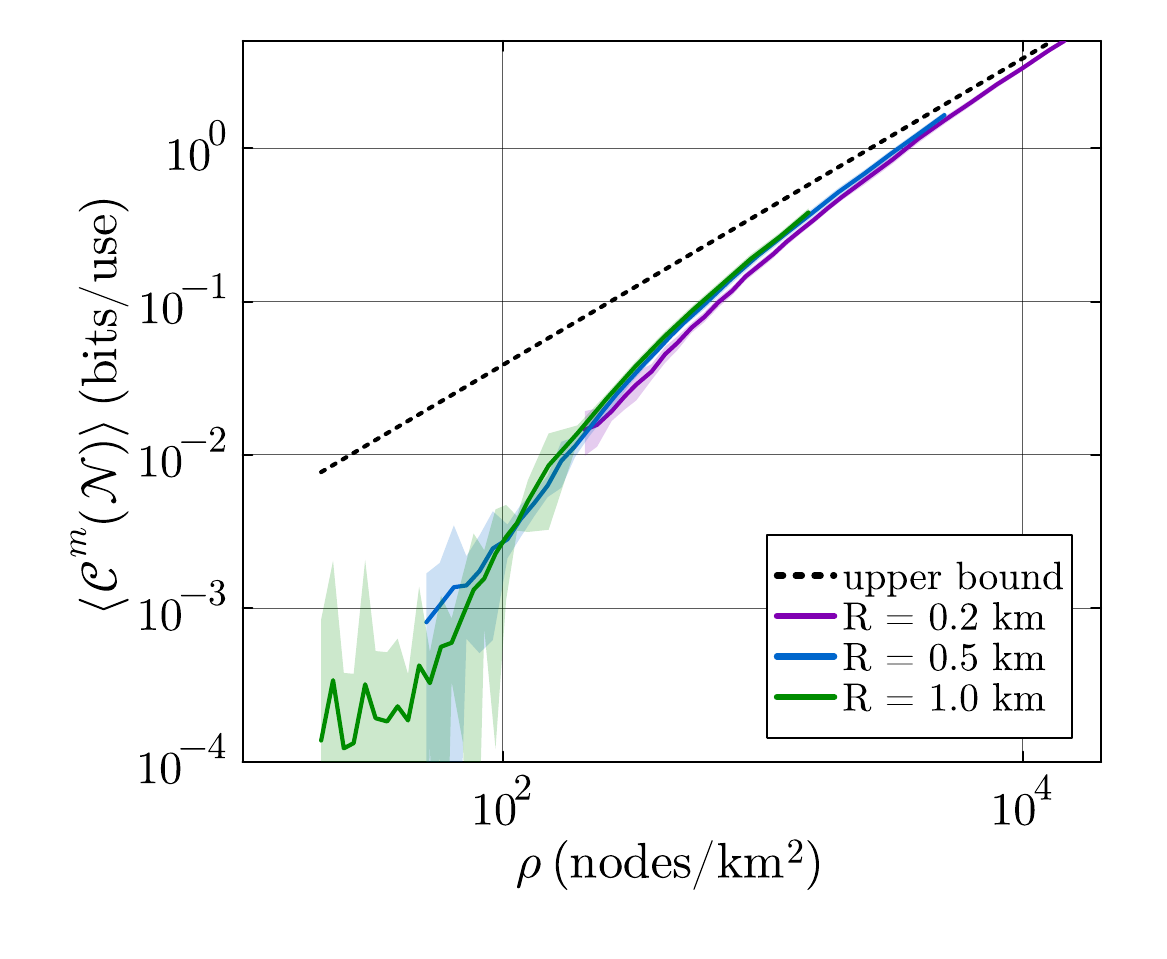}\label{fig:baseline_wax_mobile_mp}} &
    \subfloat[fixed]{\includegraphics[width=0.42\linewidth,trim=0pt 27pt 0pt 3pt,clip]{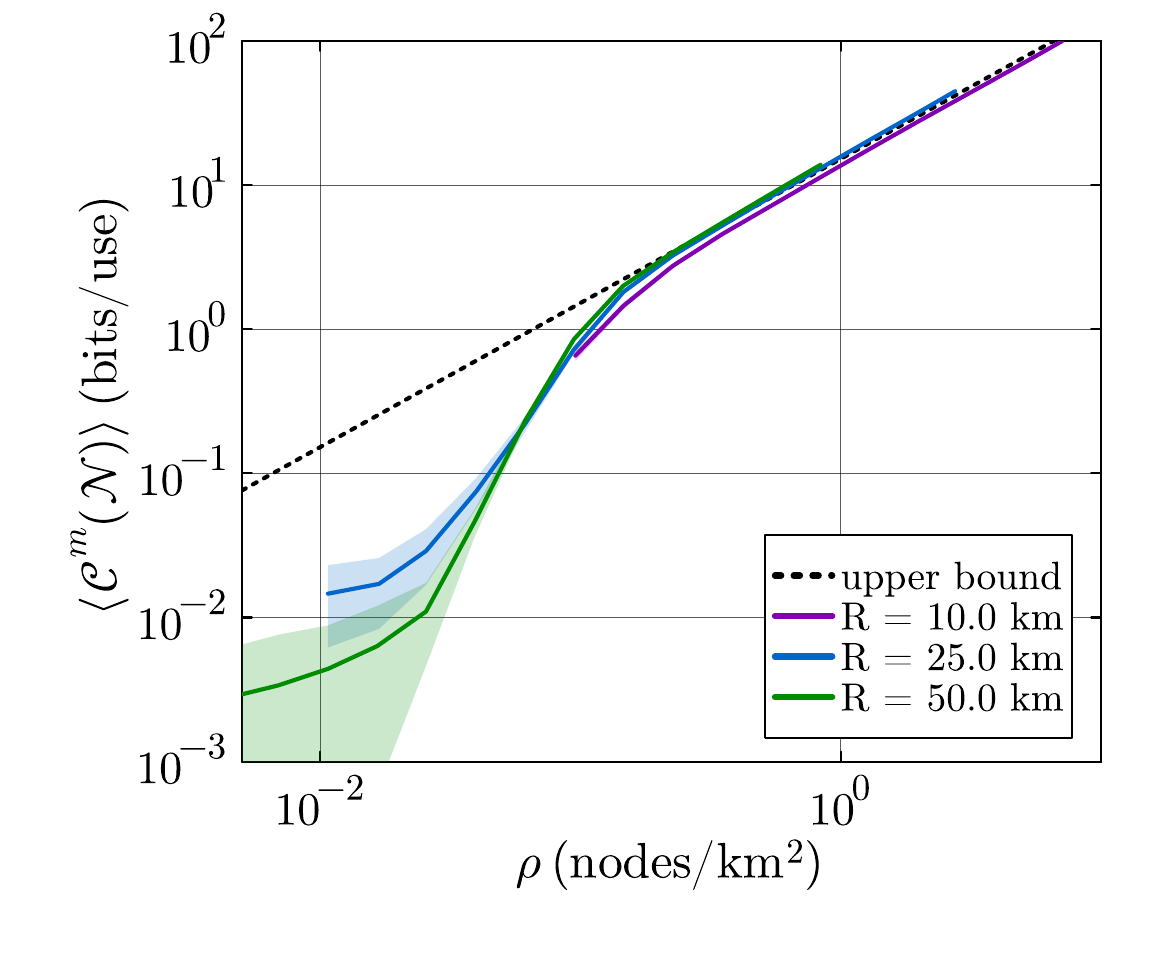}\label{fig:baseline_wax_fixed_mp}}
    \end{tabular}
    \caption{\textbf{Free-space Waxman networks.} Left column: mobile users. Right column: fixed users. (a) Mobile single-path capacity. (b) Fixed single-path capacity. (c) Mobile single-path edge consumption. (d) Fixed single-path edge consumption. (e) Mobile multi-path capacity. (f) Fixed multi-path capacity. Shaded bands show $\pm1$ standard error of the mean across network instances. In the capacity panel (b), the black dotted line is the saturation value: the maximal fixed free-space 
    capacity (computed at zero distance between two nodes). In the multi-path panels (e,f) the dotted line is the flooding common upper bound obtained from end-user neighbourhood cuts [see ~\cref{eq:flooding_upper_bound}]. The free-space channels are modeled as in \cref{table:Setups}.}
    \label{fig:baseline_wax_composite}
\end{figure*}

\subsection{Waxman Networks}\label{subsec:wax_net_baseline}
Let us now present the results for free-space Waxman quantum networks in \cref{fig:baseline_wax,fig:baseline_wax_mobile_cons,fig:baseline_wax_mobile_mp,fig:baseline_wax_fixed_cap,fig:baseline_wax_fixed_cons,fig:baseline_wax_fixed_mp} for various nodal densities $\rho$ (nodes/km$^{2}$). The results are taken from standard Waxman models with the parameters $\alpha=\beta=\eta_0=1$ and we choose to select the value of $r_0$ as the distance such that $\mathcal{C}(r_0)=0.1~$bits per use. 
\SP{In this study, we assume thermal-loss channels, so the capacity bounds are computed starting from the thermal-loss PLOB bound. More precisely, in the case of fixed (long-range) free-space devices, the link rates are given by Eq.~\eqref{eq:thermchannelbound} combined with Eqs.~\eqref{eq:noise} and~\eqref{etalongrange}, while for mobile (short-range) free-space devices, the link rates are based on the formula in Eq.~\eqref{FSplobTHshort}. Single-path and multi-path end to end capacity bounds are then computed accordingly.}

At high densities the single-path capacity bounds for the different sized regions converge to one another. In the fixed user case, the network capacity bounds approach the single-path upper bound, which is the free-space capacity bound between two fixed devices separated by zero distance (with other parameters as in Table~\ref{table:Setups}). In the mobile user case, the single-path network performance is far below the ultimate single-path upper bound but the network single-path capacity bounds still converge, albeit to a still increasing function. 

The multi-path capacity bounds show similar convergence to a common upper bound. In this case the common upper bound corresponds to the case in which the minimum cut becomes a neighbourhood cut around either $\boldsymbol{a}$ or $\boldsymbol{b}$. In fact it is relatively easy to see that this must occur asymptotically almost surely with high density. Thus the high density performance of flooding is dependent only on the nodal density and the capacities of the network channels. It was shown in Ref.~\cite{QuntaoRandQNets} that the multi-path capacity of fiber Waxman random quantum networks may be upper bounded by a linear relationship with the nodal density. The bound arises and is asymptotically tight due to the fact that the minimum capacity cut between two users $\boldsymbol{a}$ and $\boldsymbol{b}$ is asymptotically almost surely a neighbourhood cut. The expected value of such a cut can be approximated by~\footnote{This formula, proven for a pure-loss channel in Ref.~\cite{QuntaoRandQNets}, can be readily extended to a thermal-loss channel.}:
\begin{equation}
\langle\mathcal{C}(\mathcal{N})\rangle\simeq2\pi \rho \int_0^\infty r~\exp\!\left(-\left(\frac{r}{\alpha r_0}\right)^{\eta_0}\right) \mathcal{C}(r) ~dr,
\end{equation}
and therefore $\langle \mathcal{C}(\mathcal{N})\rangle \sim \xi \rho$ where:
\begin{equation}
\xi=2\pi  \int_0^\infty r~\exp\!\left(-\left(\frac{r}{\alpha r_0}\right)^\eta\right) \mathcal{C}(r) ~dr.
\label{eq:flooding_upper_bound}
\end{equation}

At lower densities the characteristic phase transition in the multi-path network capacity may be observed. The capacity rapidly increases from a very small value to approaching the common upper bound. This phase transition mirrors that of the well-known phase transitions for the connectivity of Waxman networks.

The edge consumption is presented only for the single-path routing since flooding by definition utilises every network edge. The consumption plots show an interesting feature at low densities, which is easiest seen in \cref{fig:baseline_wax_fixed_cons}. The consumption falls off, then rebounds before decaying again. This behaviour too corresponds to the phase transitions in Waxman networks~\footnote{At very low densities, the network has very sparse network construction and nodes are mostly isolated. Any two end-user nodes are therefore disconnected with high probability and unable to communicate. These instances contribute nothing to the edge consumption. On the occasion that a pair of end-user are connected this will generally correspond to short paths and direct connections. Nonetheless this will consume a relatively large number of network edges in comparison to the small overall size of the network. As the density increases, the network mostly stays disconnected so that the path length remains small but the number of overall edges increases, resulting in decreased consumption when two end-users are able to connect. At a critical density the network undergoes a phase transition forming a giant component with high probability. This is a connected component of the graph which contains a large fraction of the overall nodes. This transition causes the consumption to increase as many more pairs of nodes are connected to via longer paths. After the phase transition the edge consumption peaks. At this point, the straight line on the logarithmic plot is established and is indicative of a power law relationship. Since the number of edges in a network scales by at most $\mathcal{O}(n^2)$ and assuming then that the number of edges used in a path should scale in $n$ at some lower rate then we should get a power law relationship where $\langle \tilde{E} \rangle \sim k\rho^{-m}$ $m \in (0,2)$ for large densities. Across both mobile and fixed configurations, we find similar values of $m$, consistent with this expected range.}.

\section{Optimal Backbone Design}
\label{sec:optimal_bb_design}
We are now able to introduce fiber-based backbone structures to the free-space networks in order to improve their overall performance. Large scale network structures in the current classical internet are very frequently hierarchical. Rather than one, enormous, global network, we have a `network of networks'; in such a structure, a backbone is a core part of a network, comprised of dedicated network devices, facilitating the connection of multiple smaller networks. In our case, the backbones we shall consider will consist of a small number of additional nodes interconnected via fiber-based links. 

Voronoi-based backbone topology design was previously studied for quantum key distribution networks in Ref.~\cite{Alleaume2009TopologicalOptQKD}, which modelled a QKD backbone network as the Delaunay graph dual to a Voronoi partition of the served nodes and optimized the partition to minimize deployment cost subject to a coverage constraint. Our approach differs in objective and setting: we optimize the Voronoi generators directly for end-to-end capacity, approximating the capacity-maximizing tessellation by a centroidal Voronoi tessellation obtained via Lloyd's algorithm, and we couple the resulting fiber backbone to a random free-space access network rather than treating the full topology as fixed infrastructure to be costed. 

Let us denote this network as $\mathcal{N}_b=(V_b,E_b)$. We will consider constraining the number of backbone nodes we introduce and ask that the number of such nodes should be small in comparison to the number of nodes in the overall network. Since these nodes will act as repeaters and routers in the network this is a reasonable assumption, since the required quantum memories will be an expensive technology particularly in comparison to the links which will be comprised of comparatively cheap and straightforward optical fiber.

The idea of using a backbone to assist in the communication and improve network performance is relatively straightforward. However, the notion of optimal performance is not so clear-cut and many definitions are possible. Let us therefore discuss what it should mean for a backbone network to be optimal before putting together a suitable definition. Firstly, we are working in the scenario where the positions of the nodes and the corresponding edges are not known \textit{a priori}. This is of course appropriate in the mobile case where the positions of the nodes are by definition not fixed. In the fixed user case, the interpretation is slightly less straightforward, nonetheless it provides an average over the network class and should be accurate for a large number of nodes, which is well approximated by an even distribution of nodes across the region. Therefore whatever notion of optimality we settle on it should be calculable from the distribution of backbone nodes alone. 

We have seen that the capacities of free-space quantum channels, whether in the mobile or fixed user scenario, are substantially less than a fiber-based channel of equal length. Thus, so long as the fiber-based backbone channels do not cover excessively long distances, any capacity limiting bottlenecks will occur in the free-space network layer. For example in the multi-path case, the minimum network cut is highly likely to be comprised solely of free-space edges. Similarly in the single-path case, a free-space edge will almost certainly be the rate-limiting bottleneck. It is clear that the distance between two end users affects the end-to-end capacity. Thus, if we reduce the average distance between the nodes then this should result in increased end-to-end performance. In effect the fiber-based links act as a shortcut through the network space reducing the effective distance by which any two end-users are separated. To this end we can introduce the following two definitions:

\begin{definition}[Backbone Distance]
The backbone distance $d_b(\boldsymbol{x})$ for a node $\boldsymbol{x}$ is the distance from that node to the nearest backbone node.
\begin{equation}
d_b(\boldsymbol{x})\defeq\underset{b_i\in V_b}{\min} ~|\vec{r}_{\boldsymbol{x}}-\vec{r}_{b_i} |.
\end{equation}
\end{definition}

\begin{definition}[Effective Separation]
The effective separation, $d_{\mathrm{eff}}(\boldsymbol{x},\boldsymbol{y})$, between two nodes $\boldsymbol{x}$ and $\boldsymbol{y}$ is defined by the following expression:
\begin{equation}
d_{\mathrm{eff}}(\boldsymbol{x},\boldsymbol{y})\defeq \min\big(|\vec{r}_{\boldsymbol{x}}-\vec{r}_{\boldsymbol{y}}|,d_{b}(\boldsymbol{x})+d_b(\boldsymbol{y})\big) .
\end{equation}
\end{definition}

The effective separation captures this notion of providing a shortcut through the spatial region. If two nodes are already located very close together then the existence of the backbone should not affect the end-to-end performance between them. However if the nodes are located further apart they will instead make use of the backbone. In this case we should expect that the performance is equivalent to two nodes in a network without a backbone, separated by a distance equal to the effective separation. We may therefore expect that minimizing the average effective separation between end-user nodes is a useful heuristic for optimising the positions of the backbone nodes. In this case we seek to minimize the following integral:
\begin{equation}
\int_{\mathcal{R}} \int_{\mathcal{R}}~ d_{\text{eff}}(\boldsymbol{x},\boldsymbol{y}) ~ d \vec{r}_{\boldsymbol{x}}~ d \vec{r}_{\boldsymbol{y}}.
\end{equation}

However, this is still difficult to optimize. We simplify the problem by considering only the minimization of the backbone distance~\footnote{This has the additional advantage of also being immediately applicable to the case of router-centered star networks where all end-user nodes are connected directly to the backbone node.}. We can improve the optimization by introducing a weighting according to the capacity function. In this case we can replace minimizing the backbone distance with maximizing the following expression: 
\begin{equation}
\sum_{i=1}^k\int_{\mathcal{R}_i} \mathcal{C}(|\vec{r}-\vec{r}_{b_i}|) ~ d\vec{r},
\label{eq:heuristic int}
\end{equation}
where we have made the simplification of introducing disjoint regions $\mathcal{R}_i$ which tessellate the overall region, each of which correspond to the region in which end user nodes will connect to a particular router node. This leads us to the following definition of an optimal backbone network.

\begin{definition}[Optimal backbone network]
\label{def:opt_bb}
A backbone network $\mathcal{N}_B$ comprised of $k$ nodes with positions $\{r_{b_i}\}$ and associated disjoint regions $\{\mathcal{R}_i\}$ is described as being optimal if the functional $\mathcal{F}$ defined as:
\begin{equation}
\mathcal{F}(\{\vec{r}_{b_i}\},\{\mathcal{R}_i\})=\sum_{i=1}^k\int_{\mathcal{R}_i} \mathcal{C}\big(|\vec{r}-\vec{r}_{b_i}|\big) \ d\vec{r},
\end{equation}
achieves its maximal value.
\end{definition}

\subsection{Voronoi Tessellations}
Before we proceed further, let us introduce two connected notions that will aid in the optimization of backbone structures. Given a set of backbone nodes with positions $\{\vec{r}_{b_i}\}$, the corresponding \textit{Voronoi diagram} is a partitioning of the underlying region $\mathcal{R}_{\mathcal{N}}$ into a set of convex regions $\mathcal{R}_{b_i}$ each of which contain all the points which lie closer to that backbone node than any other. Formally we write:
\begin{equation}
\mathcal{V}_i=\{ \vec{r} \in \mathcal{R}_{\mathcal{N}} ~ | ~ |\vec{r}-\vec{r}_{b_i}| \leq |\vec{r}-\vec{r}_{b_j}| \ \forall j\neq i\}.
\end{equation}
Given the same set of backbone nodes, the corresponding \textit{Delaunay triangulation} is the dual graph to the Voronoi diagram. Two backbone nodes are connected by an edge if their Voronoi cells share a boundary. The two notions are illustrated in \cref{fig:Vor+Del}. We can see that maximizing the expression in \cref{eq:heuristic int} is therefore equivalent to the definition in \cref{def:opt_bb} when the regions $\{\mathcal{R}_i\}$ are the Voronoi regions $\{\mathcal{V}_i\}$. 
\begin{figure*}[t]
    \centering
    \begin{tabular}{ccc}
\subfloat[Voronoi Diagram]{\includegraphics[width=0.3\linewidth]{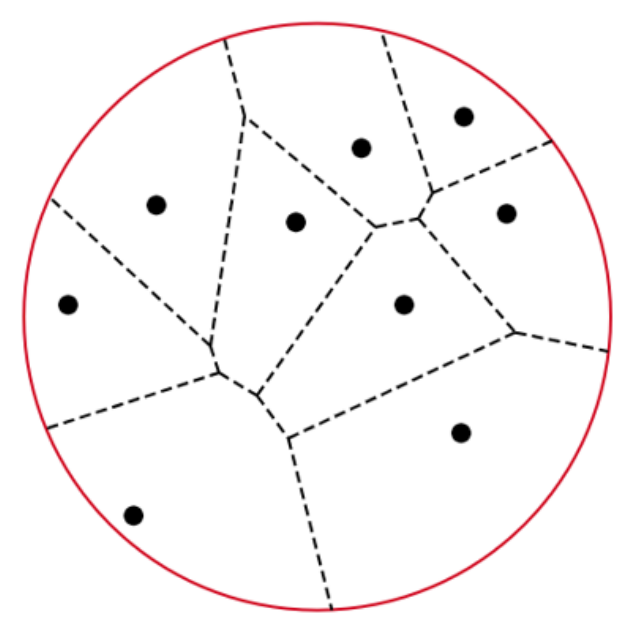}}\hspace{0.05\linewidth}
\subfloat[Dual Objects]{\includegraphics[width=0.3\linewidth]{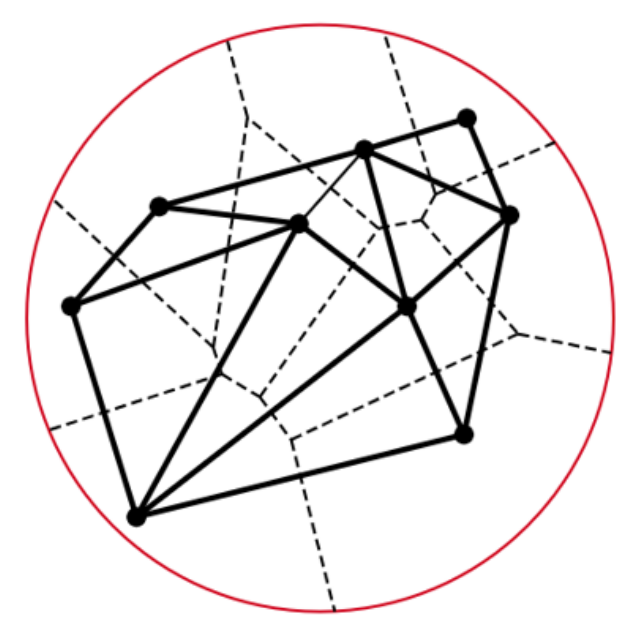}}\hspace{0.05\linewidth}
\subfloat[Delaunay Triangulation]{\includegraphics[width=0.3\linewidth]{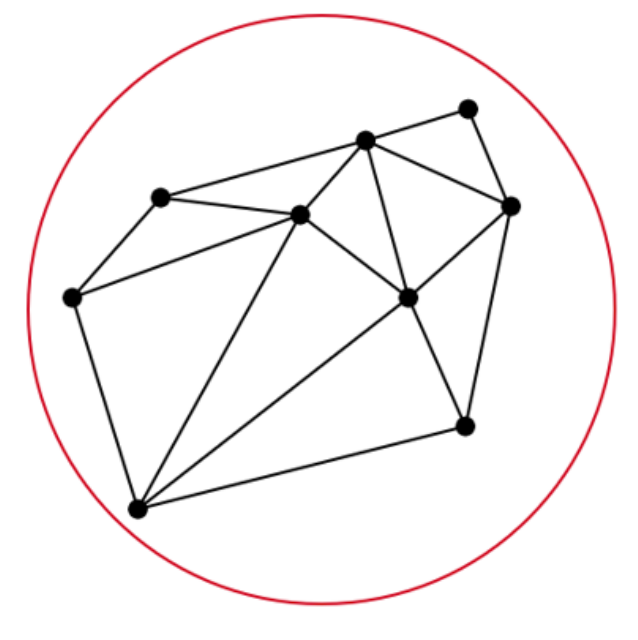}}
  \end{tabular}
    \caption{\textbf{Voronoi Diagram and Delaunay Triangulation.} A circular region is depicted with backbone node plotted as black points. The corresponding Voronoi diagram to these points is shown in (a). In (c) the Delaunay triangulation of the points is shown. The duality between these objects is shown in (b). Points in neighbouring Voronoi regions are connected by an edge in the triangulation.}
    \label{fig:Vor+Del}
\end{figure*}

Let us also now introduce the following two notions which describe notions of the centre of a region. We may define the centroid $\bar{\vec{r}}$ of a region with uniform density $\mathcal{R}\subset \mathbb{R}^2~$\footnote{Equivalent notions should hold in higher dimensions, but since we are interested in ground-based quantum networks, we restrict ourselves to the 2-dimensional case.} as:
\begin{equation}
\bar{\vec{r}}\defeq \underset{\vec{r_0}}{\argmin} \int_{\mathcal{R}} ~|\vec{r}-\vec{r}_0|^2 ~d\vec{r} .
\label{eq:centroid}
\end{equation}
The special case of a Voronoi tessellation where all the corresponding points are placed at the centroids of their Voronoi cells is called a \textit{centroidal Voronoi tessellation} (CVT). In the limit of an infinite two-dimensional region, the energy-minimizing CVT has every Voronoi region a regular hexagon \cite{Gersho_Conj,Hexagon_theorem} and the backbone points arranged on a triangular lattice. We can also define the capacity maximizing point for the same region and a capacity $\mathcal{C}$ as:
\begin{equation}
\vec{r}_{*} \defeq \underset{\vec{r}_b}{\argmax}\int_{\mathcal{R}} \mathcal{C}\big(|\vec{r}-\vec{r}_{b}|\big) ~d\vec{r}.
\label{eq:cap_maximizing_point}
\end{equation}
Using these definitions we are now able to give the following result which describes the necessary structure of an optimal backbone network.

\begin{proposition}
A necessary condition for a given backbone network $\mathcal{N}_b$ with a positive integer $k$ number of backbone nodes with positions $\{\vec{r}_{b_i}\}$ and associated regions $\{\mathcal{R}_i\}$ to be an optimal backbone network as defined in \cref{def:opt_bb} is that the regions are the Voronoi regions corresponding to the backbone nodes and simultaneously the backbone nodes are located at the capacity maximizing point of each region.
\label{prop:backbone}
\end{proposition}

\noindent For the proof, we follow a similar method to that presented in Refs.~\cite{Du1999VoroniReview,Du2010VoronoiReview} for demonstrating that centroidal Voronoi tessellations minimize the mean squared distance from a generating point. We proceed in two steps. Firstly, given any tessellation, we show that optimality requires that the points be located at the capacity maximizing point. Then we show that the tessellation must be a Voronoi tessellation.

\begin{proof}
 Consider any set of $k$ regions $\{\mathcal{R}_i\}_{i=1}^{i=k}$ which tessellate $\mathcal{R} \subset \mathbb{R}^2$ and corresponding $k$ backbone nodes $\{\vec{r}_{b_i}\in\mathcal{R}\}_{i=1}^{k}$. Then we are maximizing the expression:

\begin{equation}
\mathcal{F}(\{\vec{r}_{b_i}\}_{i=1}^k,\{\mathcal{R}_i\}_{i=1}^{k})=\sum_{i=1}^k \int_{\mathcal{R}_i} \mathcal{C} \big(|\vec{r}-\vec{r}_{b_i}|\big) ~d\vec{r}.
\end{equation}
Each term in the sum may independently be maximized by finding the optimal position for $b_i$. Clearly this is given by:

\begin{equation}
\vec{r}_{b_i}=\underset{\vec{r}_b}{\argmax} \int_{\mathcal{R}_i} \mathcal{C}\big(|\vec{r}-\vec{r}_{b}|\big) ~d\vec{r},
\end{equation}
which is the definition of the capacity maximizing point. It remains to show that the tessellation should be a Voronoi tessellation. Clearly the functional is maximized if we assign each point to the backbone node with which it has the maximum capacity. Each of the corresponding regions can therefore be given by:
\begin{equation}
\mathcal{R}_i=\big\{\vec{r} \in \mathcal{R} \ \big| \ \mathcal{C}(|\vec{r}-\vec{r_{b_i}}|)\geq \mathcal{C}(|\vec{r}-\vec{r_{b_j}}|) \ \forall j\neq i\big\},
\end{equation}
then since $\mathcal{C}$ is strictly decreasing, this is equivalent to:
\begin{equation}
\mathcal{R}_i=\big\{\vec{r}\in \mathcal{R}  \ \big| |\vec{r}-\vec{r_{b_i}}| \leq |\vec{r}-\vec{r_{b_j}}|  \ \forall j\neq i \ \big\},
\end{equation}
which is the definition of a Voronoi tessellation.
\end{proof}

We will therefore be concerned with constructing backbone networks whose nodal positions form a capacity-maximizing Voronoi tessellation. It is important to note that the preceding proposition extends to the weaker case in which $\mathcal{C}$ is merely non-increasing. In this case, the optimal tessellation need not be unique, since points lying in regions where $\mathcal{C}$ is constant may be assigned to different backbone nodes without changing the value of the objective. Nevertheless, there always exists an optimal tessellation that is Voronoi.

This observation applies in particular to the thermal loss channels. For these channels, $\mathcal{C}$ is strictly decreasing up to a distance $r'$, beyond which it is identically zero. The simple case is when $r'>d$, where $d$ is the diameter of the region $\mathcal{R}$. Since all distances within $\mathcal{R}$ are then less than $r'$, $\mathcal{C}$ is strictly decreasing over the entire range of relevant distances, and the conditions of~\cref{prop:backbone} are satisfied. 
More generally, even when the constant-capacity regime is reached within $\mathcal{R}$, the generalized result for non-increasing capacity functions guarantees the existence of a capacity-maximizing Voronoi tessellation.

\begin{figure*}[tb]
\centering
\begin{tabular}{cc}
\subfloat[Optimization Procedure]{\includegraphics[width=0.5\linewidth]{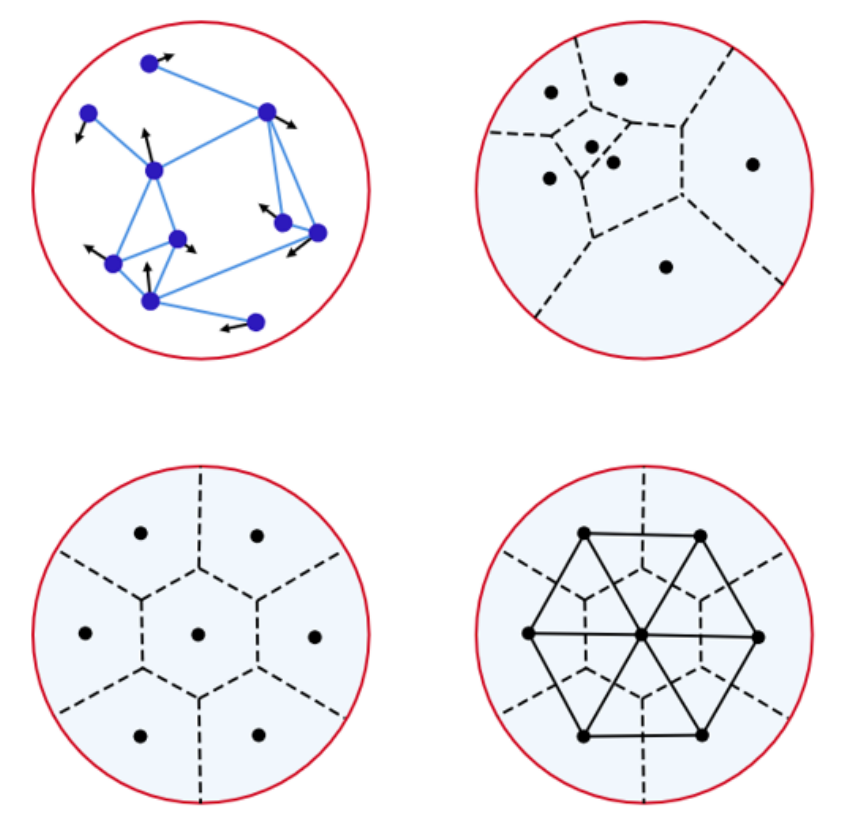}} & 
\subfloat[Example Network Instance]{\includegraphics[width=0.5\linewidth]{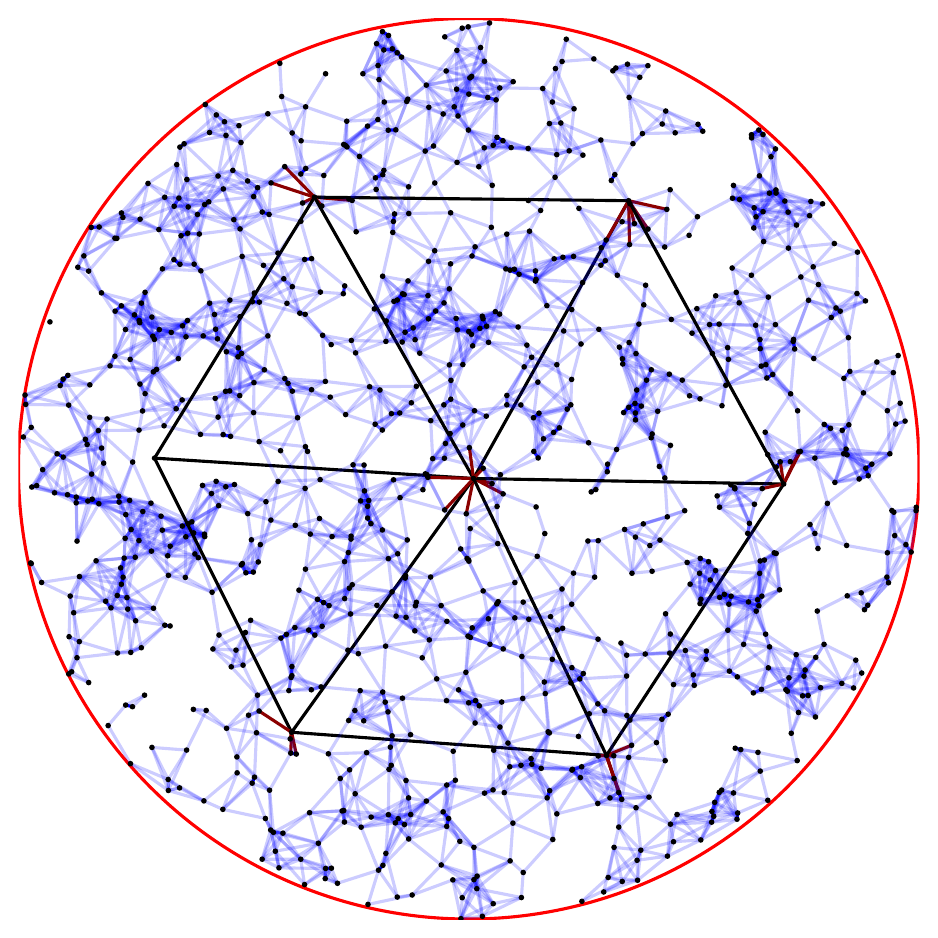}}
\end{tabular}
\caption{
\textbf{Optimization of fiber backbone for mobile user free-space quantum network}. The left panel (clockwise from top left) shows the optimization procedure for the fiber backbone. The first image shows a snapshot of the network of mobile users connecting via free-space edges; in the second image, the random positions of the user nodes are replaced with a constant density $\rho$ represented by the light blue background. A fixed number of backbone nodes (shown in black) are generated with random positions in the network region $\mathcal{R}_{\mathcal{N}}$ and the corresponding Voronoi diagram is calculated. In image three the positions of the backbone nodes are relaxed to optimal positions via Lloyd's algorithm, resulting in a centroidal Voronoi tessellation; in the final image, the backbone nodes are connected via fiber edges into a desired topology (in this case a Delaunay triangulation). An example network is shown in the right panel. Blue lines indicate free-space channels; black lines fiber backbone channels and red lines free-space channels connecting to the backbone. The plotted network shown represents the state of the network at a given time. It is comprised of a Waxman network of radius 1 km with 1000 mobile end-user nodes communicating via free-space channels (depicted in blue). The backbone consists of 7 nodes connected via fiber-based channels (depicted in black). The mobile users connect to the backbone via the same free-space channels as with each other (highlighted in red) and follow the same Waxman connectivity rule.}
\label{fig:lloyd opt}
\end{figure*}

\subsection{Computing Centroidal Voronoi Tessellations}
 Assuming that the average density of user nodes is constant across the network region $\mathcal{R}_{\mathcal{N}}$, we may optimize the position of a given number of backbone nodes using Lloyd's algorithm \cite{Lloyd's_algo}, which relaxes the positions of the backbone nodes towards an optimal configuration. The optimization procedure is illustrated in \cref{fig:lloyd opt}. Firstly, the positions of the random backbone nodes are randomly generated in the region $\mathcal{R}_{\mathcal{N}}$ and the corresponding Voronoi diagram is computed. The capacity maximizing point of each Voronoi region is computed and the corresponding backbone node is moved to the position of that point. The process is repeated until the positions of the backbone nodes converge to optimal positions and the resulting configuration is a capacity maximizing Voronoi 
tessellation. 

Nonetheless we have a problem. The capacity maximizing point is not easily calculated and involves maximizing the integral in \cref{eq:cap_maximizing_point}. This is not an issue in computing the standard centroidal Voronoi tessellation since the centroid has a nice alternative computation:
\begin{equation}
\bar{\vec{r}}=\frac{\int_{\mathcal{R}} ~\vec{r}~d\vec{r}}{\int_{\mathcal{R}} ~ d\vec{r}},
\end{equation}
which can be found by numerical integration far more quickly than by solving the optimization problem implied by the definition in \cref{eq:centroid}. Fortunately we have the following proposition, which enables us to approximate a capacity maximizing Voronoi tessellation with a centroidal Voronoi tessellation.

\begin{proposition}
Given a convex region $\mathcal{R}\subset\mathbb{R}^2$ with two lines of mirror symmetry, then there exists a unique point $\bar{\vec{r}} \in \mathcal{R}$ which maximizes (minimizes) 
\begin{equation}
\int_{\mathcal{R}}f\big(|\vec{r}-\bar{\vec{r}}|\big) d\vec{r},
\end{equation}
for any strictly decreasing (increasing) function $f$.
\label{prop:mirror_sym}
\end{proposition}

\begin{proof}
Let us briefly sketch our method. We shall show that a maximizing (minimizing) point must lie on a line of symmetry. Then two unique lines of mirror symmetry will  intersect at a single point which must lie within the convex region. This intersection point will then maximize (minimize) the integral. 

Let us work with the case of maximizing a strictly decreasing function $f$. Consider a line of symmetry $l_1$ of the region $\mathcal{R}$ and a point $\vec{r}_0 \notin l_1$. Then define a new point $\vec{r}_l \in l_1$ to be the nearest point on $l_1$ to $\vec{r}_0$. We define the following Voronoi region $\mathcal{R}_0=\{\vec{r}~\big|~|\vec{r}-\vec{r}_0|<|\vec{r}-\vec{r}_l|\}$ and the line $\tilde{l}$ as the line where $\tilde{l}=\{\vec{r}~\big|~|\vec{r}-\vec{r}_0|=|\vec{r}-\vec{r}_l|\}$ which is the boundary of $\mathcal{R}_0$ in the interior of $\mathcal{R}$ and which is necessarily parallel to $l_1$. Finally we define:
\begin{equation}
I_1=\int_{\mathcal{R}} f\big(|\vec{r}-\vec{r_0}|\big) ~ d\vec{r},
\end{equation}
\begin{equation}
I_2=\int_{\mathcal{R}} f\big(|\vec{r}-\vec{r_l}|\big) ~ d\vec{r},
\end{equation}
and we wish to show that $I_2>I_1$.

Consider reflecting any point in $\vec{r}\in\mathcal{R}_0$ about the line of mirror symmetry  $l_1$ and connecting the two points via a line. By convexity of the region all points on this line lie in $\mathcal{R}$. If the point $\vec{r}$ is instead reflected about $\tilde{l}$ then this point clearly still lies on this line, since $\tilde{l}$  is parallel to $l_1$. It is therefore possible to reflect $\mathcal{R}_0$ about $\tilde{l}$ defining the new region $\tilde{\mathcal{R}}_0 \subset \mathcal{R}$. We can partition $\mathcal{R}$ into the three regions $\mathcal{R}_0,\tilde{\mathcal{R}}_0$ and $\mathcal{R}/(\mathcal{R}_0\cup\tilde{\mathcal{R}}_0)$ for the purposes of calculating $I_1$ and $I_2$. Since reflection about $\tilde l$ is an isometry pairing $\vec r_0\leftrightarrow\vec r_l$, it preserves distances pointwise and therefore any function of them, in particular $f$. It is clear that :

\begin{equation}
 \int_{\mathcal{R}_0} f\big(|\vec{r}-\vec{r_0}|\big) ~ d\vec{r} =\int_{\tilde{\mathcal{R}}_0} f\big(|\vec{r}-\vec{r_l}|\big) ~ d\vec{r},
\end{equation}
and
\begin{equation}
 \int_{\tilde{\mathcal{R}}_0} f\big(|\vec{r}-\vec{r_0}|\big) ~ d\vec{r} =\int_{\mathcal{R}_0} f\big(|\vec{r}-\vec{r_l}|\big) ~ d\vec{r},
\end{equation}
so that
\begin{align}
I_2-I_1= \int_{\mathcal{R}/(\mathcal{R}_0 \cup \tilde{\mathcal{R}}_0)} f\big(|\vec{r}-\vec{r_l}|\big) - f\big(|\vec{r}-\vec{r_0}|\big) ~ d\vec{r} .
\end{align}
  Since $\vec{r}_0\notin l_1$, we have $\tilde l\neq l_1$; by convexity of $\mathcal{R}$, the reflected region $\tilde{\mathcal{R}}_0$ is a strict subset of the half-plane on the far side of $\tilde l$, so $\mathcal{R}/(\mathcal{R}_0\cup\tilde{\mathcal{R}}_0)$ has strictly positive measure. On this set $f\big(|\vec{r}-\vec{r_l}|\big) > f\big(|\vec{r}-\vec{r_0}|\big)$ since $f$ is strictly decreasing, and thus $I_2-I_1>0$. Thus we always get an improved value if we map to the nearest point on the line of symmetry.
\end{proof}

The centroid is the point that minimizes the integral of the squared distance over a region. Therefore, when the capacity bound is strictly decreasing with distance and the region possesses sufficient symmetry, the capacity-maximizing point coincides with its centroid. Since the original region $\mathcal{R}$ is a circle and hence convex, its Voronoi cells are also convex. Furthermore, the hexagonal unit cell, which characterizes the interior of the expected optimal tessellation, possesses multiple axes of mirror symmetry; consequently, its capacity-maximizing point coincides with its centroid. Given that the interior of the tessellation is expected to approach a hexagonal tiling, appreciable differences between the centroid and the capacity-maximizing point should arise primarily near the boundary of $\mathcal{R}$, where this symmetry is broken. Motivated by this observation, and to reduce the computational complexity of the optimization, we approximate the optimal Voronoi tessellation by a centroidal Voronoi tessellation.

It is important to emphasize that this reduction of the optimization problem is heuristic and approximate for two main reasons. First, it relies on the presence of sufficient symmetry within the Voronoi cells; although such symmetry is typical of cells in the interior of a nearly hexagonal tessellation, it is not guaranteed in the general case, particularly near the boundary. Second, the argument assumes that the capacity bound is strictly decreasing with distance throughout the regions under consideration, which need not always hold, for example when evaluating the capacity bound for long-distance thermal-loss channels. Despite these limitations, the resulting approximation nevertheless yields the desired improvements, as demonstrated in the next section on the optimized backbone results.

\subsection{Backbone Edges and Interconnections}
Our optimization has thus far focused solely on the positions of the backbone (or router) nodes. Since we expect that such nodes will be a far more expensive resource than the links and the capacity of the backbone links is so much greater than the free-space links any reasonable topology should suffice. Nonetheless we now highlight a specific desirable topology, which we shall make use of in our results. In particular, we connect the backbone according to their Delaunay triangulation. Such triangulations have a number of properties \cite{musin2003propertiesDeltri} which make them appealing potential topologies for quantum network backbones. In particular, there is linear scaling in the number of backbone network edges $|E_b|$:
\begin{equation}
|E_b| = 3n_b-3-n_h,
\end{equation}
where $n_b$ is the number of backbone vertices and $n_h$ is the number of such vertices which lie on  the convex hull of the backbone. Additionally, every node is connected to its nearest neighbours, maximizing point-to-point performance on the backbone and the  stretch factor of the graph is known to be $<2$\ \cite{Xia2013StretchFactorDelaunay} \footnote{Given any two nodes $\boldsymbol{x},\boldsymbol{y}\in V$, then there exists a path $\omega$ such that $\sum_{(i,j)\in \omega} |\vec{r}_{\boldsymbol{i}}-\vec{r}_{\boldsymbol{j}}|\ < 2|\vec{r}_{\boldsymbol{x}}-\vec{r}_{\boldsymbol{y}}|$}  which guarantees that communication across the backbone should not use an excessively high proportion of the backbone edges.

Finally, for backbone networks added to Waxman networks, it is important to specify how the base network is connected to the backbone. We assume that this follows the same Waxman connectivity rule as the rest of the network. Thus, all edges comprised of the same type of channels follow the same connection rules, which aids analysis of the network.

\begin{figure*}[t]
    \centering
    \setlength{\tabcolsep}{0pt}
    \begin{tabular}{cc}
    \subfloat[$R=0.2$km]{\includegraphics[width=0.42\linewidth,trim=0pt 27pt 0pt 3pt,clip]{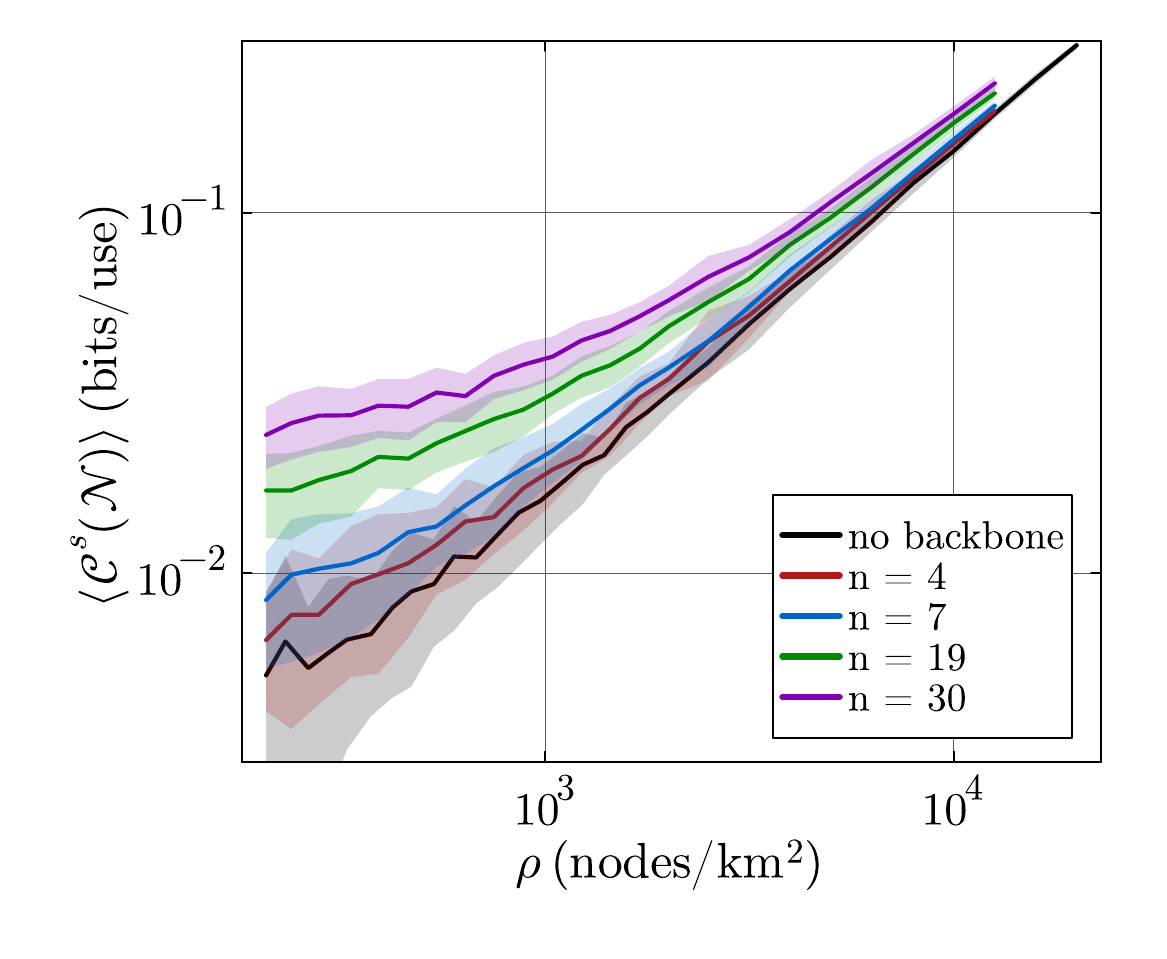}\label{fig:single_path_mobile_bb_res}} &
    \subfloat[$R=0.2$km]{\includegraphics[width=0.42\linewidth,trim=0pt 27pt 0pt 3pt,clip]{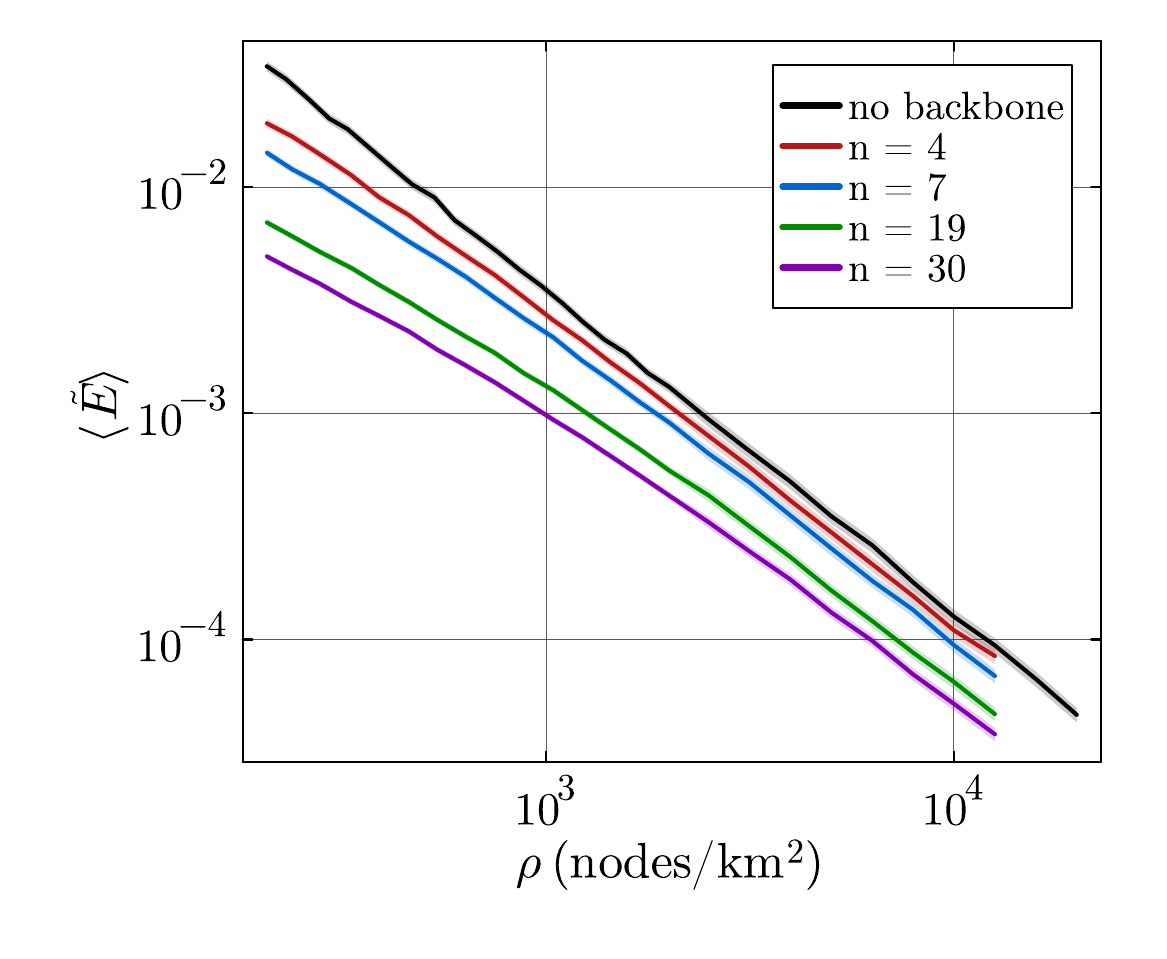}\label{fig:mobile_bb_cons_02}} \\[-6pt]
    \subfloat[$R=0.5$km]{\includegraphics[width=0.42\linewidth,trim=0pt 27pt 0pt 3pt,clip]{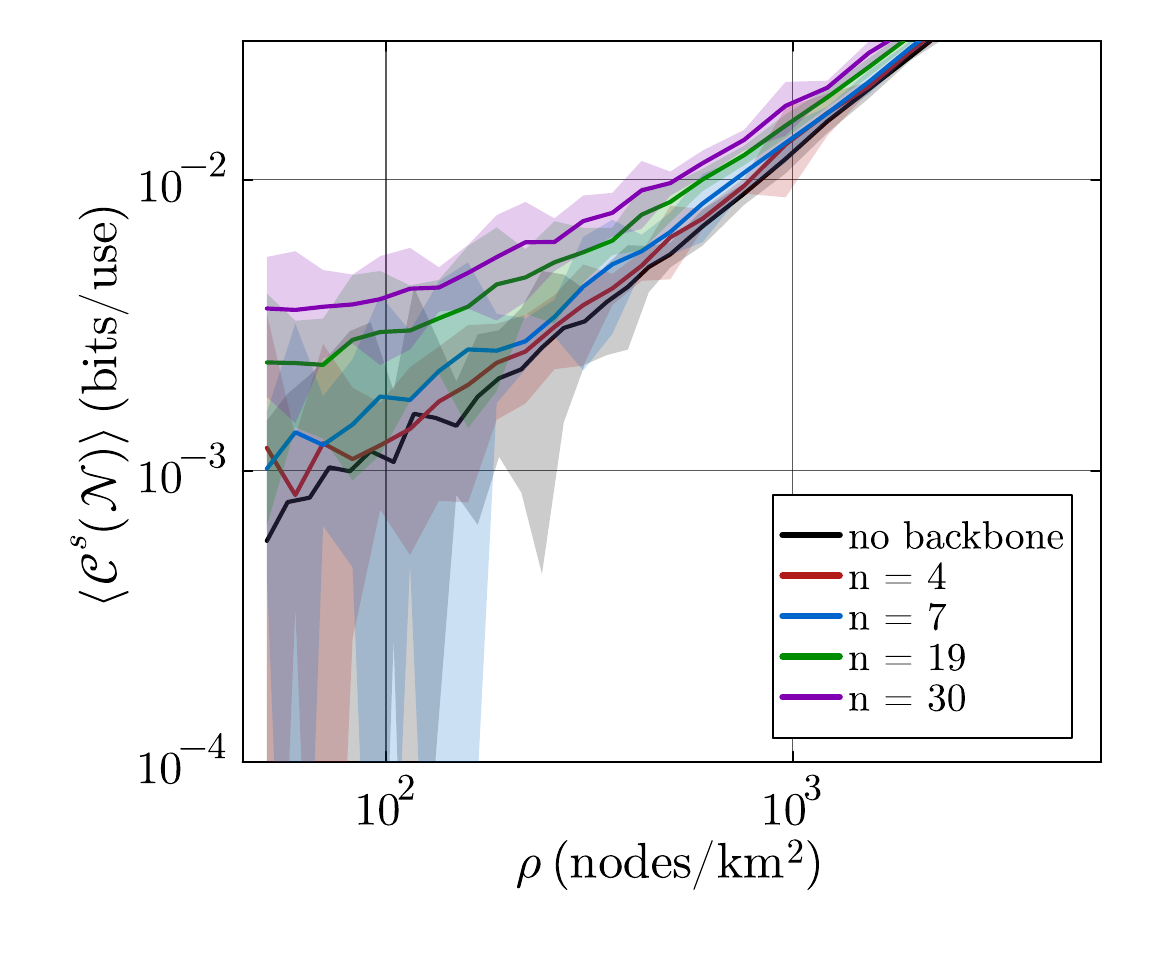}\label{fig:mobile_bb_cap_05}} &
    \subfloat[$R=0.5$km]{\includegraphics[width=0.42\linewidth,trim=0pt 27pt 0pt 3pt,clip]{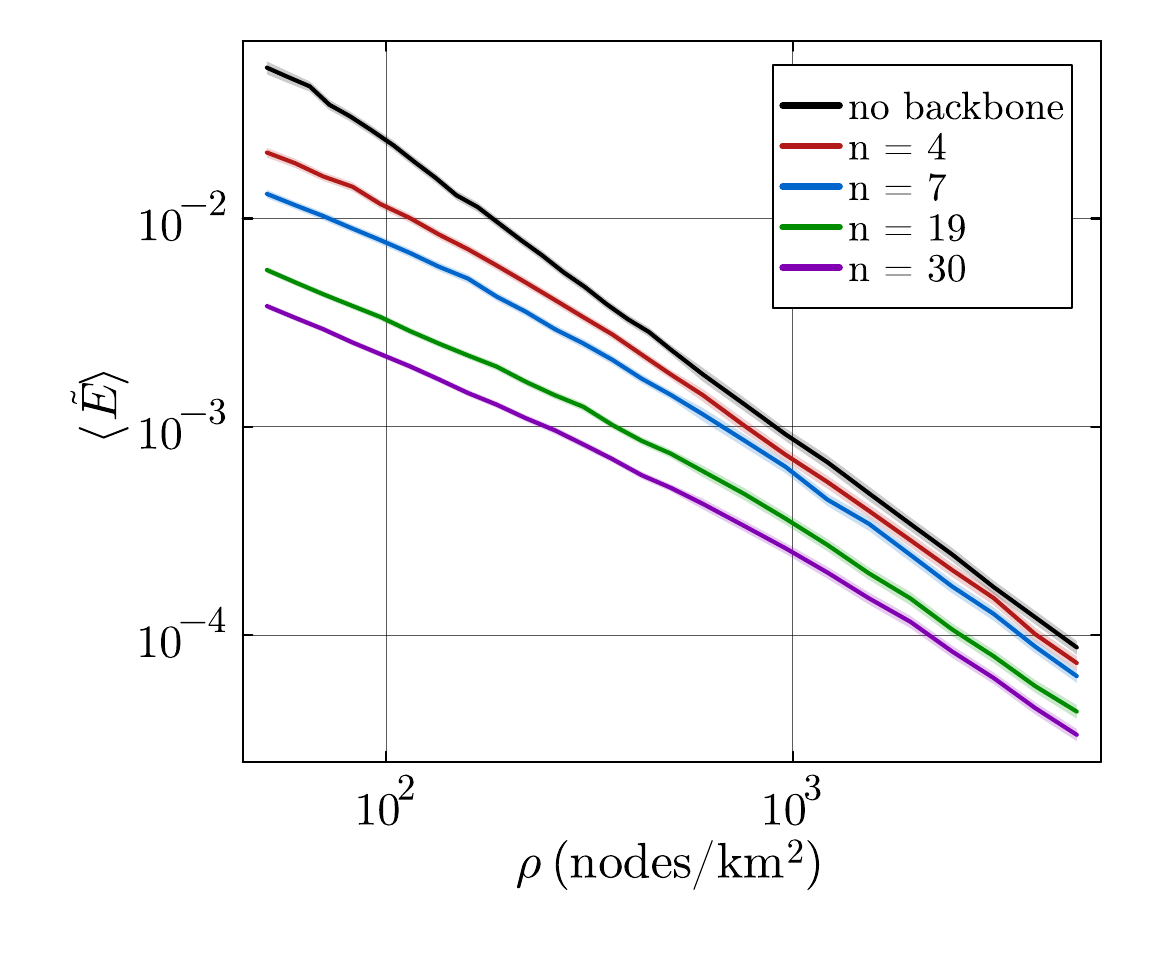}\label{fig:mobile_bb_cons_05}} \\[-6pt]
    \subfloat[$R=1.0$km]{\includegraphics[width=0.42\linewidth,trim=0pt 27pt 0pt 3pt,clip]{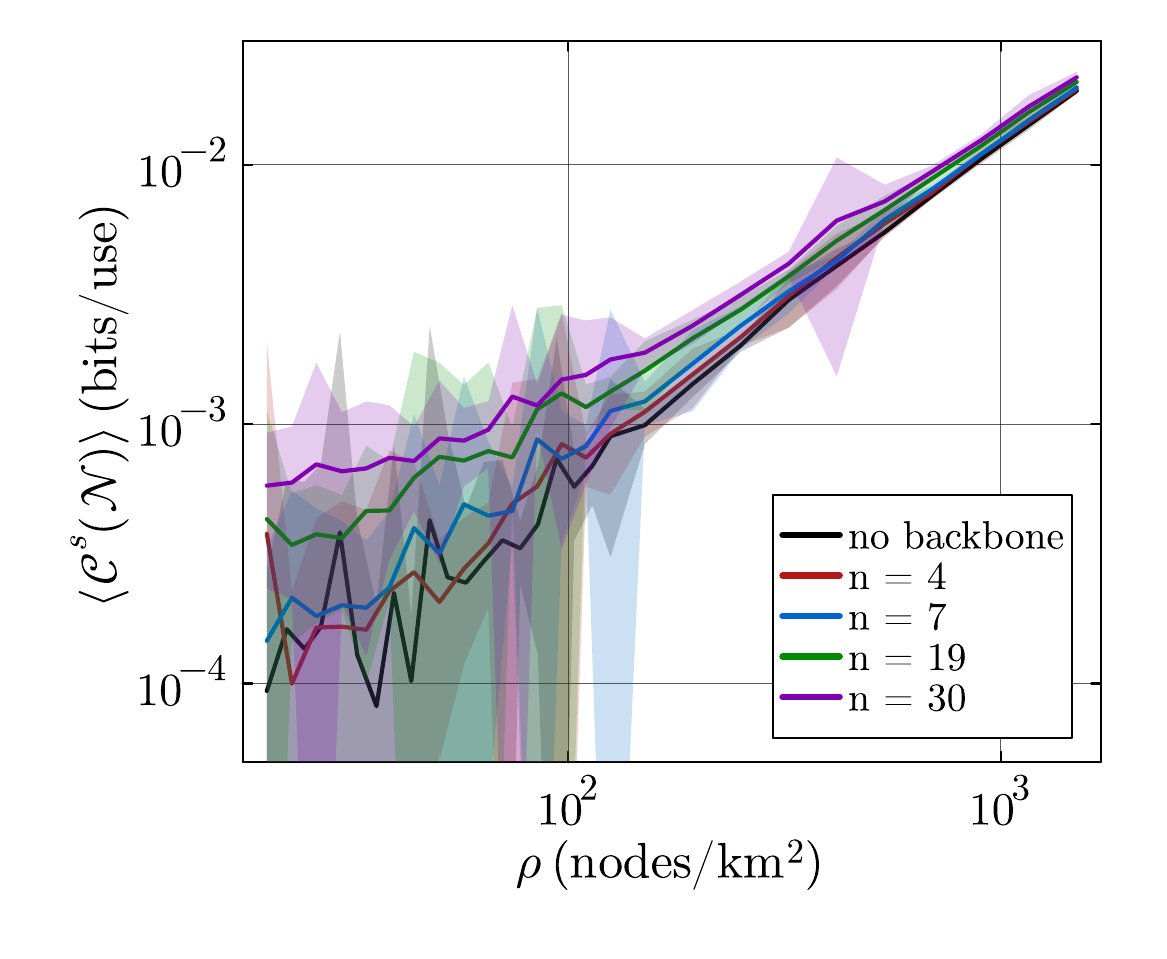}\label{fig:mobile_bb_cap_10}} &
    \subfloat[$R=1.0$km]{\includegraphics[width=0.42\linewidth,trim=0pt 27pt 0pt 3pt,clip]{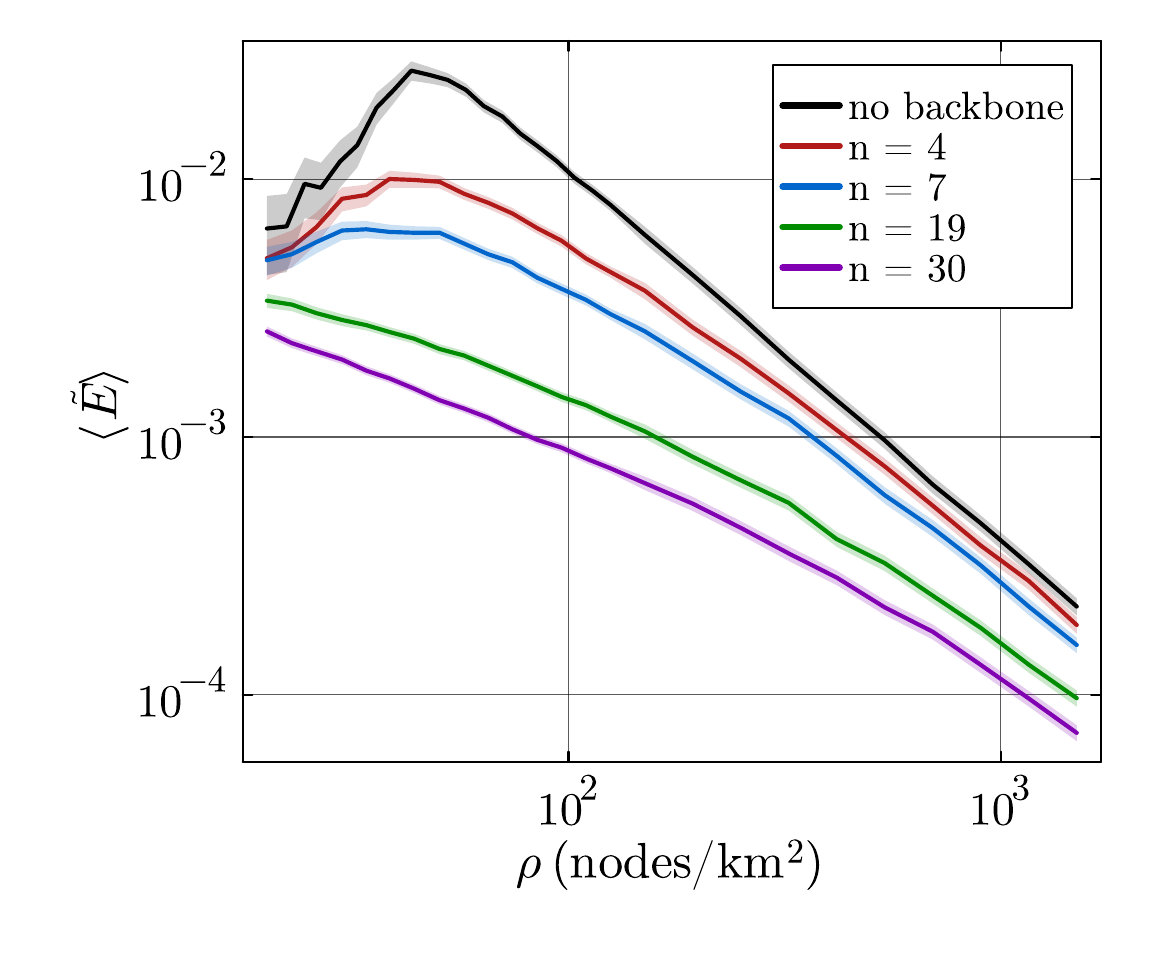}\label{fig:mobile_bb_cons_10}}
    \end{tabular}
    \caption{\textbf{Single-path mobile-user results with fiber backbone.} Left column: single-path network capacity bound vs.\ nodal density. Right column: edge consumption vs.\ nodal density. (a,b) $R=0.2$km. (c,d) $R=0.5$km. (e,f) $R=1.0$km. Shaded bands show $\pm1$ standard error of the mean, each colour corresponds to a different number of backbone nodes $n$, and each data point averages 100 network instances and 50 end-user pairs per network. See text for more details.}
    \label{fig:single_path_mobile_bb_res_composite}
\end{figure*}

\begin{figure*}[t]
    \centering
    \setlength{\tabcolsep}{0pt}
    \begin{tabular}{cc}
    \subfloat[$R=10$km]{\includegraphics[width=0.42\linewidth,trim=0pt 27pt 0pt 3pt,clip]{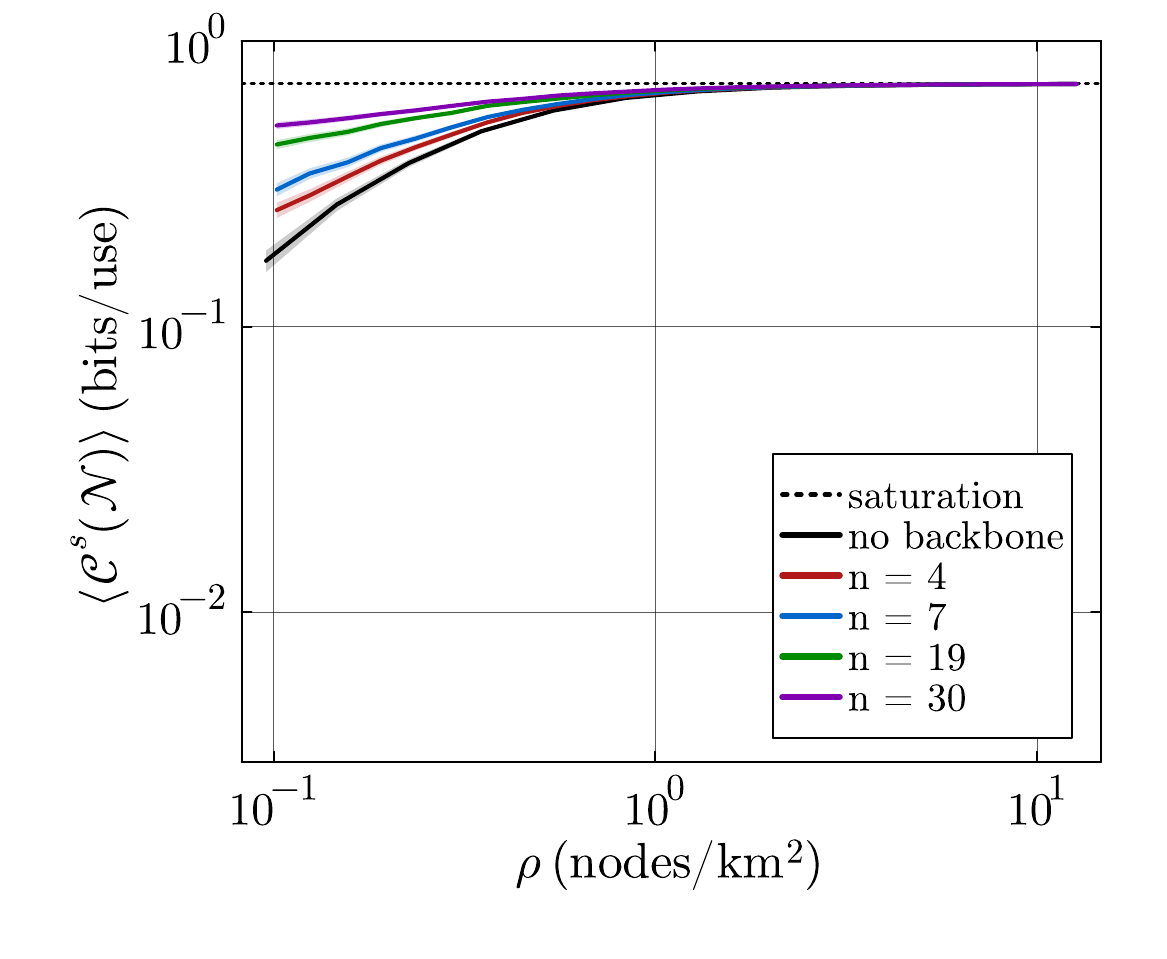}\label{fig:single_path_fixed_bb_res}} &
    \subfloat[$R=10$km]{\includegraphics[width=0.42\linewidth,trim=0pt 27pt 0pt 3pt,clip]{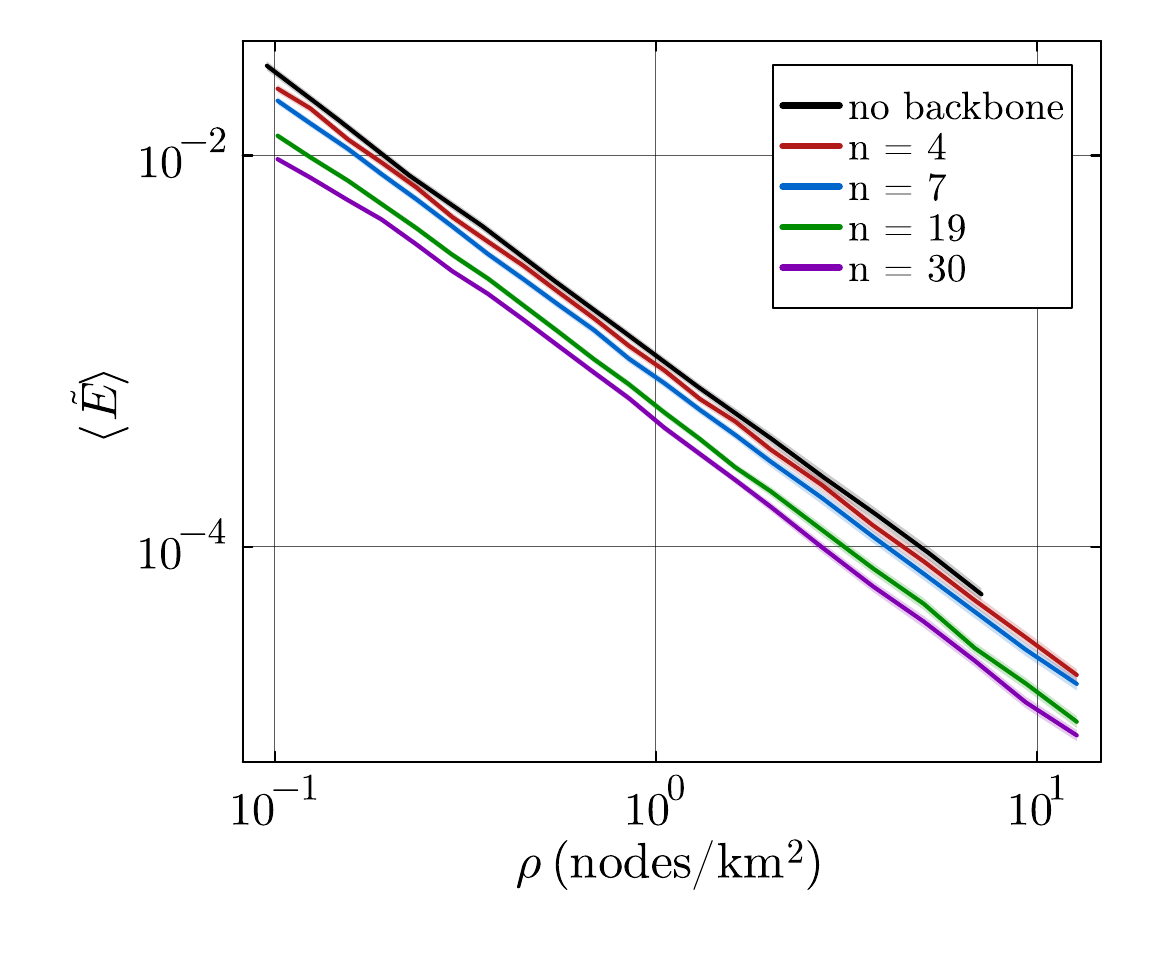}\label{fig:fixed_bb_cons_10}} \\[-6pt]
    \subfloat[$R=25$km]{\includegraphics[width=0.42\linewidth,trim=0pt 27pt 0pt 3pt,clip]{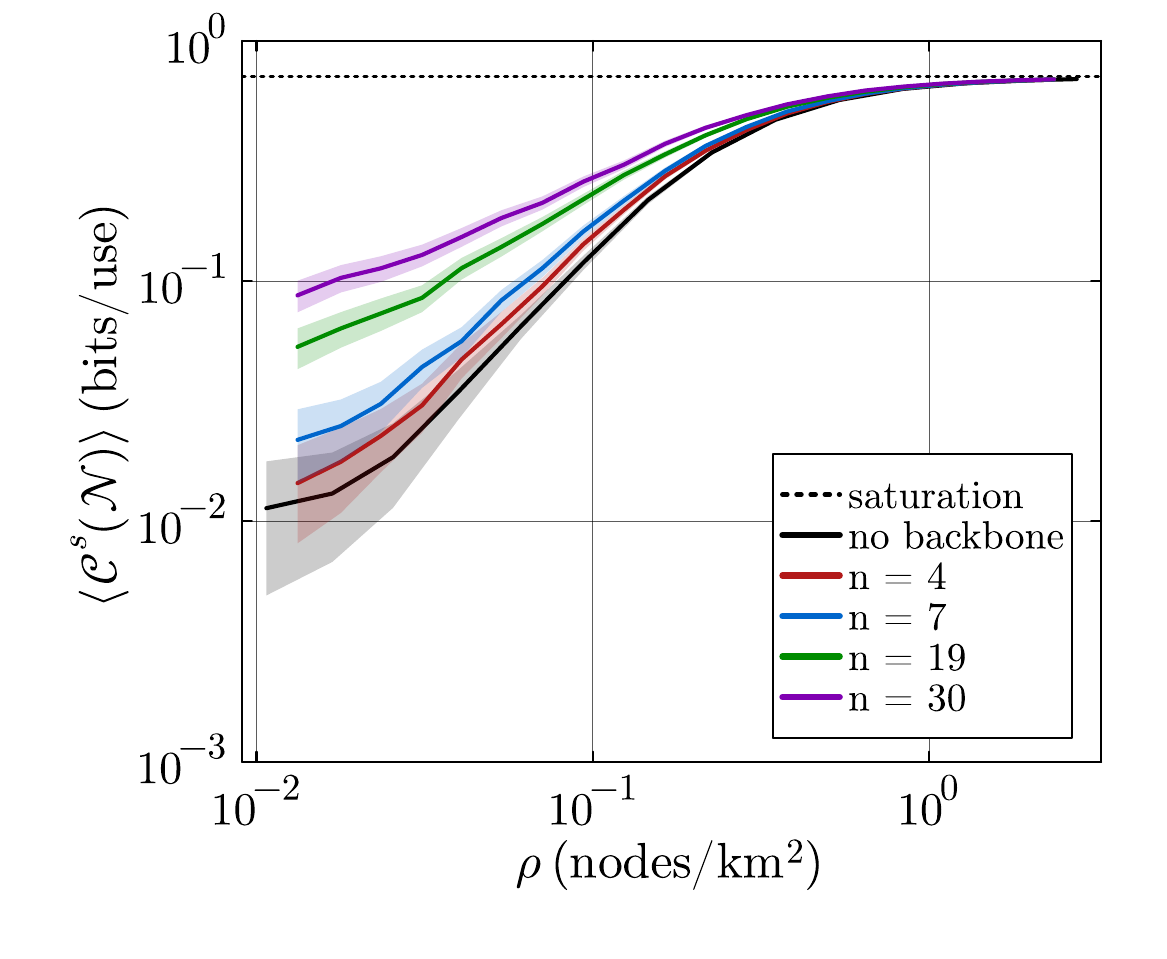}\label{fig:fixed_bb_cap_25}} &
    \subfloat[$R=25$km]{\includegraphics[width=0.42\linewidth,trim=0pt 27pt 0pt 3pt,clip]{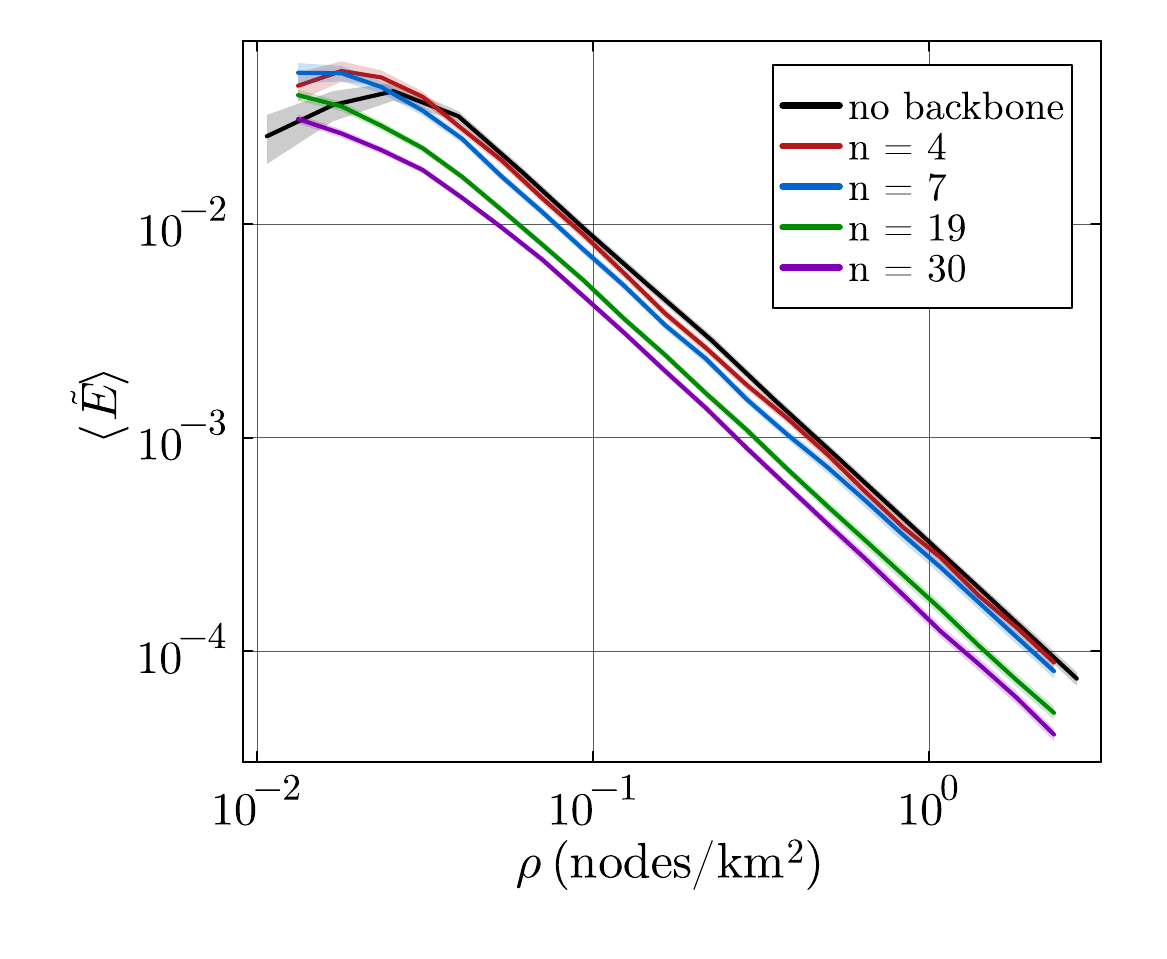}\label{fig:fixed_bb_cons_25}} \\[-6pt]
    \subfloat[$R=50$km]{\includegraphics[width=0.42\linewidth,trim=0pt 27pt 0pt 3pt,clip]{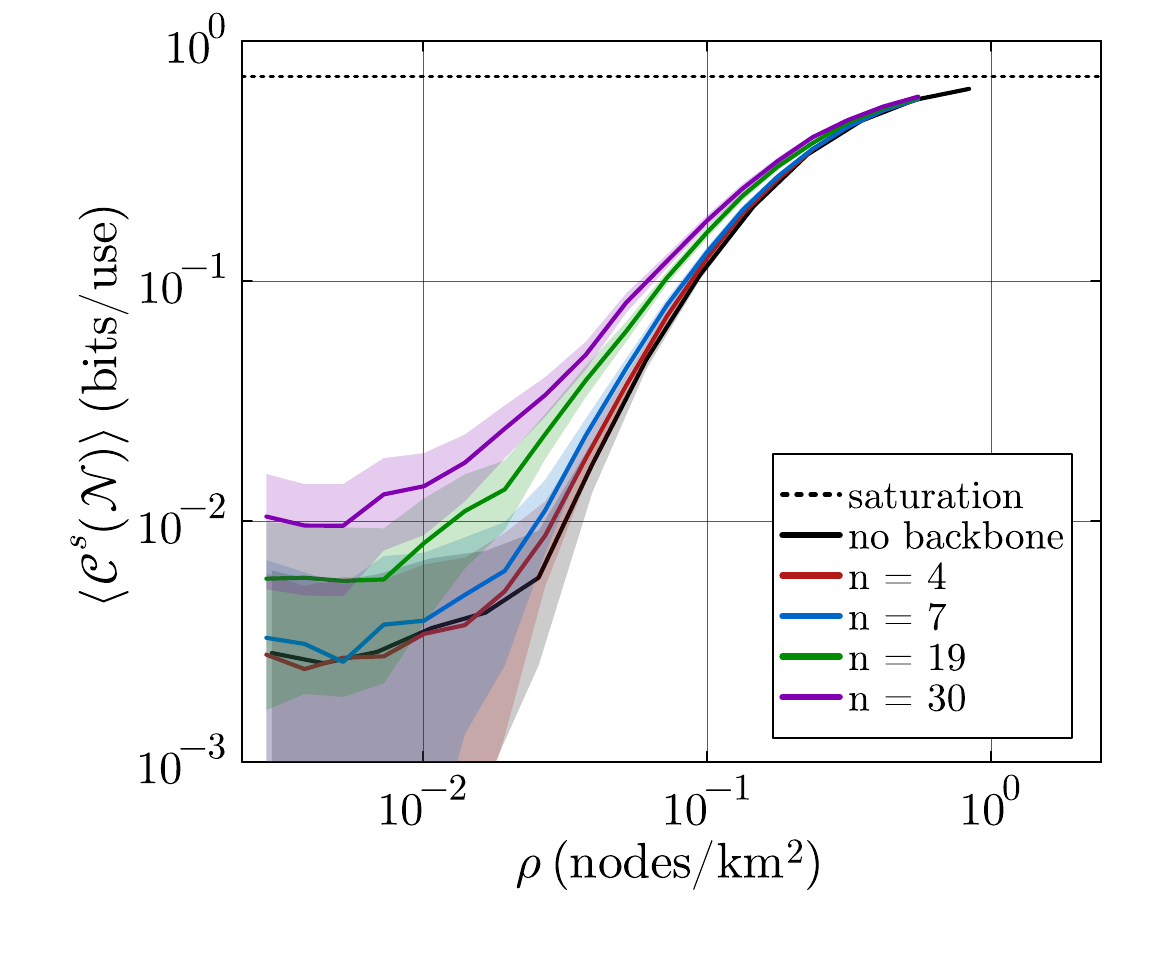}\label{fig:fixed_bb_cap_50}} &
    \subfloat[$R=50$km]{\includegraphics[width=0.42\linewidth,trim=0pt 27pt 0pt 3pt,clip]{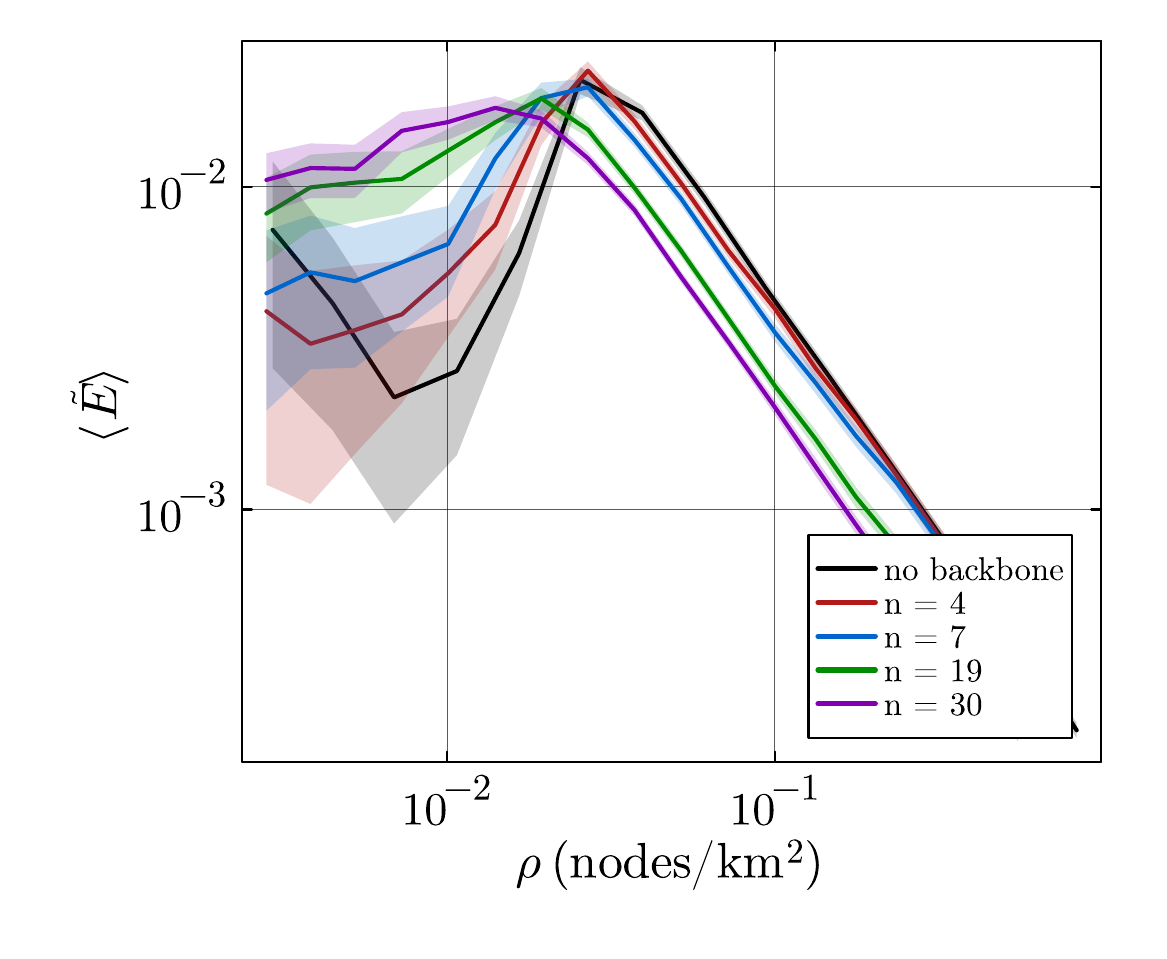}\label{fig:fixed_bb_cons_50}}
    \end{tabular}
    \caption{\textbf{Single-path fixed-user results with fiber backbone.} Left column: single-path network capacity bound vs.\ nodal density; the dotted line is the maximal fixed free-space 
    capacity, an upper bound on the single-path network capacity. Right column: edge consumption vs.\ nodal density. (a,b) $R=10$km. (c,d) $R=25$km. (e,f) $R=50$km. Shaded bands show $\pm1$ standard error of the mean, each data point averages 100 network instances and 50 end-user pairs per network, and each colour corresponds to a different number of backbone nodes. See text for more details.}
    \label{fig:single_path_fixed_bb_res_composite}
\end{figure*}

\section{optimized Backbone Results}
\label{sec:optimized bb results}
We now consider hybrid quantum networks comprised of an optimized fiber backbone and free-space base network. In particular, we show how to improve the performance of free-space Waxman networks, by allocating a backbone fiber infrastructure according to the heuristic approach specified in the previous section. \SP{As in~\cref{subsec:wax_net_baseline}, our analysis is based on thermal-loss channels, here employed for both free-space and fiber-based links (the former chosen as in Table~\ref{table:Setups} while the latter chosen with $\bar{n}=1$).}

For mobile user nodes, the single-path results are depicted in \cref{fig:single_path_mobile_bb_res,fig:mobile_bb_cap_05,fig:mobile_bb_cap_10,fig:mobile_bb_cons_02,fig:mobile_bb_cons_05,fig:mobile_bb_cons_10} while the multi-path results are shown in \cref{fig:mobile_bb_flood,fig:mobile_bb_flood_05,fig:mobile_bb_flood_10}. The results for fixed user nodes are shown in \cref{fig:single_path_fixed_bb_res,fig:fixed_bb_cap_25,fig:fixed_bb_cap_50,fig:fixed_bb_cons_10,fig:fixed_bb_cons_25,fig:fixed_bb_cons_50} and \cref{fig:fixed_bb_flood,fig:fixed_bb_flood_25,fig:fixed_bb_flood_50}. The equivalent baseline results presented in \cref{subsec:wax_net_baseline} are shown for comparison (`no backbone' lines in the figures). We remark on the following features of the results.  

For the single-path network capacity bounds, adding an increasing number of backbone nodes results in improved performance at low nodal density in both the fixed and the mobile cases. At higher densities, the capacity bounds begin to converge in the mobile case, but there is sequential improvement with each increased size of backbone over the full density range we consider. In the fixed case, the convergence is to the maximal single-path capacity so the backbone can only improve performance in the low density regime. 

The consumption is reduced with an increasing number of backbone nodes, causing an offset between the lines in the consumption plot. 
This is reminiscent of that seen in \cref{fig:baseline_wax_mobile_cons,fig:baseline_wax_fixed_cons} where, once a transition occurs, the edge consumption is well fitted by a power-law relation. 
The consumption reduction therefore seems to be explained by an effective reduction in the network size due to the shortcuts through the region enabled by the backbone.
 
The performance improvement for flooding is more limited. At lower densities, the performance is increased when the number of backbone nodes is increased. At higher densities, the network capacity bounds converge towards the neighbourhood capacity upper bound \SP{specified by~\cref{eq:flooding_upper_bound}} at which point the backbone provides no advantage. Additionally, we do not observe a substantial reduction in the density at which a phase transition occurs in any of the multi-path capacity results~\footnote{Whilst this may on the surface seem surprising, the addition of a backbone does not affect the connectivity phase transition in the underlying network. If the graph nodes are mostly isolated, then the addition of a limited number of backbone nodes, no matter how well interconnected by fiber links, will not substantially improve performance. Since connections to the backbone nodes are governed by the same probabilistic rule as in the base Waxman network, the inclusion of $n_b$ backbone nodes over a network region of area $A$ would heuristically shift the transition density from $\rho_0$ to
\[
\tilde{\rho}_0 = \rho_0 - \frac{n_b}{A}.
\]
When $n_b/A \ll \rho_0$, we have $\tilde{\rho}_0 \simeq \rho_0$, and the density at which the capacity transition occurs therefore remains effectively unchanged.}.

\begin{figure*}[h!]
    \centering
    \setlength{\tabcolsep}{0pt}
    \begin{tabular}{cc}
    \subfloat[mobile, $R=0.2$km]{\includegraphics[width=0.42\linewidth,trim=0pt 27pt 0pt 15pt,clip]{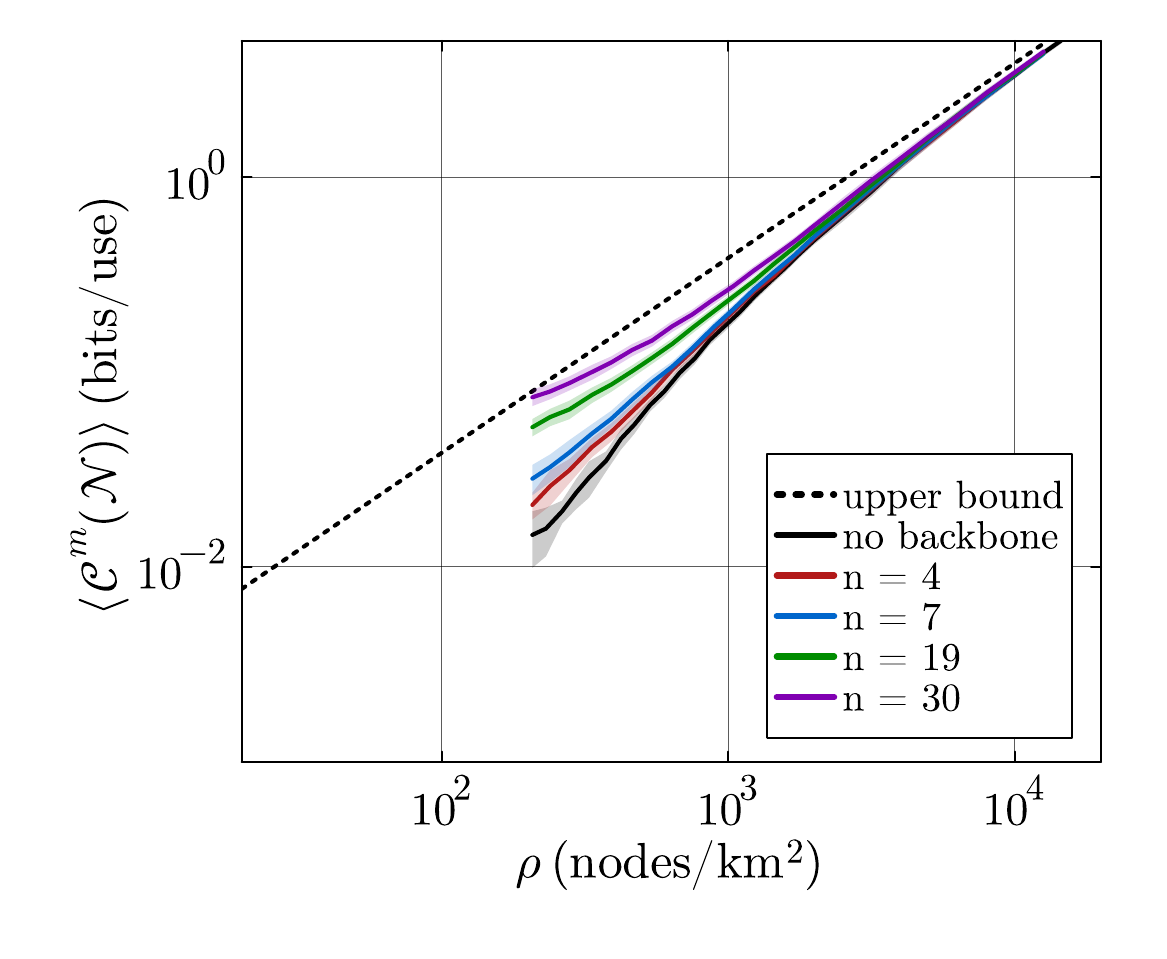}\label{fig:mobile_bb_flood}} &
    \subfloat[fixed, $R=10$km]{\includegraphics[width=0.42\linewidth,trim=0pt 27pt 0pt 3pt,clip]{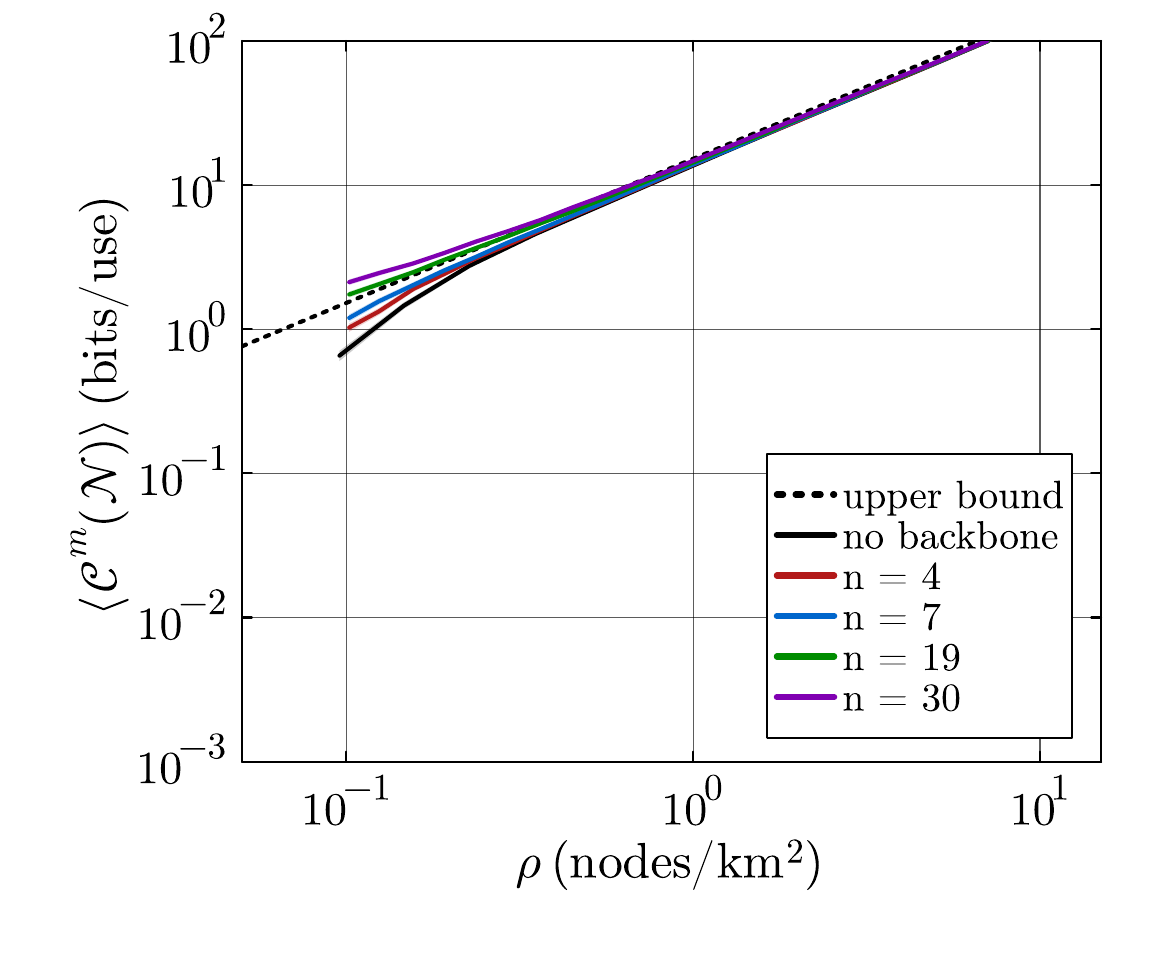}\label{fig:fixed_bb_flood}} \\[-6pt]
    \subfloat[mobile, $R=0.5$km]{\includegraphics[width=0.42\linewidth,trim=0pt 27pt 0pt 15pt,clip]{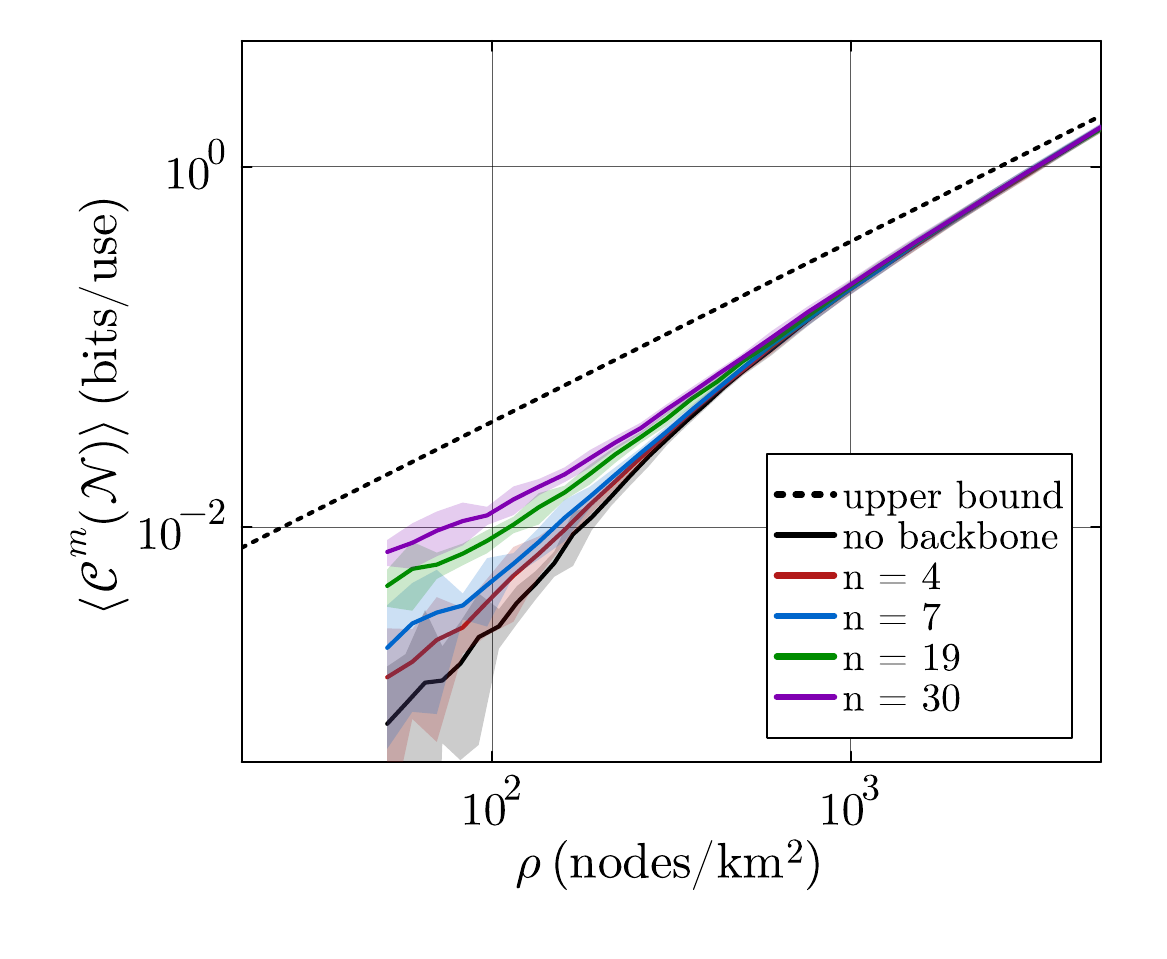}\label{fig:mobile_bb_flood_05}} &
    \subfloat[fixed, $R=25$km]{\includegraphics[width=0.42\linewidth,trim=0pt 27pt 0pt 3pt,clip]{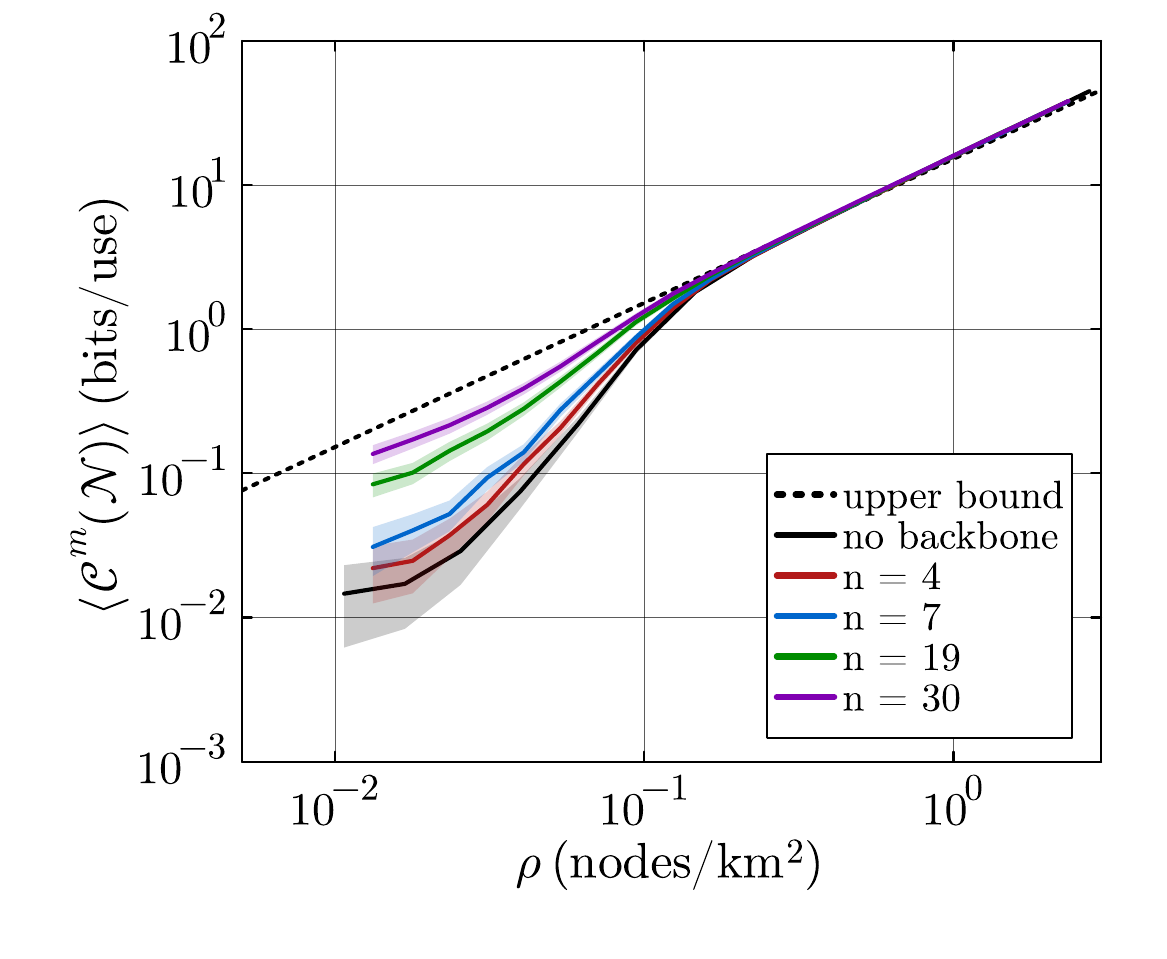}\label{fig:fixed_bb_flood_25}} \\[-6pt]
    \subfloat[mobile, $R=1.0$km]{\includegraphics[width=0.42\linewidth,trim=0pt 27pt 0pt 15pt,clip]{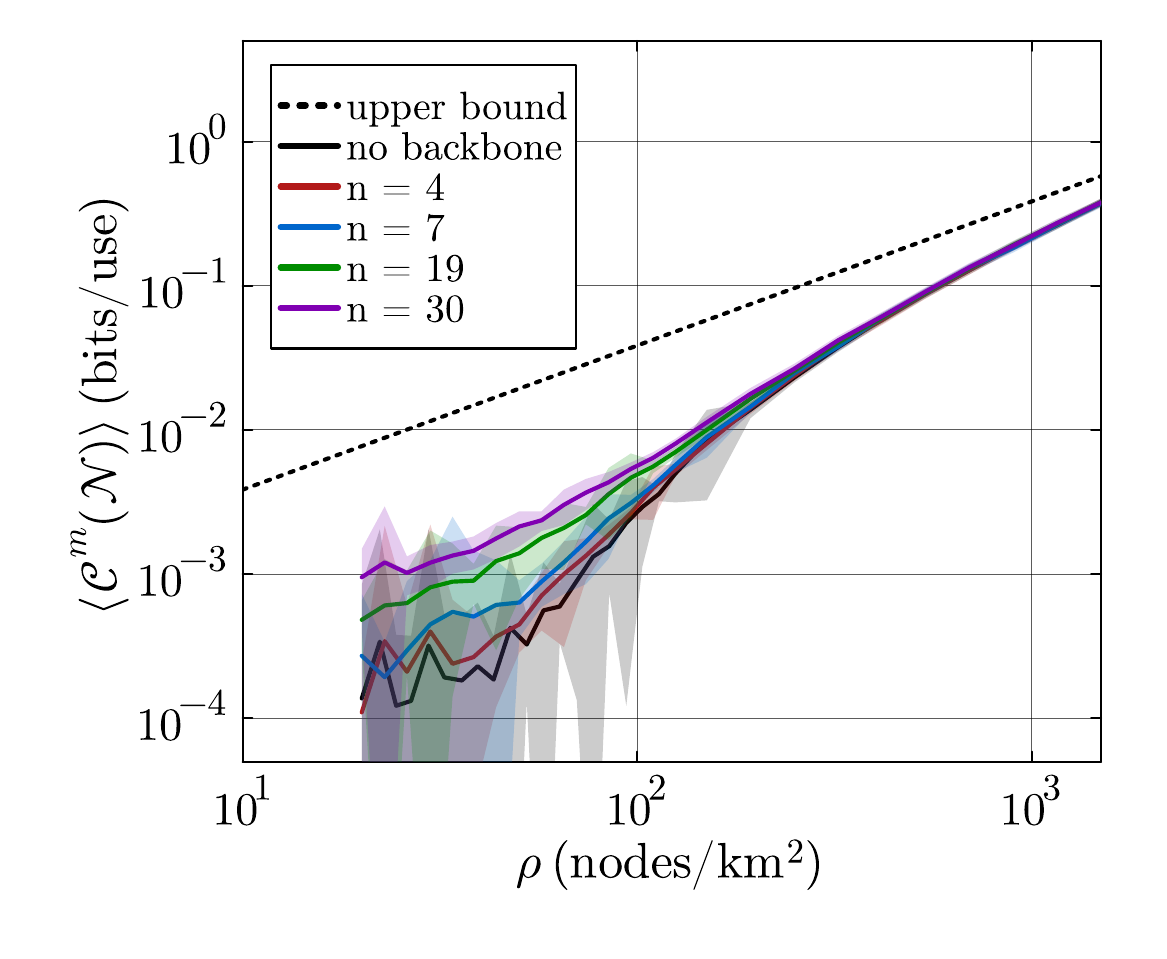}\label{fig:mobile_bb_flood_10}} &
    \subfloat[fixed, $R=50$km]{\includegraphics[width=0.42\linewidth,trim=0pt 27pt 0pt 3pt,clip]{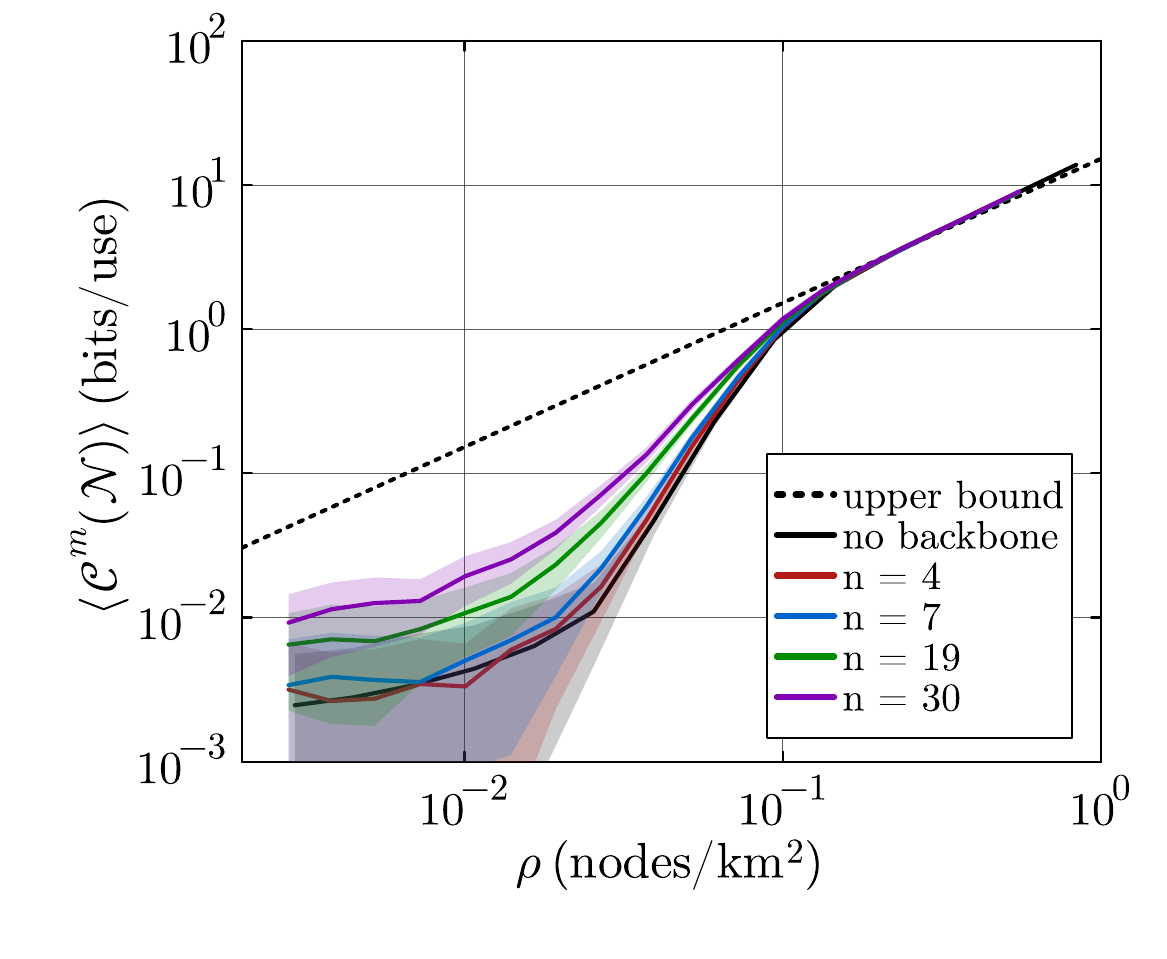}\label{fig:fixed_bb_flood_50}}
    \end{tabular}
    \caption{\textbf{Multi-path capacity bound with fiber backbone.} Left column: mobile users. Right column: fixed users. (a) Mobile, $R=0.2$km. (b) Fixed, $R=10$km. (c) Mobile, $R=0.5$km. (d) Fixed, $R=25$km. (e) Mobile, $R=1.0$km. (f) Fixed, $R=50$km. Multi-path (flooding) network capacity bound vs.\ nodal density, with and without an optimized fiber backbone. Shaded bands show $\pm1$ standard error of the mean, each data point averages 100 network instances and 50 end-user pairs per network, and each colour corresponds to a different number of backbone nodes. See text for more details.}
    \label{fig:flood_composite}
\end{figure*}

\section{Conclusion}
\label{sec:conclusions}
In this work, we have studied free-space quantum communication networks. Whilst we have focused on circular networks with uniform spatial distribution of nodes due to the tractability this provides, we note the method is readily extended to other regions whose radial distribution is known and to other radial nodal density distributions. More precisely, we have considered router-centered star networks and random Waxman networks whose link-rates are computed considering \SP{the two-way assisted capacity bound $\mathcal{C}=Q_2,D_2,K_2$ for a thermal-loss channel.}

We have shown the advantage that introducing an optimized fiber-based backbone can provide to free-space quantum networks for mobile and fixed end-users. The results demonstrate how single-path routing performance is improved by the introduction of additional backbone nodes both in terms of increased end-to-end network capacity and in terms of reduced edge consumption. In the fixed-user case this effect is reduced for the network capacity by the fact that the network capacities approach the maximal point-to-point capacity. Similar capacity improvements are found for multi-path routing.

A number of further directions are apparent. Clearly, Waxman structures are not the only random network model that may be relevant to the future topology of a large scale quantum network. In particular, scale free networks are of particular interest and have been explored in previous works on random quantum networks \cite{QuntaoRandQNets,ZhangQInt}. It would be natural to extend our exploration of optimized backbone structures to such networks. 

Additionally, we have made use of the same CVT backbones for fixed users in addition to mobile users. However, in the fixed case it may be possible to locate the backbone nodes at better positions via clustering algorithms. In particular, the CVT is optimal under an assumption of a constant density of nodes across the network region. This is of course appropriate in the mobile case, since the end-user devices move freely and the constant density approximates the positions of the nodes over a sufficiently long time period. However, in the fixed case, the nodes do not move, and the constant density is instead the underlying distribution from which the nodal positions were drawn. Thus for a fixed number of backbone nodes we may be able to find better locations. That being said, in the limit of a large number of nodes drawn from a uniform distribution, the optimal positions of the backbone nodes should converge to the mobile case. It would therefore be useful to determine density regions in which it is necessary to take such an approach over the one we have used. 

Throughout, backbone nodes have been treated as ideal and costed only by their number. A natural extension would be to model imperfect backbone nodes, for example with limited memory fidelity or multiplexing capacity, and assess the sensitivity of the Voronoi-optimality result to such imperfections. Relatedly, we have swept the backbone number $n_b$ as a free parameter throughout rather than optimising it directly; incorporating an explicit deployment cost model, as in Ref.~\cite{Alleaume2009TopologicalOptQKD}, would allow the backbone size itself to be chosen subject to a budget constraint. 

Calculating the end-to-end capacity for single pairs of end-users is useful in determining the ultimate upper limits for end-to-end communication. However, in reality, a given network will almost certainly contain more than one pair of communicating users at any given time. In this paper we have used the edge consumption as a metric to assess the scalability of the network to a higher number of concurrent users. Nonetheless it would be preferable to extend this directly to instances with multiple users. In particular this would likely elucidate the topological requirements of the backbone network which we have mostly not addressed here. 

Finally more complex routing strategies are available which strike a balance between single-path routing and flooding. In particular it is possible to utilise multi-path routing algorithms which do not necessarily utilise every edge in the network for a single communication instance. The results we present for multi-path routing should therefore be treated as ultimate upper bounds to the multi-path performances of free-space quantum networks (with or without backbones) with the acknowledgment that realistic deployments should use more scalable strategies, albeit at the cost of reduced end-to-end network capacity.

\noindent \textbf{Acknowledgments}. This work was supported by the Integrated Quantum Networks (IQN) Research Hub (EPSRC, Grant No. EP/Z533208/1).

\bibliography{FS}

@article{Du1999VoroniReview,
  author    = {Du, Qiang and Faber, Vance and Gunzburger, Max},
  title     = {Centroidal Voronoi Tessellations: Applications and Algorithms},
  journal   = {SIAM Review},
  volume    = {41},
  number    = {4},
  pages     = {637--676},
  year      = {1999},
  publisher = {Society for Industrial and Applied Mathematics},
  issn      = {1095-7200},
  doi       = {10.1137/S0036144599352836}
}

@article{Du2010VoronoiReview,
   author = {Qiang Du and Max Gunzburger and Lili Ju},
   doi = {10.4208/nmtma.2010.32s.1},
   issue = {2},
   journal = {Numer. Math. Theor. Meth. Appl},
   pages = {119-142},
   title = {REVIEW ARTICLE Advances in Studies and Applications of Centroidal Voronoi Tessellations},
   volume = {3},
   url = {http://www.global-sci.org/nmtma},
   year = {2010},
}

@ARTICLE{Gersho_Conj,
  author={Gersho, A.},
  journal={IEEE Transactions on Information Theory}, 
  title={Asymptotically optimal block quantization}, 
  year={1979},
  volume={25},
  number={4},
  pages={373-380},
  doi={10.1109/TIT.1979.1056067}}

@ARTICLE{Hexagon_theorem,
  author={Newman, D.},
  journal={IEEE Transactions on Information Theory}, 
  title={The hexagon theorem}, 
  year={1982},
  volume={28},
  number={2},
  pages={137-139},
  doi={10.1109/TIT.1982.1056492}}

@article{Wooters1982,
  title={A single quantum cannot be cloned},
  author={Wootters, William K and Zurek, Wojciech H},
  journal={Nature},
  volume={299},
  number={5886},
  pages={802--803},
  year={1982},
  publisher={Nature Publishing Group}
}

@article{Yook,
author = {Soon-Hyung Yook  and Hawoong Jeong  and Albert-László Barabási },
title = {Modeling the Internet's large-scale topology},
journal = {Proceedings of the National Academy of Sciences},
volume = {99},
number = {21},
pages = {13382-13386},
year = {2002},
URL = {https://www.pnas.org/doi/abs/10.1073/pnas.172501399},
}

@article{SF,
author = {Albert-László Barabási  and Réka Albert },
title = {Emergence of Scaling in Random Networks},
journal = {Science},
volume = {286},
number = {5439},
pages = {509-512},
year = {1999},
URL = {https://www.science.org/doi/abs/10.1126/science.286.5439.509}
}

@book{SiegmanLasers,
  title={Lasers},
  author={{A.} Siegman},
  publisher={University Science Books},
  address={Sausalito, California},
  year = {1986}
}

@book{AndrewsTurb,
  title={Laser Beam Propagation Through Random Medium},
  author={{L.~C.} Andrews and {R.~L.} Phillips},
  publisher={SPIE},
  address={Bellinghan},
  year = {2005}
}

@article{OPGQN,
  title = {Analytical Methods for High-Rate Global Quantum Networks},
  author = {Harney, Cillian and Pirandola, Stefano},
  journal = {PRX Quantum},
  volume = {3},
  issue = {1},
  pages = {010349},
  numpages = {12},
  year = {2022},
  month = {Mar},
  publisher = {American Physical Society},
  url = {https://link.aps.org/doi/10.1103/PRXQuantum.3.010349}
}

@article{IRQN,
	url = {https://doi.org/10.1088/2058-9565/ac7ba0},
	year = 2022,
	month = {jul},
	publisher = {{IOP} Publishing},
	volume = {7},
	number = {4},
	pages = {045009},
	author = {Cillian Harney and Stefano Pirandola},
	title = {End-to-end capacities of imperfect-repeater quantum networks},
	journal = {Quantum Science and Technology}
	}

@article{ZhangQInt,
	url = {https://doi.org/10.1088/2058-9565/ac1041},
	year = 2021,
	month = {jul},
	publisher = {{IOP} Publishing},
	volume = {6},
	number = {4},
	pages = {045007},
	author = {Bingzhi Zhang and Quntao Zhuang},
	title = {Quantum internet under random breakdowns and intentional attacks},
	journal = {Quantum Sci. Technol.}
}

@article{PirPatron09,
  title = {Direct and Reverse Secret-Key Capacities of a Quantum Channel},
  author = {Pirandola, Stefano and Garc\'{\i}a-Patr\'on, Raul and Braunstein, Samuel L. and Lloyd, Seth},
  journal = {Phys. Rev. Lett.},
  volume = {102},
  issue = {5},
  pages = {050503},
  numpages = {4},
  year = {2009},
  month = {Feb},
  publisher = {American Physical Society},
  url = {https://link.aps.org/doi/10.1103/PhysRevLett.102.050503}
}

@article{FS,
  title = {Limits and security of free-space quantum communications},
  author = {Pirandola, Stefano},
  journal = {Phys. Rev. Research},
  volume = {3},
  issue = {1},
  pages = {013279},
  numpages = {33},
  year = {2021},
  month = {Mar},
  publisher = {American Physical Society},
  url = {https://link.aps.org/doi/10.1103/PhysRevResearch.3.013279}
}

@article{UniteQInt,
  doi = {10.1038/532169a},
  url = {https://doi.org/10.1038/532169a},
  year = {2016},
  month = apr,
  publisher = {Springer Science and Business Media {LLC}},
  volume = {532},
  number = {7598},
  pages = {169--171},
  author = {Stefano Pirandola and Samuel L. Braunstein},
  title = {Physics: Unite to build a quantum Internet},
  journal = {Nature}
}

@book{Holevo19,
author = {Alexander S. Holevo},
title = {Quantum Systems, Channels, Information},
year = {2019},
publisher = {De Gruyter}
}

@article{QuntaoRandQNets,
  title = {Quantum communication capacity transition of complex quantum networks},
  author = {Zhuang, Quntao and Zhang, Bingzhi},
  journal = {Phys. Rev. A},
  volume = {104},
  issue = {2},
  pages = {022608},
  numpages = {9},
  year = {2021},
  month = {Aug},
  publisher = {American Physical Society},
  url = {https://link.aps.org/doi/10.1103/PhysRevA.104.022608}
}

@article{PLOB,
  author        = {Pirandola, Stefano and Laurenza, Riccardo and Ottaviani, Carlo and Banchi, Leonardo},
  title         = {Fundamental Limits of Repeaterless Quantum Communications},
  journal       = {Nature Communications},
  volume        = {8},
  pages         = {15043},
  year          = {2017},
  month         = {apr},
  publisher     = {Springer Nature},
  issn          = {2041-1723},
  doi           = {10.1038/ncomms15043}
}

@article{MasoudFS,
  url = {https://doi.org/10.1038/s42005-022-00814-5},
  year = {2022},
  month = feb,
  publisher = {Springer Science and Business Media {LLC}},
  volume = {5},
  number = {1},
  author = {Masoud Ghalaii and Stefano Pirandola},
  title = {Quantum communications in a moderate-to-strong turbulent space},
  journal = {Communications Physics}
}

@article{PracticalRouting,
  title         = {Practical Routing and Criticality in Large-Scale Quantum Communication Networks},
  author        = {Harney, Cillian and Pirandola, Stefano},
  journal       = {Physical Review Research},
  volume        = {7},
  number        = {4},
  pages         = {043168},
  numpages      = {17},
  year          = {2025},
  month         = {nov},
  publisher     = {American Physical Society},
  doi           = {10.1103/vy37-28jc}
}

@article{Harney2022HybridQNs,
  title = {End-To-End Capacities of Hybrid Quantum Networks},
  author = {Harney, Cillian and Fletcher, Alasdair I. and Pirandola, Stefano},
  journal = {Phys. Rev. Appl.},
  volume = {18},
  issue = {1},
  pages = {014012},
  numpages = {20},
  year = {2022},
  month = {Jul},
  publisher = {American Physical Society},
  doi = {10.1103/PhysRevApplied.18.014012},
  url = {https://link.aps.org/doi/10.1103/PhysRevApplied.18.014012}
}

@article{Xia2013StretchFactorDelaunay,
author = {Xia, Ge},
archivePrefix = {arXiv},
arxivId = {1103.4361},
doi = {10.1137/110832458},
issn = {00975397},
journal = {SIAM Journal on Computing},
number = {4},
pages = {1620--1659},
title = {{The stretch factor of the Delaunay triangulation is less than 1.998}},
volume = {42},
year = {2013}
}

@article{VanMeter2016PathToScalableQC,
author = {Van Meter, Rodney and Devitt, Simon J.},
doi = {10.1109/MC.2016.291},
issn = {00189162},
journal = {Computer},
number = {9},
pages = {31--42},
publisher = {IEEE},
title = {{The Path to Scalable Distributed Quantum Computing}},
volume = {49},
year = {2016}
}

@article{QINetworkChallenges2020,
author = {Cacciapuoti, Angela Sara and Caleffi, Marcello and Tafuri, Francesco and Cataliotti, Francesco Saverio and Gherardini, Stefano and Bianchi, Giuseppe},
doi = {10.1109/MNET.001.1900092},
issn = {1558156X},
journal = {IEEE Network},
number = {1},
pages = {137--143},
title = {{Quantum Internet: Networking Challenges in Distributed Quantum Computing}},
volume = {34},
year = {2020}
}

@article{Cuomo2020TowardsDistQC,
author = {Cuomo, Daniele and Caleffi, Marcello and Cacciapuoti, Angela Sara},
doi = {10.1049/iet-qtc.2020.0002},
issn = {2632-8925},
journal = {IET Quantum Communication},
number = {1},
pages = {3--8},
title = {{Towards a distributed quantum computing ecosystem}},
volume = {1},
year = {2020}
}

@article{QKD250km,
    title = {{High rate, long-distance quantum key distribution over 250 km of ultra low loss fibres}},
    year = {2009},
    journal = {New Journal of Physics},
    author = {Stucki, D. and Walenta, N. and Vannel, F. and Thew, R. T. and Gisin, N. and Zbinden, H. and Gray, S. and Towery, C. R. and Ten, S.},
    month = {7},
    volume = {11},
    pages = {075003},
    doi = {10.1088/1367-2630/11/7/075003},
    issn = {13672630}
}

@article{Joshi2020,
author = {Joshi, Siddarth Koduru and Aktas, Djeylan and Wengerowsky, S{\"o}ren and Lon{\v{c}}ari{\'c}, Martin and Neumann, Sebastian Philipp and Liu, Bo and Scheidl, Thomas and Curr{\'a}s-Lorenzo, Guillermo and Samec, {\v{Z}}eljko and Kling, Laurent and Qiu, Alex and Razavi, Mohsen and Stip{\v{c}}evi{\'c}, Mario and Rarity, John G. and Ursin, Rupert},
doi = {10.1126/sciadv.aba0959},
journal = {Science Advances},
number = {36},
pages = {eaba0959},
title = {{A trusted node-free eight-user metropolitan quantum communication network}},
url = {https://doi.org/10.1126/sciadv.aba0959},
volume = {6},
year = {2020}
}

@article{Chen2021_46nodeQN,
author = {Chen, Teng-Yun and Jiang, Xiao and Tang, Shi-Biao and Zhou, Lei and Yuan, Xiao and Zhou, Hongyi and Wang, Jian and Liu, Yang and Chen, Luo-Kan and Liu, Wei-Yue and Zhang, Hong-Fei and Cui, Ke and Liang, Hao and Li, Xiao-Gang and Mao, Yingqiu and Wang, Liu-Jun and Feng, Si-Bo and Chen, Qing and Zhang, Qiang and Li, Li and Liu, Nai-Le and Peng, Cheng-Zhi and Ma, Xiongfeng and Zhao, Yong and Pan, Jian-Wei},
doi = {10.1038/s41534-021-00474-3},
issn = {2056-6387},
journal = {npj Quantum Information},
number = {1},
pages = {134},
title = {{Implementation of a 46-node quantum metropolitan area network}},
url = {https://doi.org/10.1038/s41534-021-00474-3},
volume = {7},
year = {2021}
}

@article{CambridgeQN,
author = {Dynes, J F and Wonfor, A and Tam, W W -S. and Sharpe, A W and Takahashi, R and Lucamarini, M and Plews, A and Yuan, Z L and Dixon, A R and Cho, J and Tanizawa, Y and Elbers, J -P. and Grei{\ss}er, H and White, I H and Penty, R V and Shields, A J},
doi = {10.1038/s41534-019-0221-4},
issn = {2056-6387},
journal = {npj Quantum Information},
number = {1},
pages = {101},
title = {{Cambridge quantum network}},
url = {https://doi.org/10.1038/s41534-019-0221-4},
volume = {5},
year = {2019}
}

@inproceedings{Handley2019,
  author    = {Handley, Mark},
  title     = {Using Ground Relays for Low-Latency Wide-Area Routing in Megaconstellations},
  booktitle = {Proceedings of the 18th ACM Workshop on Hot Topics in Networks},
  series    = {HotNets '19},
  pages     = {125--132},
  year      = {2019},
  publisher = {Association for Computing Machinery},
  address   = {New York, NY, USA},
  isbn      = {9781450370202},
  doi       = {10.1145/3365609.3365859}
}

@article{Preet_2022,
	doi = {10.1149/10701.0033ecst},
	url = {https://doi.org/10.1149/10701.0033ecst},
	year = 2022,
	month = {apr},
	publisher = {The Electrochemical Society},
	volume = {107},
	number = {1},
	pages = {33--43},
	author = {Shiv Preet},
	title = {Satellite Internet Communication: A Race with Contemporary Optical Fiber Network with the Help of {SPT} Algorithm},
	journal = {{ECS} Transactions},
}

@article{Cirac1999DistQcomp,
author = {Cirac, J. I. and Ekert, A. K. and Huelga, S. F. and Macchiavello, C.},
doi = {10.1103/PhysRevA.59.4249},
issn = {10941622},
journal = {Physical Review A},
number = {6},
pages = {4249--4254},
primaryClass = {quant-ph},
title = {{Distributed quantum computation over noisy channels}},
volume = {59},
year = {1999}
}

@article{Pirandola2018ChannelSim,
  author        = {Pirandola, Stefano and Braunstein, Samuel L. and Laurenza, Riccardo and Ottaviani, Carlo and Cope, Thomas P. W. and Spedalieri, Gaetana and Banchi, Leonardo},
  title         = {Theory of channel simulation and bounds for private communication},
  journal       = {Quantum Science and Technology},
  volume        = {3},
  number        = {3},
  pages         = {035009},
  year          = {2018},
  publisher     = {IOP Publishing},
  issn          = {2058-9565},
  doi           = {10.1088/2058-9565/aac394}
}

@article{Kimble2008QInternet,
author = {Kimble, H. J.},
doi = {10.1038/nature07127},
issn = {14764687},
journal = {Nature},
number = {7198},
pages = {1023--1030},
pmid = {18563153},
title = {{The quantum internet}},
volume = {453},
year = {2008}
}

@article{Agiwal2016-5Greview,
author = {Agiwal, Mamta and Roy, Abhishek and Saxena, Navrati},
doi = {10.1109/COMST.2016.2532458},
issn = {1553877X},
journal = {IEEE Communications Surveys and Tutorials},
number = {3},
pages = {1617--1655},
publisher = {IEEE},
title = {{Next generation 5G wireless networks: A comprehensive survey}},
volume = {18},
year = {2016}
}

@article{Raychaudhuri2012FrontiersWireless,
author = {Raychaudhuri, Dipankar and Mandayam, Narayan B.},
doi = {10.1109/JPROC.2011.2182095},
isbn = {0001411101},
issn = {00189219},
journal = {Proceedings of the IEEE},
number = {4},
pages = {824--840},
title = {{Frontiers of wireless and mobile communications}},
volume = {100},
year = {2012}
}

@article{Vasylyev2019Sat,
  title = {Satellite-mediated quantum atmospheric links},
  author = {Vasylyev, D. and Vogel, W. and Moll, F.},
  journal = {Phys. Rev. A},
  volume = {99},
  issue = {5},
  pages = {053830},
  numpages = {27},
  year = {2019},
  month = {May},
  publisher = {American Physical Society},
  doi = {10.1103/PhysRevA.99.053830},
  url = {https://link.aps.org/doi/10.1103/PhysRevA.99.053830}
}

@article{PirandolaEnd-to-End19,
archivePrefix = {arXiv},
arxivId = {1905.12674},
author = {Pirandola, Stefano},
doi = {10.1038/s42005-019-0147-3},
issn = {2399-3650},
journal = {Communications Physics},
number = {1},
pages = {51},
title = {{End-to-end capacities of a quantum communication network}},
volume = {2},
year = {2019}
}

@article{ETEarxiv,
title={Capacities of repeater-assisted quantum communications
},
journal ={arxiv:1601.00966},
author = {Pirandola, Stefano},
year = {2016}
}

@book{taboga-stats,
author={Taboga, Marco},
title={Lectures on probability theory and mathematical statistics},
edition={3rd},
year={2017},
isbn={9781981369195},
publisher={Createspace Independent Publishing Platform},
url={https://www.statlect.com/}
}

@Inbook{Gazi2023,
author="Gazi, Orhan",
title="Functions of Random Variables",
bookTitle="Introduction to Probability and Random Variables",
year="2023",
publisher="Springer Nature Switzerland",
address="Cham",
pages="111--139",
isbn="978-3-031-31816-0",
doi="10.1007/978-3-031-31816-0_4",
url="https://doi.org/10.1007/978-3-031-31816-0_4"
}

@book{Penrose03RGG,
    author = {Penrose, Mathew},
    title = "{Random Geometric Graphs}",
    publisher = {Oxford University Press},
    year = {2003},
    month = {05},
    isbn = {9780198506263},
    doi = {10.1093/acprof:oso/9780198506263.001.0001},
    url = {https://doi.org/10.1093/acprof:oso/9780198506263.001.0001},
}

@inproceedings{musin2003propertiesDeltri,
  author        = {Musin, Oleg R.},
  title         = {Properties of the Delaunay Triangulation},
  booktitle     = {Proceedings of the Thirteenth Annual Symposium on Computational Geometry},
  pages         = {424--426},
  year          = {1997},
  publisher     = {ACM Press},
  doi           = {10.1145/262839.263061},
  archivePrefix = {arXiv},
  eprint        = {math/0312050},
  primaryClass  = {math.MG}
}

@article{Alleaume2009TopologicalOptQKD,
   author = {Alléaume, Romain and Roueff, François and Diamanti, Eleni and Lütkenhaus, Norbert},
   title = {Topological optimization of quantum key distribution networks},
   journal = {New Journal of Physics},
   volume = {11},
   pages = {075002},
   year = {2009},
   doi = {10.1088/1367-2630/11/7/075002}
}

@book{Simeone_2026, place={Cambridge}, title={Classical and Quantum Information Theory: Uncertainty, Information, and Correlation}, publisher={Cambridge University Press}, author={Simeone, Osvaldo}, year={2026}}

\end{document}